\documentclass[review]{elsarticle}

\usepackage{lineno,hyperref}
\modulolinenumbers[5]

\usepackage[acronym]{glossaries}

\newacronym{iop}{IOP}{Individual Optimization Problem}

\newacronym{bp-sop}{BP-SOP}{Budget Perturbed Social Optimization Problem}

\newacronym{sop-l}{SOP-L}{Social Optimization Problem with Linear Constraints}

\newacronym{iop'}{IOP'}{Individual Optimization Problem for Homogeneous Degree One Utilities}

\newacronym{bp-sop'}{BP-SOP'}{Budget Perturbed Social Optimization Problem for Homogeneous Degree One Utilities}

\newacronym{bp-sop-adm}{BP-SOP-ADM}{Budget Perturbed Social Optimization Problem for Alternating Direction Methods}

\newacronym{bp-sop-admm}{BP-SOP-ADMM}{Budget Perturbed Social Optimization Problem for Alternating Direction Method of Multipliers}

\usepackage{graphicx}

\usepackage{amsmath,inputenc}
\usepackage{csvsimple}
\usepackage{rotating}
\usepackage[letterpaper, portrait, margin=1in]{geometry}

\usepackage[nocomma]{optidef}
\usepackage{listings}
\usepackage{comment}
\usepackage{appendix}
\usepackage [english]{babel}
\usepackage [autostyle, english = american]{csquotes}
\MakeOuterQuote{"}
\usepackage{color} 
\usepackage{amsmath,amsthm}
\usepackage{thmtools, thm-restate}
\usepackage{hyperref}
\usepackage[ruled,vlined]{algorithm2e}

\DeclareMathOperator*{\argmax}{arg\,max}

\usepackage{tabularx}
    \newcolumntype{L}{>{\raggedright\arraybackslash}X}
\usepackage{array}
\definecolor{mygreen}{RGB}{28,172,0} 
\definecolor{mylilas}{RGB}{170,55,241}
\DeclareFixedFont{\ttb}{T1}{txtt}{bx}{n}{12} 
\DeclareFixedFont{\ttm}{T1}{txtt}{m}{n}{12}  
\usepackage{amssymb}
\usepackage[T1]{fontenc}
\usepackage{lipsum}

\makeatletter
\def\ps@pprintTitle{%
 \let\@oddhead\@empty
 \let\@evenhead\@empty
 \def\@oddfoot{}%
 \let\@evenfoot\@oddfoot}

\newtheorem{theorem}{Theorem}

\newtheorem{definition}{Definition}
\newtheorem{corollary}{Corollary}
\newtheorem{proposition}{Proposition}
\newtheorem{lemma}{Lemma}
\newtheorem{remark}{Remark}

\newenvironment{hproof}{%
  \proof}{\endproof}

\usepackage{color}
\definecolor{deepblue}{rgb}{0,0,0.5}
\definecolor{deepred}{rgb}{0.6,0,0}
\definecolor{deepgreen}{rgb}{0,0.5,0}

\usepackage[table,xcdraw]{xcolor}

\usepackage{listings}

\begin{document}

\begin{frontmatter}

\title{Fisher Markets with Linear Constraints: Equilibrium Properties and Efficient Distributed Algorithms}

\author[$^1$]{Devansh Jalota\corref{mycorrespondingauthor}} 
\cortext[mycorrespondingauthor]{Corresponding author}
\ead{djalota@stanford.edu}
\author{Marco Pavone, Yinyu Ye}
\address{Stanford University, Stanford CA 94305, USA}
\author{Qi Qi}
\address{Hong King University of Science and Technology, Hong Kong}

\begin{abstract}
The Fisher market is one of the most fundamental models for resource allocation problems in economic theory, wherein agents spend a budget of currency to buy goods that maximize their utilities, while producers sell capacity constrained goods in exchange for currency. However, the consideration of only two types of constraints, i.e., budgets of individual buyers and capacities of goods, makes Fisher markets less amenable for resource allocation settings when agents have additional linear constraints, e.g., knapsack constraints of buyers in e-commerce markets and proportionality constraints of schools in school choice matching markets. In this work, we introduce a modified Fisher market, where each agent may have additional linear constraints and show that this modification to classical Fisher markets fundamentally alters the properties of the market equilibrium as well as the optimal allocations. These properties of the modified Fisher market prompt us to introduce a budget perturbed social optimization problem (\acrshort{bp-sop}) and set prices based on the dual variables of \acrshort{bp-sop}'s capacity constraints. To compute the budget perturbations, we develop a fixed point iterative scheme and validate its convergence through numerical experiments.

Since this fixed point iterative scheme involves solving a centralized problem at each step, we propose a new class of distributed algorithms to compute equilibrium prices. In particular, we develop an Alternating Direction Method of Multipliers (ADMM) algorithm with strong convergence guarantees for Fisher markets with homogeneous linear constraints as well as for classical Fisher markets. In this algorithm, the prices are updated based on the tatonnement process, with a step size that is completely independent of the utilities of individual agents. Thus, our mechanism, both theoretically and computationally, overcomes a fundamental limitation of classical Fisher markets, which only consider capacity and budget constraints.
\end{abstract}


\end{frontmatter}

\section{Introduction}

The study of market equilibria is central to economic theory and traces back to 1874 with the seminal work of Walras \cite{walras1954elements}. In his work, Walras investigated the problem of setting prices and determining an allocation of goods to agents such that the \textit{market clears}, i.e., all goods are sold and each agent receives its most favoured bundle of goods that is affordable. The existence of such a market equilibrium was established in the work of Arrow and Debreu \cite{arrow-debreu} under mild conditions on the utility functions of buyers. A special case of Walras' model was independently proposed by Fisher in which consumers spend their budget of money (or artificial currency) to buy goods that maximize their utilities, while producers sell capacity constrained goods in exchange for currency \cite{Fisher-seminal}.

Since Fisher introduced his framework, numerous methods to compute market equilibria have emerged. One such approach was developed by Eisenberg and Gale who formulated Fisher's original problem with linear utilities as a convex optimization problem, and later extended this convex optimization framework to homogeneous degree one utilities \cite{EisGale,Gale}. In this convex program a social optimization problem is solved wherein the market clearing prices are computed through the dual variables of the problem's capacity constraints \cite{codenotti_varadarajan_2007}. Another common approach to compute market equilibria has been through the tatonnement process, wherein the prices of the goods are adjusted upwards when the demand for that good exceeds supply and are adjusted downwards when supply exceeds demand. An advantage of tatonnement is its distributed nature as compared to its convex programming counterpart. In addition to these approaches, other methods for market equilibrium computation have also emerged including primal-dual approaches \cite{devanur-primal-dual}, and auction based approaches \cite{vazirani_2007,nesterov-fisher-gale}.

A unifying theme across the Fisher market literature has been in leveraging the convex problem structure of Fisher's framework as well as the nature of the budget and capacity constraints to develop algorithms to compute market equilibria. Such algorithms for Fisher markets have been leveraged in applications including online advertising \cite{VaziApps} and revenue optimization \cite{FisherApps}; however, the consideration of only two types of constraints - budgets of buyers and capacities of goods - in Fisher markets limits its use in resource allocation settings when agents have additional linear constraints, e.g., knapsack and proportionality constraints. For instance, in the retail and e-commerce markets buyers have a knapsack constraint on the amount of goods they wish to purchase and in the school choice matching market, schools may wish to maintain a specified proportion of students within each racial group \cite{stable-matching-proportionality}. These applications highlight the need to consider additional linear constraints within Fisher's original model.

In this work, we consider the generalization of Fisher markets to the setting of additional linear constraints of the form $A^{(i)} \mathbf{x}_i \leq \mathbf{b}_i$ for each agent $i$. Here $A^{(i)} \in \mathbb{R}^{l_i \times m}$ is the constraint matrix, $\mathbf{b}_i \in \mathbb{R}^{l_i}_{\geq 0}$ is a non-negative constraint vector, and $\mathbf{x}_i \in \mathbb{R}^m$ is a vector of allocations. Further, $m$ is the number of goods in the market and $l_i$ is the number of additional linear constraints associated with agent $i$. We study the equilibrium properties of Fisher markets with additional linear constraints and derive equilibrium prices using a new convex program analogous to the one proposed by Eisenberg and Gale \cite{EisGale,Gale}.

To implement this market equilibrium, we use Alternating Direction Methods (ADMs) to iteratively update prices while agents distributedly solve their individual optimization problems given the prices. Similar to tatonnement mechanisms for classical Fisher markets, we ensure that (i) price updates are asynchronous, and (ii) the price update procedure converges quickly to the equilibrium prices. However, the step sizes of the price updates of traditional distributed tatonnement mechanisms depend on the specific utility functions of agents. In this work, we update prices using a step size that that is independent of the utilities of agents by developing an algorithm wherein each agent solves a perturbed individual optimization problem at each step.

This work introduces an extension to Fisher markets, which opens the doors for its application in many problem settings including the allocation of capacity constrained public goods to people, which was explored in a preliminary version of this work \cite{jalota-wine}. This problem has become increasingly important during the COVID-19 pandemic, since social distancing constraints have limited the number of people that can share public spaces \cite{NYPubGoods}. In the public goods context, we can consider pricing "time of use" permits for a public space, such as a beach. Here, the public goods are the different blocks of time that people can use the beach with a knapsack linear constraint that people can use the beach at most once per day.


\subsection{Our Contribution} \label{contributions}

This work presents a generalization of Fisher markets to the setting when each agent may have additional linear constraints and introduces algorithms to compute equilibria of this modified Fisher market. In this setting, we introduce each agent's individual optimization problem (\acrshort{iop}) with linear utilities and observe that the addition of these linear constraints to classical Fisher markets fundamentally alters the properties of the market equilibrium. In particular, (i) the equilibrium price vector may not exist and the equilibrium price set may be non-convex, (ii) the equilibrium prices of certain goods may be negative, (iii) individual agents no longer purchase goods corresponding to the maximum \textit{bang-per-buck} ratios of the different products in the market, (iv) certain goods in the market may be Giffen goods, i.e., their demand rises with an increase in the price of the goods, and (v) the classical Fisher social optimization problem with additional linear constraints fails to establish market clearing conditions.

Based on these properties of the modified Fisher market we derive sufficient conditions for the existence of a market equilibrium for non-homogeneous linear constraints and develop a convex programming method to test for the existence of market equilibria in the special case of homogeneous linear constraints. To compute the market equilibrium, we introduce a budget perturbed social optimization problem (\acrshort{bp-sop}) and set prices based on the dual variables of \acrshort{bp-sop}s capacity constraints. We further provide an economic interpretation of these budget perturbation constants and establish that the optimal allocation corresponding to the solution of \acrshort{bp-sop} is Pareto efficient. Since the perturbation constants in \acrshort{bp-sop} depend on its dual variables, which we do not know \emph{a priori}, we then present a fixed point iterative scheme to determine these constants by directly solving the centralized problem \acrshort{bp-sop} at each step. 

However, solving this centralized problem at each iteration can be quite computationally intensive and potentially impractical as a market designer may not have knowledge of agent's utilities. To overcome these concerns we present a new class of distributed tatonnement algorithms based on Alternating Direction Methods (ADMs) to determine equilibrium prices. In particular, we first introduce the Alternating Minimization Algorithm (AMA) in which agents distributedly solve their individual optimization problems at each iteration. For the AMA, we obtain a $O(\frac{1}{k})$ convergence guarantee to the market equilibrium for Fisher markets with homogeneous linear constraints, where $k$ is the number of rounds of the distributed algorithm. Since this algorithm requires setting the step sizes of the price updates based on the specific utilities of individual agents, we introduce the Alternating Direction Method of Multipliers (ADMM) algorithm. In the ADMM algorithm, we perturb the utilities of agents and obtain a $O(\frac{1}{k})$ convergence guarantee to the market equilibrium for Fisher markets with homogeneous linear constraints. We further provide a direct extension of the convergence results of the ADMM algorithm for classical Fisher markets. This result is of independent interest, since this algorithm converges for all concave homogeneous degree one utility functions including linear utilities. Finally, while the convergence guarantees of the ADMM algorithm do not extend to Fisher markets with non-homogeneous linear constraints, we present numerical convergence results and defer the theoretical treatment of convergence in this setting as an open question for future research.

\subsection{Related Work} \label{Literature}

The Fisher market framework is comprised of consumers who spend their budget of money to purchase goods that maximize their utilities and producers who sell capacity constrained goods in exchange for currency. Each consumer solves an individual optimization problem to obtain her most favoured bundle of goods given the set prices, while the producers set the prices in the market to optimize a social objective that aggregates the utilities of all consumers. We first describe each agent's individual optimization problem in Fisher markets. In this framework, the decision variable for agent $i$ is the quantity of each good $j$ they wish to purchase and is represented by $x_{ij}$. We denote the allocation vector for agent $i$ as $\mathbf{x}_i \in \mathbb{R}^m$ when there are $m$ goods in the market, with each good $j$ having a capacity $\Bar{s}_j$. A key assumption of Fisher markets is that goods are divisible and so fractional allocations are possible. 
Finally, denoting $w_i$ as the budget of agent $i$, $u_i(\mathbf{x}_i)$ as the utility of agent $i$ as a concave function of their allocation and $\mathbf{p} \in \mathbb{R}^m$ as the vector of prices for the goods, individual decision making can be modelled through the following optimization problem:

\begin{maxi!}|s|[2]                   
    {\mathbf{x}_i \in \mathbb{R}^m}                               
    {u_i(\mathbf{x}_i)  \label{eq:Fisher1}}   
    {\label{eq:Eg001}}             
    {}                                
    \addConstraint{\mathbf{p}^T \mathbf{x}_i}{\leq w_i \label{eq:Fishercon1}}    
    \addConstraint{\mathbf{x}_i}{\geq \mathbf{0}  \label{eq:Fishercon3}}  
\end{maxi!}
where~\eqref{eq:Fishercon1} is a budget constraint and~\eqref{eq:Fishercon3} are non-negativity constraints.

The vector of prices $\mathbf{p} \in \mathbb{R}^m$ that each agent observes are computed through the solution of a social optimization problem. The choice of the social objective for $n$ agents in the market is such that under certain conditions on the utility function, there exists an equilibrium price vector at which all goods are sold and all budgets are completely used.

\begin{definition} \label{def:market-eq-price} (Equilibrium Price Vector) 
A vector $\mathbf{p} \in \mathbb{R}^{m}$ is an equilibrium price vector if $\sum_{i = 1}^{n} x_{ij}^*(\mathbf{p}) = \Bar{s}_j$ for each good $j \in [m]$, i.e., all the goods are sold, where each resource $j$ has a price $p_j$ and has a strict capacity constraint of $\Bar{s}_j \geq 0$, and $\sum_{j = 1}^{m} p_j x_{ij}^*(\mathbf{p}) = w_i$, for all $i$, i.e., budgets of all agents are completely used. Furthermore, $x^*(\mathbf{p})_i \in \mathbb{R}^{m}$ is an optimal solution of the individual optimization Problem~\eqref{eq:Fisher1}-\eqref{eq:Fishercon3} for all agents $i$.
\end{definition}

When the utilities are homogeneous degree one functions, e.g., Constant Elasticity of Substitution (CES) utilities, the equilibrium price vector is computed through the dual variables of the capacity Constraints~\eqref{eq:FisherSocOpt1} of the following social optimization problem:
\begin{maxi!}|s|[2]                   
    {\mathbf{x}_i \in \mathbb{R}^m, \forall i \in [n]}                               
    {u(\mathbf{x}_1, ..., \mathbf{x}_n) = \sum_{i = 1}^{n} w_i \log( u_i(\mathbf{x}_i))  \label{eq:FisherSocOpt}}   
    {\label{eq:FisherExample1}}             
    {}                                
    \addConstraint{\sum_{i = 1}^{n} x_{ij}}{ = \Bar{s}_j, \forall j \in [m] \label{eq:FisherSocOpt1}}    
    \addConstraint{x_{ij}}{\geq 0, \forall i, j \label{eq:FisherSocOpt3}}  
\end{maxi!}
where there are $n$ agents, $m$ goods, and the objective function of the social planner is to maximize the budget weighted geometric mean of the buyer's utilities. We denote $[a] = \small\{1, 2, ..., a \small\}$.

At the equilibrium prices, the first order necessary and sufficient KKT conditions of the individual and social problems are equivalent. That is, under the prices set through the solution of the social optimization problem each agent receives their most favourable bundle of goods \cite{DevnaurPrimalDual}.

Fisher markets have been studied extensively in the computer science and algorithmic game theory communities with a recent interest in considering additional constraints to Fisher's original framework. For instance, Bei et al. \cite{bei-fisher} impose limits on sellers' earnings and question the assumption that utilities of buyers strictly increase in the amount of good allocated. A different generalization is considered by Vazirani \cite{VaziApps}, Devanur \cite{devnaur-spending} and Birnbaum et al. \cite{birnbaum-spending}, wherein utilities of buyers depend on prices of goods through spending constraints. Yet another generalization has been considered by Devanur et al. \cite{Devnaur-generalizations} in which goods can be left unsold as sellers declare an upper bound on their earnings and budgets can be left unused as buyers declare an upper bound on their utilities. Along similar lines, Chen et al. \cite{FisherPricing} study equilibrium properties when agents keep unused budget for future use. These generalizations are primarily associated with spending constraints of buyers and earning constraints of sellers; however, to the best of our knowledge there has been no generalization of Fisher markets to additional linear constraints, which are prevalent in a range of applications.

Even though such additional linear constraints have not been studied in the Fisher market literature, there have been other assignment mechanisms and market equilibrium characterizations that take into account constraints beyond budgets of buyers and capacities of goods. To accommodate distributional constraints in stable assignment problems, Fragiadakis and Troyan \cite{hard-distr-constraints} designed dynamic quota mechanisms with fairness and incentive guarantees and Ashlagi et al. \cite{ashlagi-assignment-distributional} generalized the serial dictatorship and probabilistic serial mechanisms. 
Another notable work is that on the Combinatorial Assignment problem by Budish \cite{Budish} wherein a market based mechanism is used to assign students to courses while respecting student's schedule constraints. In Budish's framework, courses have strict capacity constraints and students are endowed with budgets and must submit their preferences to a centralized mechanism that provides approximately efficient allocations. An even more general class of constraints is considered by Akbarpour and Nikzad \cite{akbarpour-restud}, who design a randomized mechanism to implement allocations when some of these constraints are ``soft'', i.e., they are treated as goals rather than as hard constraints. We study the problem of allocating goods under additional linear constraints from a different perspective by setting equilibrium prices through the maximization of a societal objective whilst also maximizing individual utilities given those prices.


A large body of the Fisher market literature has been focused on the computation of market equilibria, in particular, through tatonnement based mechanisms. In such algorithms, prices are updated using gradient descent based on the discrepancy between supply and demand until convergence to the equilibrium prices \cite{Codenotti2005MarketEV}. A major focus of the market literature has been on finding the right step sizes to update the prices to guarantee fast convergence to the equilibrium. For instance, \cite{CHEUNG2019} establishes the convergence of a discrete version of tatonnement for complementary CES utilities and \cite{saberi-price-update} computes an approximate equilibrium using tatonnement in a number of iterations that is linear in the number of goods in the market. However, the step size required to ensure fast convergence for Leontief utilities in \cite{CHEUNG2019} and for homogeneous degree $\alpha \in [0, 1]$ utilities in \cite{saberi-price-update} depends on the knowledge of agent's utilities. This inherent drawback of deriving a tatonnement algorithm with a step size that does not depend on the knowledge of agent's utilities was overcome by Cole and Fleischer \cite{Cole2008FastconvergingTA}. However, the convergence rate of their algorithm depends on the type of utility function of agents and their convergence results do not extend to agents with linear utility functions. In this work, we develop algorithms that avoid the need to choose an appropriate step size and provide a convergence guarantee for all homogeneous degree one utility functions, including linear utilities.

Distributed algorithms have become an increasingly important feature of modern convex optimization and have enabled large computations in a data rich environment. An important class of parallel optimization algorithms are the Alternating Direction Methods (ADMs), wherein small local sub-problems are solved at each step and these solutions are coordinated to find the solution for a large global problem. Two common splitting methods are the Alternating Direction Method of Multipliers (ADMM) and the Alternating Minimization Algorithm (AMA), which combine the benefits of dual decomposition methods that are amenable to distributed computation, and augmented Lagrangian approaches, which enable fast convergence guarantees. ADMM was first introduced by Glowinski and Marrocco in \cite{admm-glowinski-marroco} and some of its early convergence guarantees were established in \cite{GABAY197617,augmented-lag-glowinski-book,douglas-rachford}. Since ADMM has found many applications and has had resounding numerical success, there has been renewed interest in obtaining convergence guarantees for this procedure for convex optimization problems \cite{admm-main-conv}. Variants of the ADMM algorithm have also been proposed to improve convergence rates as in Deng and Yin \cite{fast-admm-strict-convex} and Goldstein et al. \cite{ama-convergence}. The AMA algorithm was introduced by Tseng \cite{tseng-paper1}, and its accelerated variant was studied in Goldstein et al. \cite{ama-convergence}. In this work, we leverage the AMA and ADMM algorithms to obtain fast convergence guarantees for Fisher markets with homogeneous linear constraints.

\textbf{Organization}: The rest of this paper is organized as follows. In Section~\ref{IOP-properties}, we present the individual optimization problem \acrshort{iop} with additional linear constraints and study properties of the corresponding market equilibrium. Then, in Section~\ref{Impossibility}, we propose a social convex optimization problem with budget perturbations to derive the market equilibrium with additional linear constraints. Next, we present a fixed point iterative scheme to compute these budget perturbations in Section~\ref{FixedPoint}. Finally, to compute the equilibrium prices we propose a new class of algorithms based on Alternating Direction Methods in Section~\ref{ADMM-Convergence-Guarantees} and conclude the paper in Section~\ref{Conclusion}.

\section{Properties of the Individual Optimization Problem} \label{IOP-properties}

In this section, we study the individual optimization problem of agents with additional linear constraints that are not considered in classical Fisher markets. We start by defining a new individual optimization problem (\acrshort{iop}) in Section~\ref{Model} and study properties regarding the existence, negativity and non-uniqueness of a market equilibrium in Sections~\ref{counterexample-physical}-\ref{counterexample-uniqueness}. In particular, we derive a sufficient condition to guarantee the existence of an equilibrium and show that the equilibrium price set is in general non-convex. Finally, we provide a characterization of the optimal solution of \acrshort{iop} in Section~\ref{IOP-Optimal} and perform a comparative statics analysis to study the behavior of individual agents in response to price changes in Section~\ref{giffen-behavior}. These results suggest that Fisher markets with additional linear constraints are fundamentally different from classical Fisher markets.

\subsection{Modelling Framework for Individual Optimization Problem} \label{Model}

As in classical Fisher markets, we model agents as utility maximizers and in Sections~\ref{IOP-properties}-\ref{FixedPoint}, each agent's utility function is assumed to be linear in the allocations, which is a common utility function used in the Fisher market literature \cite{EisGale,LinUt}. We note that many of our results generalize to the case when agents have homogeneous degree one concave utility functions, and we make this generalization explicit for the relevant results. 

We model the preference of an agent $i$ for one unit of good $j$ through the utility $u_{ij}$, and extend classical Fisher markets through the consideration of each agent's additional linear constraints. To model these linear constraints, we let $T_i$ denote the set of all additional linear constraints associated with agent $i$ and let $t \in T_i$ denote one such constraint. These linear constraints can be specified as $A^{(i)} \mathbf{x}_i \leq \mathbf{b}_i$, where $A^{(i)} \in \mathbb{R}^{l_i \times m}$, with $l_i = |T_i|$, and $\mathbf{b}_i \in \mathbb{R}^{l_i}_{\geq 0}$. Furthermore we denote the constraint $t \in T_i$ as $A_t^{(i)} \mathbf{x}_i \leq b_{it}$, where $A_t^{(i)}$ is a row vector of the matrix $A^{(i)}$. Using our earlier notation for budgets and prices, we have the following individual optimization problem (\acrshort{iop})

\begin{maxi!}|s|[2]                   
    {\mathbf{x}_i \in \mathbb{R}^m}                               
    {u_i(\mathbf{x}_i) = \sum_{j=1}^{m} u_{ij} x_{ij}  \label{eq:eq1}}   
    {\label{eq:Example001}}             
    {}                                
    \addConstraint{\mathbf{p}^T \mathbf{x}_i}{\leq w_i \label{eq:con1}}    
    \addConstraint{A_t^{(i)} \mathbf{x}_i}{\leq b_{it}, \forall t \in T_i \label{eq:con2}}
    \addConstraint{\mathbf{x}_i}{\geq \mathbf{0}  \label{eq:con3}}  
\end{maxi!}
with a budget Constraint~\eqref{eq:con1}, additional linear Constraints~\eqref{eq:con2} and non-negativity Constraints~\eqref{eq:con3}. To provide specific examples of the linear Constraints~\eqref{eq:con2}, we note in the case of knapsack constraints that $A_t^{(i)} \mathbf{x}_i \leq b_{it}$ is identical to $\sum_{j \in t} x_{ij} \leq b_{it}$, where $j \in t$ represents that good $j$ belongs to the knapsack $t$. Another example would be a proportionality constraint, which involves constraints of the form $x_{ij} \leq c x_{ij'}$ for two goods $j, j'$ and some constant $c\in \mathbb{R}_{\geq 0}$.


\subsection{Market Equilibrium May Not Exist} \label{counterexample-physical}

In classical Fisher markets with linear utilities, there exists a unique market equilibrium under mild assumptions \cite{algo-game-theory}. However, in the presence of additional homogeneous or non-homogeneous linear constraints, an equilibrium price is not guaranteed to exist. The non-existence of a market equilibrium in the sense of Definition~\ref{def:market-eq-price} is elucidated through Proposition~\ref{prop:non-existence-eq}. Specifically, we show this result by exhibiting examples of two markets, one for homogeneous linear Constraints~\eqref{eq:con2}, i.e., when $b_{it} = 0$ for each $t \in T_i$ for each agent $i$, and another for non-homogeneous linear Constraints~\eqref{eq:con2}. In these examples, we establish that the additional linear constraints either prevent agents from completely using up their budgets or completely purchasing all the goods in the market, which precludes the existence of an equilibrium price vector.


\begin{restatable}{proposition}{nonexistenceeq} (Non-Existence of a Market Equilibrium)
\label{prop:non-existence-eq}
There exists a market wherein each good $j \in [m]$ has a potential buyer $i \in [n]$, i.e., $u_{ij}>0$, but no equilibrium price vector for the \acrshort{iop} exists when agents have either homogeneous and non-homogeneous linear Constraints~\eqref{eq:con2}.
\end{restatable}
We refer to~\ref{non-existence-prop1} for the proof of this proposition. Proposition~\ref{prop:non-existence-eq} indicates that in general we cannot expect a market equilibrium to exist for \acrshort{iop}, which suggests the need to impose assumptions on the linear Constraints~\eqref{eq:con2} to guarantee the existence of an equilibrium price vector. This is in contrast to the generality of equilibrium existence for classical Fisher markets.



\subsection{Equilibrium Price May Be Negative} \label{negative-price-eq}

Another striking feature of Fisher markets with additional linear constraints is that the equilibrium prices of certain goods may be negative. This is contrary to traditional economic assumptions wherein the prices of all goods in the market are non-negative at a market equilibrium. Proposition~\ref{prop:negative-price-eq} establishes the existence of a market wherein there exists a good with a strictly negative equilibrium price.

\begin{restatable}{proposition}{negativeeq} (Negative Equilibrium Price)
\label{prop:negative-price-eq}
There exists a market wherein each good $j \in [m]$ has a potential buyer $i \in [n]$, i.e., $u_{ij}>0$, and the equilibrium price vector $\mathbf{p} \notin \mathbb{R}_{\geq 0}^m$.
\end{restatable}

We refer to~\ref{apdx:negative-price-eq} for the proof of this proposition. This result highlights an equilibrium property of Fisher markets with additional linear constraints that significantly departs from that of classical linear Fisher markets, which have non-negative equilibrium prices.

\subsection{Condition to Guarantee Existence of Market Equilibrium} \label{guarantee-existence}


Given the results on the non-existence of the market equilibrium and the negativity of the equilibrium price vector, we show that under certain assumptions we can in fact guarantee the existence of a market equilibrium. To this end, we provide two conditions for equilibrium existence when agents have an \acrshort{iop} as given in Equations~\eqref{eq:eq1}-\eqref{eq:con3}. The first condition we require is that for each good $j$ there is an agent that can purchase any amount of $j$. This requirement ensures that the prices in the market are restricted to be strictly positive, as otherwise at negative prices agents could potentially buy an infinite amount of a good preventing the market from clearing. The second sufficient condition is that for each agent there is a good $j$ that is not associated with any linear Constraints~\eqref{eq:con2}. This condition arises as the additional constraints may preclude agents from spending their entire budget, and thus prevent the market from clearing. Thus, it must be ensured that there is a good not restrained by linear Constraints~\eqref{eq:con2} so that agents can purchase more units of it to spend their budget. Note that similar to the approach in \cite{FisherPricing}, by allowing agents to keep unused budget, we can treat budget as a good. As a result the technical assumption that there is a good $j$ that is not associated with any linear Constraints~\eqref{eq:con2} is not very demanding. We now formally state these conditions required for the existence of a market equilibrium in Theorem~\ref{thm:market-eq}.

\begin{restatable}{theorem}{marketeq}
\label{thm:market-eq} (Existence of Market Equilibrium)
Suppose that (i) for each good $j$ there is an agent $i$ with $u_{ij}>0$ that can purchase any amount of $j$, and (ii) for each agent $i$ there is a good $j$ that is not associated with any linear Constraints~\eqref{eq:con2}, where $u_{ij}>0$. Then there exists an equilibrium price vector.
\end{restatable}

\begin{hproof}
We normalize the capacities of each good and the total budget of all agents to $1$, and consider an excess demand function $f_j(\mathbf{p}) = \sum_{i = 1}^{n} x_{ij}(\mathbf{p}) - 1$. Due to condition (i) it suffices for us to consider $\mathbf{p} \in \Delta_m$, where $\Delta_m$ is a standard simplex. Next, we define a coloring function $c: \mathbf{p} \mapsto \small\{0, 1, ..., m\small\}$, such that $c(\mathbf{p}) = j$ for one such $j$ that satisfies $f_j(\mathbf{p}) \leq 0$ and $p_j \neq 0$. Such a coloring function on the standard simplex satisfies Sperner's lemma, which implies that we can find a $\mathbf{p}^*$, such that $f_j(\mathbf{p}^*) \leq 0$, for each $j$, showing for each $j$ that $\sum_{i = 1}^{n} x_{ij}(\mathbf{p}^*) \leq 1$.

To prove that the above inequality is an equality, we suppose that  there exists $j$, such that  $\sum_{i = 1}^{n} x_{ij}(\mathbf{p}^*) < 1$. Then we find a contradiction and prove the strict inequality is impossible under the condition that there exists a good $j$ without any linear Constraints~\eqref{eq:con2}. This establishes our claim that $\mathbf{p}^*$ is the equilibrium price vector.
\end{hproof}

We refer to~\ref{existence-market-eq} for a detailed proof of Theorem~\ref{thm:market-eq}. 
The result of Theorem~\ref{thm:market-eq} provides a methodology to test for and guarantee the existence of a market equilibrium for Fisher markets with additional linear constraints. Other sufficient conditions can be developed when conditions (i) and (ii) of Theorem~\ref{thm:market-eq} do not hold. For instance, in Section~\ref{homogeneous-equilibrium-computation}, we provide a convex program to test for the existence of an equilibrium price when the linear Constraints~\eqref{eq:con2} are homogeneous, i.e., $b_{it} = 0$ for each constraint $t \in T_i$ for each agent $i$.


\subsection{Equilibrium Price Set May Be Non-Convex} \label{counterexample-uniqueness}

We now show that even if a market equilibrium exists, it may be that the equilibrium price vector is not unique and in general the equilibrium price set may even be non-convex. Theorem~\ref{thm:nonconvexity} establishes this result through an example of a market where agents have knapsack linear constraints.

\begin{restatable}{theorem}{nonconvexity} (Non-Convexity of the Equilibrium Price Set)
\label{thm:nonconvexity}
There exists a market wherein the equilibrium price vector exists but is not unique and the corresponding equilibrium price set is non-convex for the \acrshort{iop} when agents have additional linear Constraints~\eqref{eq:con2}.
\end{restatable}


For a detailed proof of this result see~\ref{non-convex-appendix}, wherein we first establish two equilibrium price vectors for a market with additional knapsack constraints. Then, we show that a convex combination of these two price vectors is not a market equilibrium to establish the non-convexity of the equilibrium price set. This result establishes that the problem of determining a market equilibrium with linear Constraints~\eqref{eq:con2} is fundamentally different from classical Fisher markets. In particular, this result suggests that the problem of computing the market equilibria for Fisher markets with linear Constraints~\eqref{eq:con2} may be computationally hard due to the non-convexity in the equilibrium price set.

\subsection{Characterizing Optimal Solution of \acrshort{iop}} \label{IOP-Optimal}

In this section, we characterize the optimal solution of the \acrshort{iop} and show that this differs from the optimal solution of Problem~\eqref{eq:Fisher1}-\eqref{eq:Fishercon3} for classical linear Fisher markets. In classical linear Fisher markets, each agent purchases the goods $j^* = \argmax_j \left\{\frac{u_{ij}}{p_j} \right\}$. However, with linear Constraints~\eqref{eq:con2} agents no longer purchase goods corresponding to the highest \textit{bang-per-buck} ratio. In order to understand the purchasing behavior of agents, we study the \acrshort{iop} with knapsack linear constraints and characterize its optimal solution. To this end, we introduce the notion of a \textit{virtual product} and show that agents purchase goods corresponding to the maximum \textit{bang-per-buck} ratios of the \textit{virtual products}.





To characterize the optimal solution of the \acrshort{iop} with knapsack linear constraints, we begin with the consideration of a feasible solution set for buyer $i$ and knapsack constraint $t \in T_i$.

\begin{definition}(Feasible Set). \label{def-sol-set}
Given a price vector $\mathbf{p} \in \mathbb{R}_{\geq 0}^{m}$, a feasible solution set for buyer $i$ and knapsack constraint $t$ is given by:

$ S_{t} = \left\{(u_t, w_t) | \exists \left\{x_{ij} \right\}_{j \in t}, \sum_{j \in t} x_{ij} \leq 1, x_{ij} \geq 0, \forall j \in t, u_t = \sum_{j \in t} u_{ij} x_{ij}, w_t = \sum_{j \in t} x_{ij} p_j \right\} $
\end{definition}

Definition~\ref{def-sol-set} specifies agent $i$'s utility and budget when consuming goods in knapsack $t$. Note that in the above definition we have normalized $b_{it} = 1$ for each knapsack constraint. This is without loss of generality for the case when $b_{it}>0$. In the instance when $b_{it}=0$ then the feasible set $S_t$ is $\small\{(0,0)\small\}$, i.e., no good is purchased in knapsack $t$ by agent $i$. Thus, to determine the goods in the optimal bundle of agent $i$ it suffices to restrict attention to the case when $b_{it} = 1$ for each knapsack constraint, as described by the feasible set in Definition~\ref{def-sol-set}.

The solution set $S_t$ can be viewed as lying in the convex hull of the points defined by $(u_{ij}, p_{ij})$, $j \in t$ and the origin in the price-utility plane, as shown by the enclosed region in Figure~\ref{virtual-products}. The lower frontier of this convex hull, as shown in bold, from the origin to $(u_{ij_{\max}}, p_{j_{\max}})$, where $j_{\max} = \argmax_{j \in t} \left\{u_{ij}\right\}$, is piece-wise linear and is characterised by slopes $\theta^t = (\theta_1^t, \theta_2^t, ..., \theta_{k_t}^t)$, where $k_t = |\{j:j \in t\}|$. As shown on the right in Figure~\ref{virtual-products}, given a fixed budget $w_t$ for knapsack $t$, the maximal utility that can be obtained from knapsack $t$ must be the intersection of the line $p = w_t$ and the lower frontier of the convex hull when $w_t \leq p_{j_{\max}}$. Otherwise the maximal utility obtained is $u_{ij_{\max}}$. Therefore, an optimal solution of \acrshort{iop} must lie on the lower frontier, with endpoints of the line segments corresponding to goods and line segments corresponding to \textit{virtual products}. A \textit{virtual product} is the set of all convex combinations of utility-price pairs of two goods $j_1$ and $j_2$ belonging to knapsack $t$.

\begin{definition} (Virtual Product).
A \textit{virtual product} is the set of points $\alpha (u_{ij_1}, p_{j_1}) + (1-\alpha) (u_{ij_2}, p_{j_2})$ for all $\alpha \in [0, 1]$ for two goods $j_1$ and $j_2$ belonging to knapsack $t$. It is characterized by its two endpoints $A = (u_{ij_1}, p_{j_1})$, $B = (u_{ij_2}, p_{j_2})$, and a slope $\theta_{j_1j_2} = \frac{p_{j_2} - p_{j_1}}{u_{ij_2} - u_{ij_1}}$.
\end{definition}
A key property of interest of a \textit{virtual product} is its \textit{bang-per-buck}, which is the reciprocal of the slope of the \textit{virtual product}.

\begin{definition} \label{def:bang-per-buck-virtual} (Bang-Per-Buck of a Virtual Product)
The \textit{bang-per-buck} ratio of a \textit{virtual product} characterized by its two endpoints $A = (u_{ij_1}, p_{j_1})$, $B = (u_{ij_2}, p_{j_2})$, and a slope $\theta_{j_1j_2} = \frac{p_{j_2} - p_{j_1}}{u_{ij_2} - u_{ij_1}}$ is given by $\frac{1}{\theta_{j_1j_2}} = \frac{u_{ij_2} - u_{ij_1}}{p_{j_2} - p_{j_1}}$.
\end{definition}

\begin{figure}[!h]
      \centering
      \includegraphics[width=0.8\linewidth]{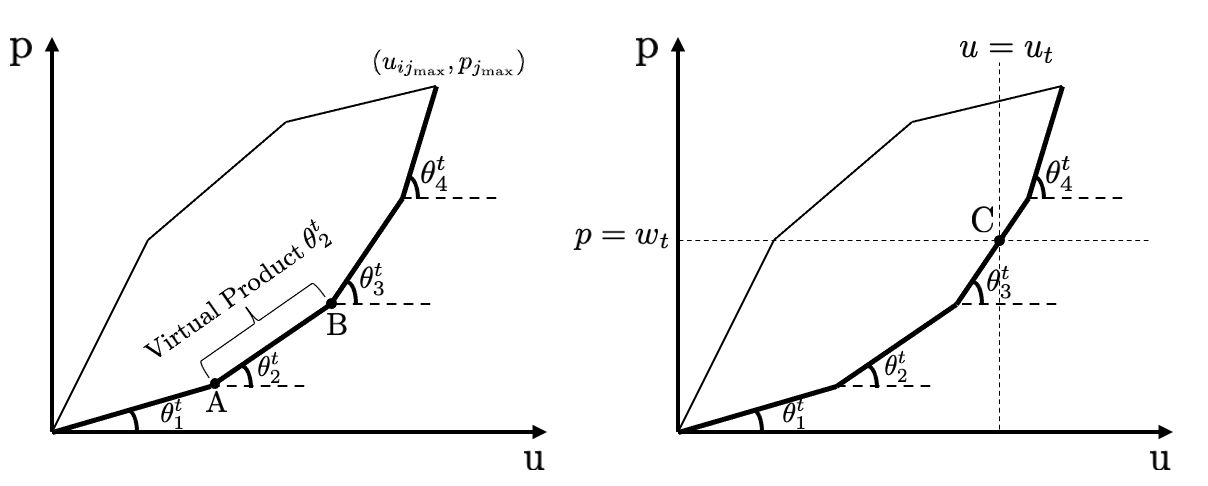}
      \caption{The enclosed region represents the convex hull corresponding to the solution set $S_t$. The vertices on the lower frontier (in bold) correspond to the goods and the segments correspond to \textit{virtual products}. The figure on the right shows that any optimal solution must lie on the lower frontier of the convex hull, as indicated by the point $C$.}
      \label{virtual-products}
   \end{figure}

While in Fisher markets, buyers purchase goods that maximize their \textit{bang-per-buck} ratios, in the presence of knapsack constraints, we show that agents purchase goods in the descending order of the \textit{virtual products}' \textit{bang-per-buck} ratios. Theorem~\ref{thm:iopoptimal} formalizes this result. 

\begin{restatable}{theorem}{iopopt} (Optimal Solution of \acrshort{iop})
\label{thm:iopoptimal}
Suppose that agents have knapsack linear Constraints~\eqref{eq:con2} where each good belongs to at most one linear Constraint~\eqref{eq:con2}. Then, given a price vector $\mathbf{p} \in \mathbb{R}_{\geq 0}^{m}$, agent $i$ can obtain the optimal solution $\mathbf{x}_i^* \in \mathbb{R}^{m}$ of the \acrshort{iop} by mixing all \textit{virtual products} from different knapsacks together and spending their budget $w_i$ in the descending order of the \textit{virtual products}' bang-per-buck ratios. Furthermore, at most one unit of each virtual product can be purchased by agent $i$.
\end{restatable}

Note that when two \textit{virtual products} have the same slope, irrespective of whether they correspond to the same knapsack constraint or different knapsack constraints, ties can be broken arbitrarily. We can derive another property of \acrshort{iop} through an immediate corollary of Theorem~\ref{thm:iopoptimal}. In particular, Corollary~\ref{cor:coroptsol} establishes the number of different goods an agent will purchase corresponding to each knapsack constraint.

\begin{restatable}{corollary}{coroptsol} (Quantity of Goods Purchased with Knapsack Linear Constraints)
\label{cor:coroptsol}
For any agent $i$, there exists an optimal solution $\mathbf{x}_i^* \in \mathbb{R}^{m}$ of the \acrshort{iop}, such that $i$ purchases two different goods in at most one knapsack linear constraint. For all other knapsacks, agent $i$ buys at most one good.
\end{restatable}

We refer to~\ref{iop-optimal-pf} and~\ref{cor-optimal-pf} for a detailed proof of Theorem~\ref{thm:iopoptimal} and Corollary~\ref{cor:coroptsol} respectively. The characterization of the optimal solution of \acrshort{iop} given by Theorem~\ref{thm:iopoptimal} and Corollary~\ref{cor:coroptsol} holds for any price vector $\mathbf{p}$, irrespective of whether it is an equilibrium price or not. The optimal solution may not be unique and Corollary~\ref{cor:coroptsol} states there exists an optimal solution such that agent $i$ purchases two different goods in at most one knapsack and in all other knapsacks agents purchase at most one good. However, Corollary~\ref{cor:coroptsol} does not imply that all solutions must satisfy these conditions. Furthermore, note that these results can easily be extended to the case when there exists some good that does not have any knapsack constraint. The detailed discussion and interpretation by examples of the optimal solution of the \acrshort{iop} are presented in~\ref{rmk-iopopt} and~\ref{examples-iopopt}.

\subsection{Individual Behavior in Response to Price Changes} \label{giffen-behavior}

In this section, we use our characterization of the optimal solution of the \acrshort{iop} to perform a comparative statics analysis. Since comparative statics is concerned with determining the sign of the changes in the endogenous variables as a result of changes in an exogenous parameter, we study the change in consumer demand when solving the \acrshort{iop} under changes to the price vector $\mathbf{p}$ and income level $w_i$. In particular, we show that consumers with linear Constraints~\eqref{eq:con2} exhibit Giffen behavior, where the demand for a good may rise with an increase in its price, and that this Giffen behavior disappears under an appropriate level of income compensation. 

The first of these results presents a significant departure from consumer behavior in classical linear Fisher markets. This is because when agents only have budget constraints and a fixed level of income $w_i$, then as the price of good $j$ is increased from $p_j$ to $p_j'$ (and the prices of all other goods are kept fixed) the demand for good $j$ can never increase. To see this, we note that the initial demand for the good is either $\frac{w_i}{p_j}$ if good $j$ has the maximum \textit{bang-per-buck} ratio or it is 0. At the new price $p_j' > p_j$, the \textit{bang-per-buck} ratio for this product reduces and so the new demand will either be $\frac{w_i}{p_j'} < \frac{w_i}{p_j}$ or 0, indicating that the demand for the good $j$ can never rise with an increase in its price.

However, we now show in Proposition~\ref{prop:propGiffen} that with linear Constraints~\eqref{eq:con2} if the price of good $j$ is increased while the prices of all other goods is kept fixed then agents may purchase a larger quantity of good $j$, indicating the existence of Giffen goods. Specifically, we exhibit an example of a market with knapsack linear constraints such that an increase in the price of a good corresponds to an increase in its demand.

\begin{restatable}{proposition}{propGiffen} (Existence of Giffen Goods)
\label{prop:propGiffen}
Suppose that under the price vector $\mathbf{p}$ and income level $w_i$, the optimal consumption bundle corresponding to \acrshort{iop} is $\mathbf{x}_i$. If the price of good $j$ is increased from $p_j$ to $p_j'$, and the prices of all other goods are fixed then at the new price vector $\mathbf{p}'$ there is a market such that the consumption of good $j$ is higher, i.e., $x_{ij}' > x_{ij}$. Here $\mathbf{x}_i'$ is the optimal consumption bundle for the \acrshort{iop} at the price $\mathbf{p}'$.
\end{restatable}

\begin{proof}
We provide an example of a market and show that if we increase the price of one good then its consumption will increase when an agent solves the \acrshort{iop} with linear Constraints~\eqref{eq:con2}. Consider a two good market, where we set $p_1 = 0.5$, $p_2 = 3$ and we have the knapsack constraint that $x_{i1}+x_{i2} \leq 1$. Further suppose $u_{i1} = 1$, $u_{i2} = 2$ and $w_i = 1$. In this case, we have that the chosen bundle for the agent is $x_1 = (0.8, 0.2)$ using the characterization of the optimal solution in Theorem~\ref{thm:iopoptimal}, as indicated by the red dot in the left of Figure~\ref{fig:giffen-goods}. Now, suppose that we increase the price of good 1 to $p_1' = 1$ and keep the price of good 2 fixed, i.e., $p_2 = 3$. In this case, we will have that the agent purchases the bundle $(1, 0)$, as in the right of Figure~\ref{fig:giffen-goods}, which demonstrates that for an increase in the price of a good the quantity demanded of the good in fact increases. Thus, good 1 is a Giffen good in this market, proving our claim.
\begin{figure}[!h]
      \centering
      \includegraphics[width=0.8\linewidth]{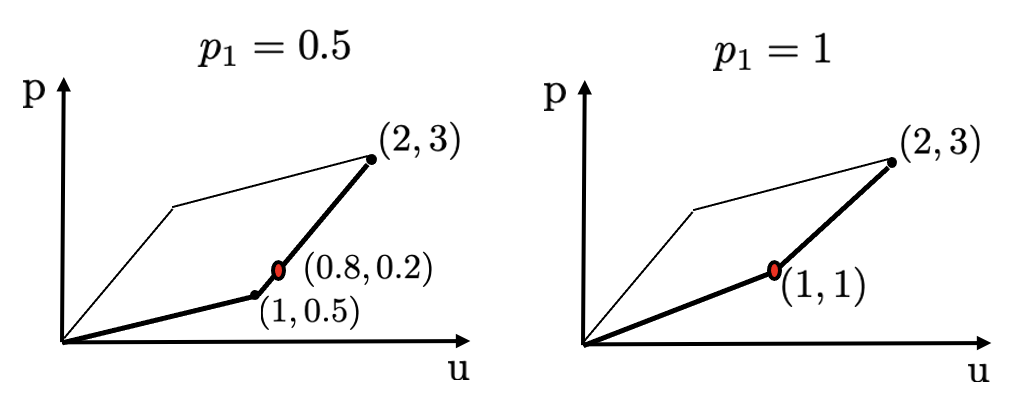}
      \caption{Optimal consumption bundle in a market with prices $p_1 = 0.5$ and $p_1 = 1$ to show that good 1 is a Giffen good. When the price of good 1 is increased from $0.5$ to $1$ then the optimal quantity purchased of good 1 increases from $0.8$ units to $1$ unit indicated by the red dots.}
      \label{fig:giffen-goods}
   \end{figure}
\end{proof}

The counter-intuitive purchasing behavior of agents, as elucidated in Proposition~\ref{prop:propGiffen}, further establishes that the setting of Fisher markets with linear Constraints~\eqref{eq:con2} fundamentally differs from classical Fisher markets. In particular, the standard inverse relationship between the price and quantity of a good demanded no longer applies for Fisher markets with linear constraints.


\begin{remark}
In economic theory, the existence of Giffen goods has been largely attributed to the ``income effect''. This is because when the price of good $j$ rises then the ``level of real wealth'' of a consumer falls, which is in particular true for poorer consumers \cite{giffen-behavior}. Thus, if good $j$ is an \textit{inferior} good then its demand by poorer consumers will increase with this relative ``fall'' in the income level of the consumer \cite{kreps-book}. We note that in the case of the \acrshort{iop} the ``income-effect'', which causes Giffen behavior, is induced because of the presence of additional linear constraints.
\end{remark}

We now turn to the second result of this section, wherein we isolate the effect of a change in the relative prices from the ``income effect'' by compensating the agent for an increase in the price of good $j$. This is so that the consumer's level of ``real income'' is the same as that before the price change. That is, the original bundle of goods purchased by the agent prior to the price change is affordable at the new prices. While Giffen behavior was observed at a fixed original income level, we now show that with income compensation Giffen behavior ceases to persist for agents with knapsack linear Constraints~\eqref{eq:con2} where each good belongs to at most one knapsack. A formal statement of this result is provided in Proposition~\ref{prop:propCompensation}.

\begin{restatable}{proposition}{propCompensation} (Disappearance of Giffen Behavior under Income Compensation)
\label{prop:propCompensation}
Suppose that agents have knapsack linear Constraints~\eqref{eq:con2} where each good belongs to at most one linear Constraint~\eqref{eq:con2}. Let the optimal consumption bundle at the price vector $\mathbf{p} \in \mathbb{R}_{\geq 0}^m$ and income level $w_i$ corresponding to \acrshort{iop} be $\mathbf{x}_i$. Further suppose that at the price $\mathbf{p}' \in \mathbb{R}_{\geq 0}^m$ and income $w_i' = \mathbf{p}' \cdot \mathbf{x}_i$, where $\mathbf{p}'$ is $\mathbf{p}$ except $p_j' > p_j$, the optimal consumption bundle is $\mathbf{x}_i'$. Then, $x_{ij}' \leq x_{ij}$.
\end{restatable}

We refer to~\ref{prop-Compensation} for a detailed proof of Proposition~\ref{prop:propCompensation}. Intuitively, this proposition states that when the change in the ``level of real wealth'' arising from an increase in the price of good $j$ is compensated for then the agent will no longer consume more of good $j$ at the new price. This result aligns with classical economic theory, wherein income compensation for agents has traditionally accompanied the disappearance of Giffen behavior \cite{kreps-book}.

Having established properties of the \acrshort{iop}, we now turn to the problem of deriving equilibrium prices in the market with additional linear constraints.


\section{Social Optimization Problem with Additional Linear Constraints} \label{Impossibility}

A desirable property of Fisher markets is that the market equilibrium outcome maximizes a social objective while individuals receive their most favoured bundle of goods given the set prices. As this property of Fisher markets holds when we consider budget and capacity constraints, a natural question to ask is whether we can still achieve this property under the addition of linear Constraints~\eqref{eq:con2}. We start by showing that with the addition of these constraints and under no further modifications to Fisher markets, market clearing conditions fail to hold. To do this we define the social optimization problem (\acrshort{sop-l}) with additional linear constraints in Section~\ref{DefSocOpt} and compare its KKT conditions to that of \acrshort{iop} in Section~\ref{KKTImposs}.

We then address this negative result by defining a budget perturbed social optimization problem (\acrshort{bp-sop}) in Section~\ref{ReformSocOpt} in which we adjust the budgets of the agents. Then, in Section~\ref{SocOptKKT}, we show how to choose these budget perturbations to guarantee the equivalence of its KKT conditions with that of the \acrshort{iop} when prices are set through the dual variables of \acrshort{bp-sop}s capacity constraints. We then investigate the equilibrium computation properties in the special case of homogeneous linear constraints in Section~\ref{homogeneous-equilibrium-computation} and establish that the optimal allocation vectors corresponding to the solution of the budget perturbed social optimization problem guarantees a Pareto efficient outcome in Section~\ref{pareto-efficiency}. Finally, we provide an economic interpretation of the budget perturbed formulation in Section~\ref{InterpSocOpt}.


\subsection{A Social Optimization Problem with Additional Constraints} \label{DefSocOpt}

We first define the natural extension of the Fisher market social optimization Problem~\eqref{eq:FisherSocOpt}-\eqref{eq:FisherSocOpt3} with the addition of linear Constraints~\eqref{eq:ImpossSocOpt2}, giving the problem \acrshort{sop-l}:

\begin{maxi!}|s|[2]                   
    {\mathbf{x}_i \in \mathbb{R}^m, \forall i \in [n]}                               
    {u(\mathbf{x}_1, ..., \mathbf{x}_n) = \sum_{i=1}^{n} w_i \log \left( \sum_{j=1}^{m} u_{ij}x_{ij} \right)  \label{eq:ImpossSocOpt}}   
    {\label{eq:ImpossExample1}}             
    {}                                
    \addConstraint{\sum_{i=1}^{n} x_{ij}}{ = \Bar{s}_j, \forall j \in [m] \label{eq:ImpossSocOpt1}}    
    \addConstraint{A_t^{(i)} \mathbf{x}_i}{\leq b_{it}, \forall t \in T_i, \forall i \in [n] \label{eq:ImpossSocOpt2}}
    \addConstraint{x_{ij}}{\geq 0, \forall i, j \label{eq:ImpossSocOpt3}}  
\end{maxi!}
with capacity Constraints~\eqref{eq:ImpossSocOpt1}, additional linear Constraints~\eqref{eq:ImpossSocOpt2} and non-negativity Constraints~\eqref{eq:ImpossSocOpt3}.

\subsection{A KKT Comparison of \acrshort{iop} and \acrshort{sop-l}} \label{KKTImposs}

In classical Fisher Markets, the equilibrium price corresponds to the dual variables of the capacity constraints of the social optimization problem, and at this equilibrium, the KKT conditions of the individual and social optimization problems are equivalent \cite{DevnaurPrimalDual}. In contrast to this result, we show in Theorem~\ref{thm:imposs} that for Fisher markets with additional linear constraints market clearing conditions may fail
to exist. To prove this claim, we show that the KKT conditions of the individual and social optimization problems with
additional linear constraints (\acrshort{iop} and \acrshort{sop-l} respectively) are not equivalent.

\begin{restatable}{theorem}{imposs}
\label{thm:imposs} (Insufficiency of \acrshort{sop-l} in Computing Market Equilibria)
The price vector $\mathbf{p} \in \mathbb{R}^{m}$ corresponding to the optimal dual variables of the capacity constraint~\eqref{eq:ImpossSocOpt1} of \acrshort{sop-l} may not be an equilibrium price vector, i.e., the market clearing KKT conditions of \acrshort{iop} and \acrshort{sop-l} may not be equivalent.
\end{restatable}

\begin{hproof}
We derive the first order necessary and sufficient KKT conditions of the social optimization problem \acrshort{sop-l} and show that under the optimal price vector corresponding to the dual variables of the capacity constraint, the budgets of the agents will not be completely used up. As a result, a market clearing equilibrium cannot hold.
\end{hproof}
For a detailed proof of this claim see~\ref{Appendix1}. This result establishes that the prices a social planner would set through the solution of \acrshort{sop-l} may not clear the market, as there would be agents with unused budgets.

\subsection{A Budget Perturbed Social Optimization Problem} \label{ReformSocOpt}

We now address this negative result through a reformulated social optimization problem in which we modify the budget of agents through a variable $\lambda_i$ for each agent $i$. This variable is introduced because of the additional linear constraints not present in Fisher markets and its exact value is derived in the KKT analysis in Section~\ref{SocOptKKT}. We further provide an economic interpretation of these perturbations in Section~\ref{InterpSocOpt}. The Budget Perturbed Social Optimization Problem (\acrshort{bp-sop}) is:

\begin{maxi!}|s|[2]                   
    {\mathbf{x}_i \in \mathbb{R}^m, \forall i \in [n]}                               
    {u(\mathbf{x}_1, ..., \mathbf{x}_n) = \sum_{i=1}^{n} (w_i + \lambda_i) \log \left(\sum_{j = 1}^{m} u_{ij}x_{ij} \right)
    \label{eq:SocOpt}}   
    {\label{eq:Example01}}             
    {}                                
    \addConstraint{\sum_{i=1}^{n} x_{ij}}{ = \Bar{s}_j, \forall j \label{eq:SocOpt1}}    
    \addConstraint{A_t^{(i)} \mathbf{x}_i}{\leq b_{it}, \forall t \in T_i, \forall i \in [n] \label{eq:SocOpt2}}
    \addConstraint{x_{ij}}{\geq 0, \forall i, j \label{eq:SocOpt3}}  
\end{maxi!}
with capacity Constraints~\eqref{eq:SocOpt1}, additional linear Constraints~\eqref{eq:SocOpt2} and non-negativity Constraints~\eqref{eq:SocOpt3}. 

\subsection{Deriving Perturbation Constants Using KKT Conditions} \label{SocOptKKT}

We now show that under an appropriate choice of the $\lambda_i$ perturbations for all agents $i$, the KKT conditions of \acrshort{bp-sop} and \acrshort{iop} are equivalent when prices are set through the dual variables of the capacity Constraints~\eqref{eq:SocOpt1}. Observing that for any choice of $\lambda = (\lambda_1, ..., \lambda_n) $, \acrshort{bp-sop} remains a convex optimization problem, it is necessary and sufficient to verify the first order KKT conditions for \acrshort{bp-sop} and \acrshort{iop}. To establish the first order KKT equivalence between the two problems, we define $r_{it}$ as the dual variable for the linear Constraints~\eqref{eq:SocOpt2} associated with constraint $t \in T_i$ for each agent $i$. Further, we define a \textit{fixed point} of the problem \acrshort{bp-sop} as one when $\lambda_i = \sum_{t \in T_i} r_{it} b_{it}$. With this definition, we now present the main result of this work, which establishes that the convex program \acrshort{bp-sop} computes the market equilibrium. Theorem~\ref{thm:thm2} states in one direction that the market clearing KKT conditions of \acrshort{bp-sop} are equivalent to that of the \acrshort{iop} if $\lambda_i = \sum_{t \in T_i} r_{it} b_{it}$. Furthermore, it also states the converse that any equilibrium price in the market for \acrshort{bp-sop} must correspond to a \textit{fixed point}, i.e., $\lambda_i = \sum_{t \in T_i} r_{it} b_{it}$, establishing a one-to-one correspondence between the equilibrium price vector and a fixed point solution of \acrshort{bp-sop}.


\begin{restatable}{theorem}{thmm}
\label{thm:thm2} (Convex Program to Compute Market Equilibrium)
Suppose that for each agent $i$ there is a good $j$ that is not associated with any linear Constraints~\eqref{eq:con2} and $u_{ij}>0$. Then the optimal solution of \acrshort{iop} is equivalent to that of \acrshort{bp-sop} if and only if the price vector $\mathbf{p} \in \mathbb{R}^{m}$ in the \acrshort{iop} is set based on the dual variables of the capacity Constraints~\eqref{eq:SocOpt1} at a fixed point solution of \acrshort{bp-sop}, i.e., $\lambda_i = \sum_{t \in T_i} r_{it} b_{it}$ for all $i$, where $r_{it}$ is the optimal dual multiplier of the constraint $A_t^{(i)} \mathbf{x}_i \leq b_{it}$ in \acrshort{bp-sop}.
\end{restatable}

\begin{hproof}
We first derive the necessary and sufficient first order KKT conditions for \acrshort{bp-sop} and \acrshort{iop}. The forward direction of our claim follows from considering a market equilibrium of the \acrshort{iop} and using this to show that $\lambda_i = \sum_{t \in T_i} r_{it} b_{it}$, for each agent $i$ is the \textit{fixed point} of \acrshort{bp-sop}. For the converse, we can show that if we set $\lambda_i = \sum_{t \in T_i} r_{it} b_{it}$, for each agent $i$, then each agent completely uses up their budget, while all the goods are sold to capacity.
\end{hproof}

We refer to~\ref{proof-thm2} for the complete derivation of Theorem~\ref{thm:thm2}. This result generalizes Eisenberg and Gale's convex programming methodology for classical Fisher markets to the setting of Fisher markets with additional linear constraints. Note that the technical assumption that there is a good $j$ that is not associated with any linear Constraints~\eqref{eq:con2} is the same as condition (ii) in Theorem~\ref{thm:market-eq} to establish the existence of a market equilibrium. The necessity of this assumption arises since the additional constraints may preclude agents from spending their entire budget, and thus prevent the market from clearing. As noted in Section~\ref{guarantee-existence}, we reiterate that this technical assumption is not very demanding, since agents can be allowed to keep unused budget, which can be treated as a good.


We further note that the above result does not require the linearity of the utility functions and in fact can be generalized to any homogeneous degree one concave utility function, including CES utilities. We formalize this generalization through the following corollary of Theorem~\ref{thm:thm2}.

\begin{restatable}{corollary}{corHomegeneousUtils}
\label{cor:corHomogeneousUtils} (Convex Program to Compute Market Equilibrium for Homogeneous Degree One Utilities)
Consider the problem \acrshort{iop'}, which has the same same constraints as \acrshort{iop} but its objective is replaced by a concave homogeneous degree one utility function $u_i(\mathbf{x}_i)$. Suppose that for each agent $i$ there is a good $j$ that is not associated with any linear Constraints~\eqref{eq:con2} and the agent's utility is strictly positive and increasing in the amount of good $j$ consumed. Then the optimal solution of \acrshort{iop'} is equivalent to that of \acrshort{bp-sop'} if and only if the price vector $\mathbf{p} \in \mathbb{R}^{m}$ in the \acrshort{iop'} is set based on the dual variables of the capacity Constraints~\eqref{eq:SocOpt1} at a fixed point solution of \acrshort{bp-sop'}, i.e., $\lambda_i = \sum_{t \in T_i} r_{it} b_{it}$ for all $i$, where $r_{it}$ is the optimal dual multiplier of the constraint $A_t^{(i)} \mathbf{x}_i \leq b_{it}$ in \acrshort{bp-sop'}. Here \acrshort{bp-sop'} has the same constraints as \acrshort{bp-sop} but its concave objective is $\sum_{i=1}^{n} (w_i + \lambda_i) \log \left(u_i(\mathbf{x}_i) \right)$, where $u_i(\mathbf{x}_i)$ is a concave homogeneous degree one utility function.
\end{restatable}

The proof of the above corollary follows from the fact that for any homogeneous degree one utility function $u_i(\mathbf{x}_i)$, which includes the important class of CES utilities, we have that $u_i(\mathbf{x}_i) = \nabla u_i(\mathbf{x}_i) \cdot \mathbf{x}_i$. This was precisely the property of linear utility functions that was used in the proof of Theorem~\ref{thm:thm2} and thus, we omit the proof of the above corollary. Corollary~\ref{cor:corHomogeneousUtils} establishes that market equilibria can be computed using \acrshort{bp-sop} for a broad range of utility functions in Fisher markets with additional linear constraints.

\subsection{Market Equilibrium with Homogeneous Linear Constraints} \label{homogeneous-equilibrium-computation}

We now consider the important special case of homogeneous linear constraints, i.e., where $b_{it} = 0$ for each $t \in T_i$ for all $i$. This setting is relevant since it includes the set of proportionality constraints and as compared to non-homogeneous linear constraints offers better computational properties, which we elucidate in this section.


We begin by noting a direct consequence of Theorem~\ref{thm:thm2} for homogeneous linear constraints. In particular, when agents have homogeneous linear constraints then market equilibria can in fact be computed using the convex program \acrshort{sop-l}. That is, \acrshort{bp-sop} reduces to \acrshort{sop-l} in the setting of homogeneous linear constraints. Corollary~\ref{cor:homogeneous-SOP1} formalizes this result.

\begin{corollary} \label{cor:homogeneous-SOP1} (Convex Program to Compute Market Equilibrium for Homogeneous Linear Constraints)
Suppose that a market equilibrium exists for \acrshort{iop} with homogeneous linear constraints. Then \acrshort{bp-sop} is equivalent to \acrshort{sop-l}, and the price vector $\mathbf{p}$ corresponding to the optimal dual variables of the capacity constraint~\eqref{eq:ImpossSocOpt1} of \acrshort{sop-l} establishes the market equilibrium prices.
\end{corollary}
The proof of this corollary follows from Theorem~\ref{thm:thm2} with the observation that $\lambda_i = \sum_{t \in T_i} r_{it} b_{it} = 0$, since $b_{it} = 0$ for each $i, t$. This result implies that equilibrium prices for Fisher markets with homogeneous linear constraints can be computed using a convex program whose objective is independent of the dual variables of the additional linear constraints. Note that this is contrary to the setting of non-homogeneous linear constraints, where $b_{it}>0$ for some $i, t$. As a result, market equilibria for Fisher markets with homogeneous linear constraints can be easily computed using state-of-the-art convex programming solvers.



Using the result of Corollary~\ref{cor:homogeneous-SOP1}, we can further show that the solution of \acrshort{sop-l} can be used to test for the existence of non-negative equilibrium price vectors.


\begin{corollary}  \label{cor:computing-homogeneous} (Convex Program to Test for Existence of Non-negative Equilibrium Price for Homogeneous Linear Constraints)
Suppose that all agents have homogeneous linear constraints, i.e., $b_{it} = 0$ for each $i, t$, and consider the convex program where the capacity Constraints~\eqref{eq:ImpossSocOpt1} of \acrshort{sop-l} are replaced with inequalities, i.e., $\sum_{i = 1}^{n} x_{ij} \leq \Bar{s}_j$ for all $j \in [m]$. Then a non-negative market equilibrium exists if and only if the capacity constraints are met with equality at the optimal solution of this convex program.
\end{corollary}

\begin{proof}
In one direction, if a non-negative market equilibrium exists then by definition of an equilibrium price vector the capacity constraints of the convex program must be met with equality. In the other direction, if the capacity constraints of the new convex program are met with equality at its optimal solution, then the prices can be set based on the optimal dual variables of the capacity constraints of \acrshort{sop-l}, which corresponds to a market equilibrium price by Corollary~\ref{cor:homogeneous-SOP1}.
\end{proof}


The above corollary implies that the test for existence of a non-negative equilibrium price vector with homogeneous linear constraints involves solving a convex program and checking if the capacity constraints are met with equality. This result overcomes the question of the existence of an equilibrium price vector for Fisher markets with homogeneous linear constraints by providing a necessary and sufficient condition to test for equilibrium existence. We note that the result of Corollary~\ref{cor:computing-homogeneous} cannot be directly extended to the case of non-homogeneous linear constraints. This is because when some $b_{it} >0$ then the convex program \acrshort{bp-sop} has variables $\lambda_i$ in the objective that depend on the dual variables of the additional linear constraints, which we do not have knowledge of \emph{a priori}.



\subsection{Pareto Efficiency of the Optimal Allocation} \label{pareto-efficiency}

In this section, we establish the Pareto efficiency of the optimal allocation resulting from the solution of \acrshort{bp-sop} corresponding to the \textit{fixed point} budget perturbations. Before presenting the result on the Pareto efficiency of the optimal allocation of \acrshort{bp-sop}, we begin by formalizing the notion of Pareto efficiency through the following definitions.

\begin{definition} (Pareto Dominance)
An allocation $\left( \mathbf{x}_1, ...,\mathbf{x}_n \right)$ is said to Pareto dominate an allocation $\left( \mathbf{y}_1, ..., \mathbf{y}_n \right)$ if $u_i(\mathbf{x}_i) \geq u_i(\mathbf{y}_i)$, for all $i$ and there exists an agent $k$, such that $u_k(\mathbf{x}_k) > u_k(\mathbf{y}_k)$. Here $\mathbf{x}_i, \mathbf{y}_i \in \mathbb{R}^m$ for all $i \in [n]$.
\end{definition}

\begin{definition} (Pareto Efficiency)
An allocation ($\mathbf{x}_1$, ..., $\mathbf{x}_n$) is said to be Pareto efficient if it is not Pareto dominated by any other feasible allocation. Here $\mathbf{x}_i \in \mathbb{R}^m$ for all $i \in [n]$.
\end{definition}

To rephrase the above notion of Pareto efficiency, we have that if an allocation is Pareto efficient then no agent can become better off at some other allocation without making some other agent worse off. Now, to establish that the optimal allocation corresponding to the solution of \acrshort{bp-sop} is Pareto efficient, we will leverage the following theorem from \cite{kreps-book}.

\begin{theorem} \label{thm:pareto-efficiency-kreps} (\cite{kreps-book} Pareto Optimality of General Social Choice Functions)
Suppose that $W: \mathbb{R}^{n} \mapsto \mathbb{R}$ is strictly increasing and $u$ is the vector function of individual utilities $u_i : \mathbb{R}^{m} \mapsto \mathbb{R}$. If $(\mathbf{x}_1^*, ..., \mathbf{x}_n^*)$ is a solution to the problem
$\max_{(\mathbf{x}_1, ..., \mathbf{x}_n) \in X} W (u(\mathbf{x}_1, ..., \mathbf{x}_n))$
then $(\mathbf{x}_1^*, ..., \mathbf{x}_n^*)$ is Pareto efficient. Here $X$ represents the feasible space of all allocations and $\mathbf{x}_i \in \mathbb{R}^m$ for all $i \in [n]$.
\end{theorem}

We can apply Theorem~\ref{thm:pareto-efficiency-kreps} to show that the optimal solution of \acrshort{bp-sop} is pareto optimal through the following proposition.

\begin{proposition} \label{pareto-eff} (Pareto Optimality of \acrshort{bp-sop})
The optimal allocation $\mathbf{x}_i^*$ for all agents $i$ corresponding to a \textit{fixed point} solution of \acrshort{bp-sop} is Pareto optimal.
\end{proposition}

\begin{proof}
To prove this result, we need to find a strictly increasing function $W$, such that $W \circ u(\mathbf{x}_1, ..., \mathbf{x}_n) = \sum_i (w_i + \lambda_i) \log(u_i(\mathbf{x}_i))$. Now, since $u(\mathbf{x}_1, ..., \mathbf{x}_n) = (u_1(\mathbf{x}_1), ..., u_n(\mathbf{x}_n))$, we can take $W$ as the function with the $i$'th component as $W_i(y) = (w_i + \lambda_i) \log(y)$, where $\lambda_i$ corresponds to the fixed point solution of \acrshort{bp-sop}. Note that without loss of generality $w_i>0$, as if $w_i = 0$, then this agent cannot purchase any goods and so can be removed from consideration when formulating the social optimization problem. This implies that $w_i + \lambda_i > 0$, as $w_i>0$ and $\lambda_i \geq 0$. Furthermore, since $\log(y)$ is strictly increasing in $y$, we have that $W_i(y)$ is strictly increasing in $y$. Now certainly for such a function $W$, it follows that $W \circ u(\mathbf{x}_1, ..., \mathbf{x}_n) = \sum_i (w_i + \lambda_i) \log(u_i(\mathbf{x}_i))$. Invoking Theorem~\ref{thm:pareto-efficiency-kreps} where the set $X$ is defined by the Constraints~\eqref{eq:SocOpt1}-\eqref{eq:SocOpt3}, we have that the optimal allocation $\mathbf{x}_i^*$ for all agents $i$ corresponding to a fixed point solution of \acrshort{bp-sop} is Pareto optimal, proving our claim.
\end{proof}

Proposition~\ref{pareto-eff} establishes that the optimal solution to \acrshort{bp-sop} at the \textit{fixed point} is a desirable market outcome since there is no other feasible allocation that makes all agents at least as well off and at least one agent strictly better off.


\subsection{Economic Relevance of Solution of \acrshort{bp-sop}} \label{InterpSocOpt}

The solution of the convex program \acrshort{bp-sop} provides a method to set market clearing prices when agents have additional linear constraints. In this section we provide an economic interpretation to motivate the necessity of the budget perturbations in setting market clearing prices. To this end, we first consider the classical Fisher market social objective~\eqref{eq:FisherSocOpt}, which is also the objective function of \acrshort{sop-l}. For this ``budget weighted log utility'' objective, we can interpret the budget $w_i$ as the priority of agent $i$ in the market. This method of prioritization based on agent's budgets is the right one to set market clearing prices when each agent in the market has just budget constraints and goods have capacity constraints. However, prioritizing agents based on their budgets is insufficient in guaranteeing that the market clears when agents have additional linear Constraints~\eqref{eq:con2}, as was established in Theorem~\ref{thm:imposs}.

As a result, we require a different method of prioritization of agents when they have additional linear Constraints~\eqref{eq:con2}. Mathematically, the objective function of \acrshort{bp-sop}~\eqref{eq:SocOpt} provides us with the ``right'' level of prioritization of agents through appropriately chosen budget perturbations. Since the budget perturbations $\lambda_i$ depend on the dual variables of the additional linear constraints, as was established in Theorem~\ref{thm:thm2}, we have that agents are re-prioritized based on how ``constrained'' they are. That is, if agents are highly constrained, e.g., healthcare workers in a pandemic environment who may have little availability to use public spaces in a public goods market, then they are assigned a greater priority through a larger weight $\lambda_i$ as compared to less constrained agents. This re-prioritization of agents ensures that more constrained agents receive goods that lie within their small feasible constraint set, which they would not have been able to avail through the solution of \acrshort{sop-l} with the classical Fisher objective. As for less constrained agents, they have far more slack in the goods that lie within their feasible constraint set and thus a smaller weight $\lambda_i$, i.e., a lower level of prioritization, suffices for the market to clear.

Finally, we reiterate that while the prices are set using the solution of \acrshort{bp-sop} the budget perturbations are irrelevant from the point of view of the agents. That is, each agent solves their own \acrshort{iop} in which they only see the equilibrium price vector \textbf{p}. This observation follows directly from the result of Theorem~\ref{thm:thm2}, establishing that each agent obtains their most preferred bundle of goods given the prices set through the solution of \acrshort{bp-sop}.


\section{Fixed Point Iteration to Determine Budget Perturbations} \label{FixedPoint}

We have established that to set market clearing prices, a social planner should solve \acrshort{bp-sop}, where the budget perturbations are set such that $\lambda_i = \sum_{t \in T_i} r_{it} b_{it}$. However, this implies that $\lambda_i$ depends on the dual variables of the problem, which we have no knowledge of \emph{a priori}. In this section, we show how to compute the appropriate value of $\lambda_i$, through a fixed point iteration that directly solves the centralized problem \acrshort{bp-sop} (Section~\ref{algorithm-fixed-pt}) and numerically validate its convergence through experiments (Section~\ref{experiments}).


\subsection{Fixed Point Iteration Algorithm} \label{algorithm-fixed-pt}

To determine the true value of the perturbation parameters specified by the vector $\lambda \in \mathbb{R}_{\geq 0}^{n}$, we consider an iterative scheme of the form $G\left(\lambda_1^{(k)}, ..., \lambda_n^{(k)}\right) = \left(\mathbf{r}_{1}^{(k)}, ..., \mathbf{r}_n^{(k)}\right)$, where we update the perturbations as $\left(\lambda_1^{(k+1)}, ..., \lambda_n^{(k+1)}\right) = \left(\sum_{t \in T_i} r_{1t}^{(k)} b_{1t}, ..., \sum_{t \in T_i} r_{nt}^{(k)} b_{nt} \right)$. Here $G$ is a function that takes in the $k^{th}$ iterate of the budget perturbations $\lambda_i^{(k)}$ for all agents $i$, solves the corresponding social optimization problem \acrshort{bp-sop} and returns the dual variables, $\mathbf{r}_i^{(k)} \in \mathbb{R}_{\geq 0}^{l_i}$ of the linear Constraints~\eqref{eq:SocOpt2}.

Algorithm~\ref{alg:Algo1} depicts the fixed point iterative scheme, where $\lambda = (\lambda_1, ..., \lambda_{n})$ is the vector of budget perturbations, and $\mathbf{R} = (\mathbf{r}_1, ..., \mathbf{r}_n)$ is the matrix of the dual variables of \acrshort{bp-sop}s linear Constraints~\eqref{eq:SocOpt2}. In this algorithm, first the budget perturbations for each agent are initialized to some non-negative constant, e.g., zero, and the problem \acrshort{bp-sop} is solved given these budget perturbations. The dual variables of this problem are then used to reset the budget perturbations in \acrshort{bp-sop} using the relationship $\lambda_i = \sum_{t \in T_i} r_{it} b_{it}$ after which \acrshort{bp-sop} is re-computed. This procedure is repeated until the $L_2$ norm of the difference between the $\lambda$ perturbations and the dual variables of \acrshort{bp-sop} becomes small.

\begin{algorithm} 
\label{alg:Algo1}
\SetAlgoLined
\SetKwInOut{Input}{Input}\SetKwInOut{Output}{Output}
\Input{Function $G(\cdot)$ to calculate dual variables of linear Constraints~\eqref{eq:SocOpt2} of \acrshort{bp-sop}, $\lambda^{(0)} \in \mathbb{R}_{\geq 0}^{n}$}
\Output{Budget Perturbation Parameters $\lambda$}
$\mathbf{R}^{(0)} \leftarrow G(\lambda)$ \;
 \For{$k = 0, 1, 2, ...$}{
  $\lambda_i^{(k+1)} \leftarrow \sum_{t \in T_i} r_{it}^{(k)} b_{it}$ $\forall i$ \;
  $\mathbf{R}^{(k+1)} \leftarrow G(\lambda^{(k+1)})$ \;
  }
\caption{Fixed Point Iterative Scheme}
\end{algorithm}

\subsection{Numerical Experiments with Fixed Point Iterative Scheme} \label{experiments}

We now numerically evaluate the convergence of Algorithm~\ref{alg:Algo1} on a problem with $n=200$ agents and $m = 6$ goods, with three identical knapsack constraints for each agent. In particular, we suppose that $x_{i1} + x_{i2} \leq 1$, $x_{i3} + x_{i4} \leq 1$ and $x_{i5} + x_{i6} \leq 1$ for all agents $i$ in the market. Each agent $i$ has three linear Constraints~\eqref{eq:SocOpt2} and so we denote the set $T_i = \{1, 2, 3 \}$. The capacities of the goods are $\Bar{s}_j = 100$, for each good $j \in [m]$. Furthermore, in this experiment each agent $i$ is endowed with a random budget $w_i$ and their preferences are captured through randomly generated utilities. Both the budgets and utilities are generated from the uniform $[0, 1]$ distribution.

For the above problem instance, Figure~\ref{BudgetUtilityPerturbed} plots the $L_2$ norm of the difference between the $\lambda$ perturbations and the dual variables of \acrshort{bp-sop} at each iteration, given by $\left\|\lambda^{(k)} - \sum_{t = 1}^{3} \mathbf{r}_t^{(k)}\right\|_2$, since $b_{it} = 1$ for all $t \in T_i$ for each agent $i$. Here $\lambda^{(k)} = \left(\lambda_1^{(k)}, ..., \lambda_{n}^{(k)}\right)$ and $\mathbf{r}_t^{(k)} = \left(r_{1t}^{(k)}, ..., r_{nt}^{(k)}\right)$, where $r_{it}^{(k)}$ is the dual variable of the \acrshort{bp-sop} at iteration $k$, and $n$ is the number of agents. Figure~\ref{BudgetUtilityPerturbed} confirms that the iterative scheme converges quickly to a fixed point by showing convergence within 40 iterations, which highlights the computational feasibility of our mechanism. We note that the experiments also confirmed feasibility of allocations with respect to additional linear Constraints~\eqref{eq:SocOpt2}, which may have been violated with classical Fisher markets.

\begin{figure}[!h]
      \centering
      \includegraphics[width=0.7\linewidth]{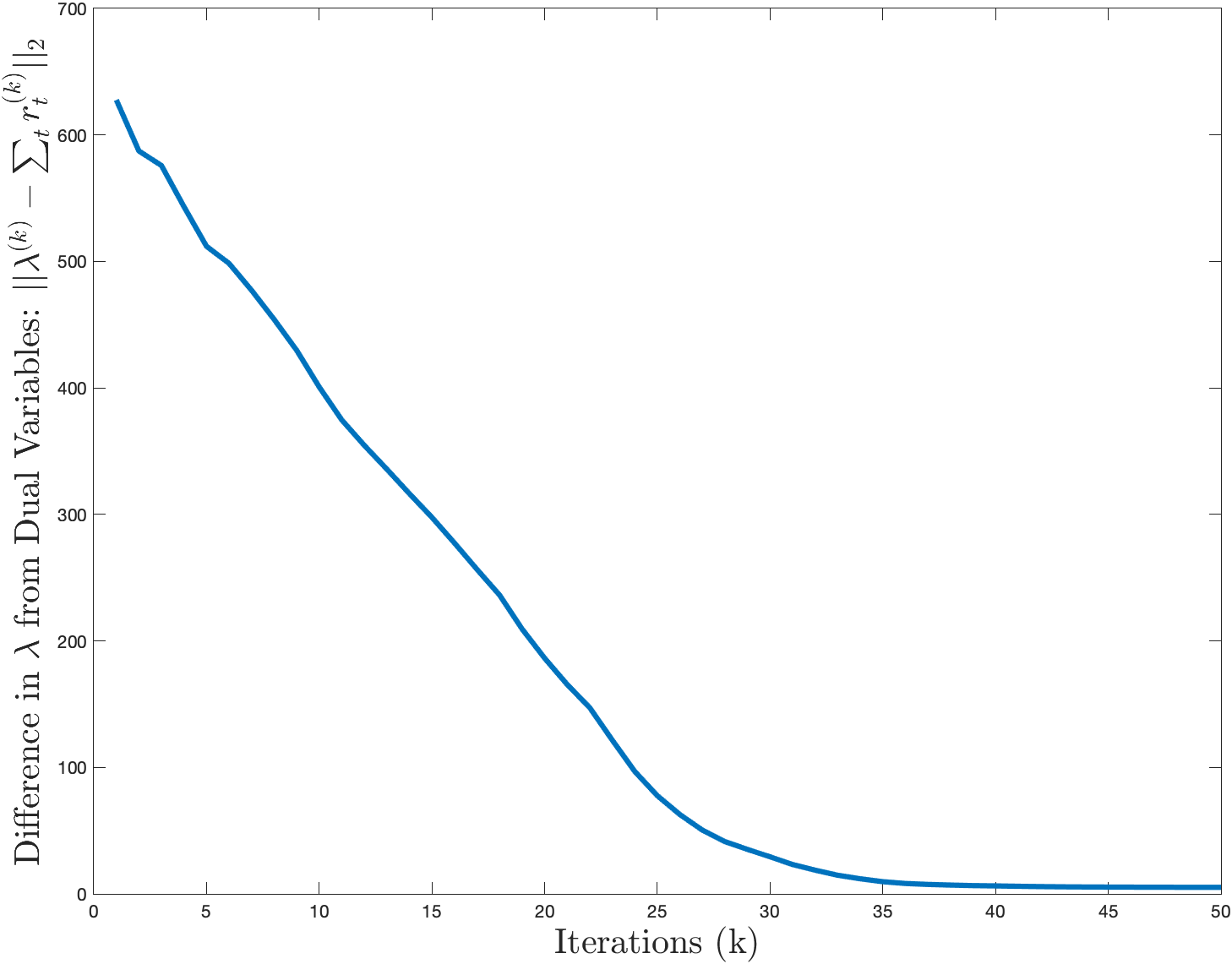}
      \caption{Numerical Convergence of fixed point iterative scheme with 200 agents, 6 goods and 3 knapsack constraints for each agent. The budgets and utilities were assigned randomly using a Uniform $[0, 1]$ distribution to agents.}
      \label{BudgetUtilityPerturbed}
   \end{figure}

\section{Distributed Algorithms to Compute Equilibrium Prices} \label{ADMM-Convergence-Guarantees}


While the fixed point iterative scheme provided fast numerical convergence results, it (i) involved the solution of a centralized optimization problem at each step, which can be quite computationally intensive, and (ii) required complete information of agent's utilities to derive prices. In this section, we present a new class of distributed algorithms, based on Alternating Direction Methods (ADMs), to determine the equilibrium prices whilst overcoming the implementation hurdles of the centralized fixed point iterative scheme. We also present theoretical convergence results for the distributed algorithms for Fisher markets with homogeneous linear constraints. Furthermore, we provide a direct extension of the convergence results of the Alternating Direction Method of Multipliers (ADMM) algorithm for classical Fisher markets, which is of independent interest, since this algorithm converges for all concave homogeneous degree one utility functions including linear utilities. This overcomes a fundamental barrier in equilibrium computation in classical Fisher markets, since existing distributed algorithms require the utility functions of agents to be strongly concave and have a convergence rate that depends on the utility functions of agents. Our algorithms are motivated by the tatonnement based distributed algorithms in the market literature, wherein price updates are made based on the discrepancy between supply and demand. However, contrary to the tatonnement literature wherein the step size for the price updates is often dependent on the specific utility function of agents, we develop a distributed algorithm that is ``step-size free'', i.e., the step size is fixed to a parameter $\beta > 0$ irrespective of the problem instance. We further believe that the application of methods such as ADMM may be of independent interest in the economics and game theory communities.


We begin this section by providing a survey of the algorithms and results for Alternating Direction Methods (ADMs) relevant to this work (Section~\ref{ADMM-math-prelims}). We then introduce an AMA and ADMM algorithm for Fisher markets with homogeneous linear constraints and present convergence results in Sections~\ref{AMA-convergence} and~\ref{ADMM-classical-fisher}. Next, we validate these theoretical convergence guarantees through numerical experiments in Section~\ref{numerical-conv-admm}. While the convergence guarantees of ADMM do not extend to Fisher markets with non-homogeneous linear constraints, we present numerical convergence results for Fisher markets with non-homogeneous linear constraints in Section~\ref{ADMM-physical-constraints}.

\subsection{Alternating Direction Methods} \label{ADMM-math-prelims}

Alternating Direction Methods (ADMs) have been used widely for distributed optimization purposes and these algorithms are primarily concerned with the following linearly constrained programs that have a separable structure. We present the results and algorithms for the following maximization problem:

\begin{maxi!}|s|[2]                   
    {\mathbf{x} \in \mathcal{X}, \mathbf{y} \in \mathcal{Y} }                               
    {h(\mathbf{x}, \mathbf{y}) = f(\mathbf{x}) + g(\mathbf{y}) \label{eq:eqADMMIntro1}}   
    {\label{eq:ADMMIntroExample001}}             
    {}                                
    \addConstraint{A \mathbf{x} + B \mathbf{y}}{= \mathbf{c} \label{eq:conADMMIntro1}}    
\end{maxi!}
where $\mathcal{X} \in \mathbb{R}^{n_1}$ and $\mathcal{Y} \in \mathbb{R}^{n_2}$ are closed convex sets, $A \in \mathbb{R}^{ \Bar{m} \times n_1}$, $B \in \mathbb{R}^{\Bar{m} \times n_2}$ and $\mathbf{c} \in \mathbb{R}^{\Bar{m}}$ for some $n_1, n_2, \Bar{m} \in \mathbb{N}$. Furthermore, $f:\mathbb{R}^{n_1} \rightarrow \mathbb{R}$ and $g:\mathbb{R}^{n_2} \rightarrow \mathbb{R}$ and $h:\mathbb{R}^{n_1 + n_2} \rightarrow \mathbb{R}$. 


Denoting $\mathbf{\lambda} \in \mathbb{R}^{\Bar{m}}$ as the dual variable of constraint~\eqref{eq:conADMMIntro1}, we can formulate the following augmented Lagrangian of the optimization Problem~\eqref{eq:eqADMMIntro1}-\eqref{eq:conADMMIntro1} for some parameter $\beta>0$:

\[ \mathcal{L}_{\beta}(\mathbf{x}, \mathbf{y}) = f(\mathbf{x}) + g(\mathbf{y}) - \mathbf{\lambda}^{T} (A \mathbf{x} + B \mathbf{y} - \mathbf{c}) - \frac{\beta}{2} \left \lVert A \mathbf{x} + B \mathbf{y} - \mathbf{c}  \right \rVert^2 \]
Note that $\beta = 0$ gives the Lagrangian.

We first present the two-block Alternating Direction Method of Multipliers (ADMM) in Algorithm~\ref{alg:AlgoStandardADMM}. In every iteration of ADMM, we first maximize the augmented Lagrangian only with respect to $\mathbf{x}$ and then only with respect to $\mathbf{y}$, after which we update the dual multiplier $\lambda$.


\begin{algorithm} 
\label{alg:AlgoStandardADMM}
\SetAlgoLined
\SetKwInOut{Input}{Input}\SetKwInOut{Output}{Output}
\Input{Initial dual multiplier $\mathbf{\lambda}^{(0)}$, and initial vector $\mathbf{y}^{(0)}$}
 \For{$k = 0, 1, 2, ...$}{
  $\mathbf{x}^{(k+1)} = \argmax_{\mathbf{x} \in \mathcal{X}}  \mathcal{L}_{\beta} (\mathbf{x}, \mathbf{y}^{(k)}) $ \;
  $\mathbf{y}^{(k+1)} = \argmax_{\mathbf{y} \in \mathcal{Y}} \mathcal{L}_{\beta}(\mathbf{x}^{(k+1)}, \mathbf{y}) $ \;
  $\mathbf{\lambda}^{(k+1)} \leftarrow \mathbf{\lambda}^{(k)} + \beta (A \mathbf{x}^{(k+1)} + B \mathbf{y}^{(k+1)} - \mathbf{c})$ \;}
\caption{Two Block ADMM}
\end{algorithm}

A major focus of the ADMM literature has been towards the characterization of the convergence rate of Algorithm~\ref{alg:AlgoStandardADMM}. In particular, it has been established in \cite{admm-main-conv} that Algorithm~\ref{alg:AlgoStandardADMM} converges under concavity assumptions of the functions $f(\mathbf{x})$ and $g(\mathbf{y})$.

\begin{theorem} (\cite{admm-main-conv} Two Block ADMM Convergence) \label{thm:ADMMStandardConvergence}
If $f(\mathbf{x})$ and $g(\mathbf{y})$ are concave functions in their respective arguments then Algorithm~\ref{alg:AlgoStandardADMM} converges to the optimal solution of Problem~\eqref{eq:eqADMMIntro1}-\eqref{eq:conADMMIntro1} with rate $O(\frac{1}{k})$, where $k$ is the number of iterations of Algorithm~\ref{alg:AlgoStandardADMM}.
\end{theorem}



Another well studied alternating direction method is the Alternating Minimization Algorithm (AMA) as proposed by Tseng \cite{ama-convergence}. This algorithm closely mirrors the ADMM algorithm other than the update of the first block of variables, which are updated using the Lagrangian $\mathcal{L}_{0}$ rather than the augmented Lagrangian $\mathcal{L}_{\beta}$ for some $\beta>0$. The maximization variant of the AMA algorithm is presented in Algorithm~\ref{alg:AlgoStandardAMA}.

\begin{algorithm} 
\label{alg:AlgoStandardAMA}
\SetAlgoLined
\SetKwInOut{Input}{Input}\SetKwInOut{Output}{Output}
\Input{Initial dual multiplier $\mathbf{\lambda}^{(0)}$, and initial vector $\mathbf{y}^{(0)}$}
 \For{$k = 0, 1, 2, ...$}{
  $\mathbf{x}^{(k+1)} = \argmax_{\mathbf{x} \in \mathcal{X}}  \mathcal{L}_{0}(\mathbf{x}, \mathbf{y}^{(k)}) $ \;
  $\mathbf{y}^{(k+1)} = \argmax_{\mathbf{y} \in \mathcal{Y}}  \mathcal{L}_{\beta}(\mathbf{x}^{(k+1)}, \mathbf{y}) $ \;
  $\mathbf{\lambda}^{(k+1)} \leftarrow \mathbf{\lambda}^{(k)} + \beta (A \mathbf{x}^{(k+1)} + B \mathbf{y}^{(k+1)} - \mathbf{c})$ \;}
\caption{Two Block AMA}
\end{algorithm}

The convergence of the above algorithm to the optimal solution of Problem~\eqref{eq:eqADMMIntro1}-\eqref{eq:conADMMIntro1} is guaranteed under the assumption that the function $f(\mathbf{x})$ is strongly concave \cite{ama-convergence}. A formal statement of this result is provided below. 

\begin{theorem} (\cite{ama-convergence} Two Block AMA Convergence) \label{thm:AMAStandardConvergence}
If $f(\mathbf{x})$ is a strongly concave function with strong concavity modulus $\sigma_f$ and $g(\mathbf{y})$ is a concave function then Algorithm~\ref{alg:AlgoStandardAMA} converges to the optimal solution of Problem~\eqref{eq:eqADMMIntro1}-\eqref{eq:conADMMIntro1} with rate $O(\frac{1}{k})$, when $\beta < 2\frac{\sigma_f}{\rho(B^T B)}$. Here $\rho(B^T B)$ is the spectral radius of the matrix $B^T B$ and $k$ is the number of iterations of Algorithm~\ref{alg:AlgoStandardAMA},.
\end{theorem}
Here the strong concavity modulus is defined as follows.

\begin{definition} (Strong Convexity Modulus) \label{def:strong-convexity}
We say a function $f: \mathbb{R}^{n} \rightarrow \mathbb{R}$ has a strong concavity modulus $\sigma_f \geq 0$ if for all $x, y \in \mathbb{R}^{n}$ and $t \in[0,1]$ we have that: 
$$
f(t y + (1-t) x) \geq t  f(y)+(1-t) f(x) + \frac{\sigma_f}{2}  t (1-t) \|x-y\|_{2}^{2}
$$
\end{definition}




\subsection{AMA for Fisher Markets with Homogeneous Linear Constraints} \label{AMA-convergence}

Traditional tatonnement algorithms for classical Fisher markets involve agents distributedly solving their individual optimization problems with prices of the goods being adjusted upwards when the demand for that good exceeds supply and adjusted downwards when supply exceeds demand. In this section, we present and characterize the convergence of an AMA based tatonnement algorithm to distributedly compute market equilibrium prices when agents have homogeneous linear constraints. We use the algorithm and results for AMA presented in Section~\ref{ADMM-math-prelims}.

To apply AMA, we begin by reformulating \acrshort{bp-sop} in the two-block form as in the canonical form of Problem~\eqref{eq:eqADMMIntro1}-\eqref{eq:conADMMIntro1}. To motivate the two-block form that we choose, we note that Algorithm~\ref{alg:AlgoStandardAMA} for AMA involves the formulation of an augmented Lagrangian, which is non-separable due to the capacity constraints that couple the allocations of all agents. Thus, to achieve a distributed implementation we introduce a new vector of consumption bundles $\mathbf{y}_i \in \mathbb{R}^m$ for each agent $i$ and define the following reformulation of \acrshort{bp-sop}, denoted as \acrshort{bp-sop-adm}:

\begin{maxi!}|s|[2]                   
    {\mathbf{x}_i \in \mathcal{X}_i, \mathbf{y}_i \in \mathcal{Y}_i, \forall i \in [n]}                               
    {\sum_i w_i \log \left(u_i(\mathbf{x}_i)\right) \label{eq:AMAOptPhys}}   
    {\label{eq:AMAPhysExample1}}             
    {}                                
    \addConstraint{\sum_i y_{ij}}{ = \Bar{s}_j, \forall j \in [m] \label{eq:AMAOptPhys1}}    
    \addConstraint{\mathbf{y}_i}{ = \mathbf{x}_i, \forall i \in [n] \label{eq:AMAOptPhys2}}
\end{maxi!}
Here $\mathcal{X}_i =  \{\mathbf{x}_i: \mathbf{x}_i \in \mathbb{R}_{\geq 0}^{m}, A_t^{(i)} \mathbf{x}_i \leq 0, \forall t \in T_i \}$, $\mathcal{Y}_i = \{\mathbf{y}_i: \mathbf{y}_i \in \mathbb{R}^m\}$ and $u_i(\mathbf{x}_i)$ is a homogeneous degree one utility function. Note further that $\lambda_i = 0$ by Theorem~\ref{thm:thm2} since all the additional linear constraints are homogeneous, i.e., $b_{it} = 0$ for all $i, t$. To ensure the separability of the augmented Lagrangian in terms of the individual consumption bundles $\mathbf{x}_i$, we replaced $x_{ij}$ with $y_{ij}$ in~\eqref{eq:AMAOptPhys1} and added the constraint $\mathbf{x}_i = \mathbf{y}_i$ in~\eqref{eq:AMAOptPhys2}. Note that with the nature of the added constraints the problem \acrshort{bp-sop-adm} will yield the same solution as the problem \acrshort{bp-sop}. Further, note that the objective is independent of $\mathbf{y}_i$, which implies that the objective is $h(\mathbf{x}, \mathbf{y}) = f(\mathbf{x}) + g(\mathbf{y})$, where $\mathbf{x} = (\mathbf{x}_1, ..., \mathbf{x}_n)$, $\mathbf{y} = (\mathbf{y}_1, ..., \mathbf{y}_n)$, $f(\mathbf{x}) = \sum_i w_i \log \left(\sum_j u_{ij} x_{ij}\right)$ and $g(\mathbf{y}) = 0$.

To present the AMA algorithm for the reformulated problem \acrshort{bp-sop-adm}, we denote $\mathbf{p}$ as the dual variable of the capacity constraint~\eqref{eq:AMAOptPhys1} and $\mathbf{\Tilde{p}}_i$ as the dual variable of the constraint~\eqref{eq:AMAOptPhys2} for each agent $i$. Here $\mathbf{p}$ is interpreted as the price set in the market and $\mathbf{\Tilde{p}}_i$ as the individual specific price vectors. Algorithm~\ref{alg:Algo-BP-SOP-AMA} presents the distributed implementation to solve \acrshort{bp-sop-adm} using the AMA algorithm elucidated in Algorithm~\ref{alg:AlgoStandardAMA}.

\begin{algorithm} 
\label{alg:Algo-BP-SOP-AMA}
\SetAlgoLined
\SetKwInOut{Input}{Input}\SetKwInOut{Output}{Output}
\Input{Initial price vector $\mathbf{p}^{(0)}$, Individual price vectors $\mathbf{\Tilde{p}}_i^{(0)}$ for all $i$, Initial baseline demand $\mathbf{y}_i^{(0)}$, and Step size $\beta$}
\Output{Equilibrium Price vector $\mathbf{p}^*$}
 \For{$k = 0, 1, 2, ...$}{
  {$\mathbf{x}_i^{(k+1)} = \argmax_{\mathbf{x}_i \in \mathcal{X}_i} \left\{w_i \log \left(u_i(\mathbf{x}_i) \right) - \sum_j \Tilde{p}_{ij}^{(k)} x_{ij}  \right\} $}, for all $i$ \;
  $\mathbf{y}^{(k+1)} = \argmax_{\mathbf{y}} \small\{\sum_i \sum_j \Tilde{p}_{ij}^{(k)} y_{ij} - \frac{\beta}{2} \sum_{i, j} (x_{ij}^{(k+1)}-y_{ij})^2 - \sum_j p_j^{(k)} (\sum_i y_{ij} - \Bar{s}_j) - \frac{\beta}{2} \sum_j \left(\sum_i y_{ij} - \Bar{s}_j \right)^2  \small\}$ \;
  $\mathbf{p}^{(k+1)} \leftarrow \mathbf{p}^{(k)} + \beta (\sum_{i} \mathbf{y}_i^{(k+1)} - \Bar{\mathbf{s}})$, where $\Bar{\mathbf{s}} = (\Bar{s}_1, ..., \Bar{s}_m)$ \;
  $\mathbf{\Tilde{p}}_i^{(k+1)} \leftarrow \mathbf{\Tilde{p}}_i^{(k)} + \beta ( \mathbf{x}_i^{(k+1)} - \mathbf{y}_i^{(k+1)}))$ for all $i$
  }
\caption{Two Block AMA for \acrshort{bp-sop-adm}}
\end{algorithm}

A few comments about Algorithm~\ref{alg:Algo-BP-SOP-AMA} are in order. First, the separability of the augmented Lagrangian of \acrshort{bp-sop-adm} enabled each agent to solve an individual objective through an update of the $\mathbf{x}_i$'s at each step, which gives the desired distributed implementation. In fact each individual agent solves their individual optimization problem when observing a price vector $\mathbf{\Tilde{p}}_i^{(k)}$ at iteration $k$ due to the equivalence of the KKT conditions of the \acrshort{iop} and the objective ($w_i \log \left(u_i(\mathbf{x}_i) \right) - \sum_j \Tilde{p}_{ij}^{(k)} x_{ij}$) in Algorithm~\ref{alg:Algo-BP-SOP-AMA} \cite{nesterov-fisher-gale}. Finally, we note that the individual objectives of agents in the update of $\mathbf{x}_i$'s depend on individually observed prices rather than the market prices. However, if we initialize the individual price vectors to the initial market price, then at each iteration the individual price vectors of agents will always equal the market price. We formalize this in the following proposition.


\begin{proposition} \label{prop-prices-same} (Equivalence of Individual and Market Prices)
If $\mathbf{p}^{(0)} = \mathbf{\Tilde{p}}_i^{(0)}$ for all $i$ in Algorithm~\ref{alg:Algo-BP-SOP-AMA} then $\mathbf{p}^{(k)} = \mathbf{\Tilde{p}}_i^{(k)}$ for all $i$ at each iteration $k$.
\end{proposition}

\begin{proof}
Since we have $\mathbf{p}^{(0)} = \mathbf{\Tilde{p}}_i^{(0)}$, to prove this claim it suffices to show that if $\mathbf{p}^{(k)} = \mathbf{\Tilde{p}}_i^{(k)}$ at iteration $k$ then $\mathbf{p}^{(k+1)} = \mathbf{\Tilde{p}}_i^{(k+1)}$. Thus, suppose $\mathbf{p}^{(k)} = \mathbf{\Tilde{p}}_i^{(k)}$ at iteration $k$. Then to update $\mathbf{y}_i$ for all $i$, we consider:

\[ \max_{\{\mathbf{y}_i\}_{i = 1}^{n}} \left\{\sum_i \sum_j \Tilde{p}_{ij}^{(k)} \mathbf{y}_{ij} - \frac{\beta}{2} \sum_{i, j} (x_{ij}^{(k+1)}-y_{ij})^2 - \sum_j p_j^{(k)} (\sum_i y_{ij} - \Bar{s}_j) - \frac{\beta}{2} \sum_j \left(\sum_i y_{ij} - \Bar{s}_j \right)^2  \right\} \]
Taking the first order condition of the above objective with respect to $y_{ij}$, we have that:

\[ \Tilde{p}_{ij}^{(k)} + \beta (x_{ij}^{(k+1)}-y_{ij}^{(k+1)}) - p_j^{(k)} - \beta (\sum_{j} y_{ij}^{(k+1)} - \Bar{s}_j) = 0 \]
This implies that:

\[ \Tilde{p}_{ij}^{(k)} + \beta (x_{ij}^{(k+1)}-y_{ij}^{(k+1)}) = p_j^{(k)} + \beta (\sum_{j} y_{ij}^{(k+1)} - \Bar{s}_j) \]
Since the above two terms are equivalent, we have that:

\[ \Tilde{p}_{ij}^{(k+1)} = \Tilde{p}_{ij}^{(k)} + \beta (x_{ij}^{(k+1)} - y_{ij}^{(k+1)}) = p_j^{(k)} + \beta (\sum_{j} y_{ij}^{(k+1)} - \Bar{s}_j) = p_j^{(k+1)} \]
Since $j$ is arbitrary, we have proven our claim that for all agents $i$ we have $\mathbf{p}^{(k+1)} = \mathbf{\Tilde{p}}_i^{(k+1)}$.
\end{proof}

The significance of Proposition~\ref{prop-prices-same} is that each agent distributedly solves her individual optimization problem given the market prices rather than individually observed prices. Furthermore, the equivalence between the individual and market prices implies that the individual price vectors can be completely dropped from Algorithm~\ref{alg:Algo-BP-SOP-AMA}, thereby simplifying it. We elucidate this simplification when presenting the ADMM algorithm in Section~\ref{ADMM-classical-fisher}.

\subsubsection{Convergence Properties of Algorithm~\ref{alg:Algo-BP-SOP-AMA}}

We now turn to characterizing the conditions for the convergence of Algorithm~\ref{alg:Algo-BP-SOP-AMA} to the market equilibrium price. To this end, we require by Theorem~\ref{thm:AMAStandardConvergence} that one of the terms of the objective function $h(\mathbf{x}, \mathbf{y}) = f(\mathbf{x}) + g(\mathbf{y})$ is strictly concave. Note that strict concavity is sufficient for strong concavity provided that the function takes on values in a compact set, which is guaranteed, for instance, when the prices of the goods are strictly positive at each iteration. Since in Problem~\eqref{eq:AMAOptPhys}-\eqref{eq:AMAOptPhys2}, $g(\mathbf{y}) = 0$, we must have that $f(\mathbf{x}) = \sum_i w_i  \log \left(u_i(\mathbf{x}_i) \right)$ is strictly concave. Theorem~\ref{thm:AMAClassicalFisher} establishes that Algorithm~\ref{alg:Algo-BP-SOP-AMA} converges to the equilibrium price vector with rate $O(\frac{1}{k})$ if the homogeneous degree one utility functions $u_i(\mathbf{x}_i)$ are strictly concave. This is formalized through the following result.

\begin{theorem} (Convergence of AMA for \acrshort{bp-sop-adm}) \label{thm:AMAClassicalFisher}
Suppose a market equilibrium for \acrshort{bp-sop-adm} exists. If $u_i(\mathbf{x}_i)$ is any homogeneous degree 1 utility function that is strictly concave then Algorithm~\ref{alg:Algo-BP-SOP-AMA} converges to the equilibrium price vector of \acrshort{bp-sop-adm} with rate $O(\frac{1}{k})$ if $\beta < \frac{2 \sigma_{f}}{\rho(B^T B)}$ and the price vector is strictly positive at each iteration. Here $k$ is the number of iterations of Algorithm~\ref{alg:Algo-BP-SOP-AMA}, $\sigma_f$ is the strong concavity parameter for the function $f(\mathbf{x}) = \sum_i w_i  \log \left(u_i(\mathbf{x}_i) \right)$ and $B$ is the matrix corresponding to the coefficients of $\mathbf{x}$ in the Constraints~\eqref{eq:AMAOptPhys1}-\eqref{eq:AMAOptPhys2}.
\end{theorem}
\begin{proof}
To prove the above theorem, it suffices for us to ensure that we satisfy the assumptions of Theorem~\ref{thm:AMAStandardConvergence}, since the algorithm that we have applied is identical to Algorithm~\ref{alg:AlgoStandardAMA}. We first observe that there is no restriction on $\mathbf{y}_i$ for all agents $i$ and so $\mathcal{Y} = \mathbb{R}^{nm}$ is a closed convex subset of $\mathbb{R}^{n_2}$, where $n_2 = nm$. Next, we have that $\mathcal{X} = \{(\mathbf{x}_1, ..., \mathbf{x}_n): (\mathbf{x}_1, ..., \mathbf{x}_n): \mathbf{x}_i \in \mathbb{R}_{\geq 0}^m, A_t^{(i)} \mathbf{x}_i \leq 0, \forall t \in T_i\}$ is also a closed convex subset of $\mathbb{R}^{n_1}$, where $n_1 = nm$.

Next, we observe that $g(\mathbf{y}_1, ..., \mathbf{y}_n) = 0$ and so it is certainly a concave function. Furthermore, $f(\mathbf{x}_1, ..., \mathbf{x}_n) = \sum_i w_i \log(u_i(\mathbf{x}_i))$ is a strictly concave function, since $u_i(\mathbf{x}_i)$ is strictly concave homogeneous degree 1 utility function.

Since, we have satisfied all of the requisite assumptions of Theorem~\ref{thm:AMAStandardConvergence}, we have that Algorithm~\ref{alg:Algo-BP-SOP-AMA} converges to the equilibrium price vector with a rate $O(\frac{1}{k})$ for any strictly utility function that is homogeneous of degree 1 provided $\beta < \frac{2 \sigma_{f}}{\rho(B^T B)}$.
\end{proof}
This result implies a convergence to the equilibrium price vector with agents distributedly solving their individual optimization problems as long as the utility functions are strictly concave. However, the step size $\beta$ depends on the strong concavity parameter of $f$, which depends on the utility function of agents that a market designer may not have complete knowledge of \emph{a priori}. We resolve this drawback of AMA with the ADMM algorithm, which we present in Section~\ref{ADMM-classical-fisher}.

This convergence result also extends to classical Fisher markets without additional linear constraints. In the classical Fisher market setup the set $\mathcal{X}$ is now $\{(\mathbf{x}_1, ..., \mathbf{x}_n): \mathbf{x}_i \in \mathbb{R}_{\geq 0}^m \}$, which is also a closed and convex set. Note that at each iteration if the prices are strictly positive then due to the budget constraints of buyers the set of values $\mathbf{x}_i$ for each agent $i$ at each iteration lies in a bounded set. Corollary~\ref{cor:AMAClassicalFisher} shows that Algorithm~\ref{alg:Algo-BP-SOP-AMA} converges to the equilibrium price vector for classical Fisher markets.

\begin{corollary} (Convergence of AMA for Classical Fisher Markets) \label{cor:AMAClassicalFisher}
If $u_i(\mathbf{x}_i)$ is a strictly concave homogeneous degree one utility function then Algorithm~\ref{alg:Algo-BP-SOP-AMA} converges to the equilibrium price vector for classical Fisher markets with rate $O(\frac{1}{k})$ if $\beta < \frac{2 \sigma_{f}}{\rho(B^T B)}$ and the price vector is strictly positive at each iteration. Here $k$ is the number of iterations of Algorithm~\ref{alg:Algo-BP-SOP-AMA}, $\mathcal{X}_i = \{\mathbf{x}_i : \mathbf{x}_i \in \mathbb{R}_{\geq 0}^m \}$, $\sigma_f$ is the strong concavity parameter for the function $f(\mathbf{x}) = \sum_i w_i  \log \left(u_i(\mathbf{x}_i) \right)$ and $B$ is the matrix corresponding to the coefficients of $\mathbf{x}$ in the Constraints~\eqref{eq:AMAOptPhys1}-\eqref{eq:AMAOptPhys2}.
\end{corollary}

This result contributes to the literature on classical Fisher markets since it provides a distributed methodology to compute market equilibria. However, as with Theorem~\ref{thm:AMAClassicalFisher} the step size of the price updates depends on the knowledge of agent's utilities.

\subsubsection{Market Implementation Details of Algorithm~\ref{alg:Algo-BP-SOP-AMA}}

We now summarize the methodology to implement Algorithm~\ref{alg:Algo-BP-SOP-AMA} to derive equilibrium prices in Algorithm~\ref{alg:Market-implementation-bp-sop-ama}.
\begin{algorithm} 
\label{alg:Market-implementation-bp-sop-ama}
\SetAlgoLined
\SetKwInOut{Input}{Input}\SetKwInOut{Output}{Output}
\Input{Price vector \textbf{p} in the market and Baseline demand $\mathbf{y}$}
Repeat Until Convergence to Equilibrium Price:
\begin{enumerate}
    \item Agents distributedly solve their \acrshort{iop} based on the market price \textbf{p} 
    \item Market designer updates baseline demand $\mathbf{y}$ based on the observed demands $\mathbf{x}$
    \item Prices are updated in the market asynchronously using tatonnement with a fixed step size $\beta>0$
\end{enumerate}
\caption{Market Implementation of Algorithm~\ref{alg:Algo-BP-SOP-AMA}}
\end{algorithm}

Note that Algorithm~\ref{alg:Market-implementation-bp-sop-ama} requires no information of agent's utilities in step (1) since agents solve their own \acrshort{iop}s; however, information on agent's utilities is necessary to set the right step size $\beta$ for the price updates in step (3) to guarantee convergence of Algorithm~\ref{alg:Algo-BP-SOP-AMA}.

\subsection{ADMM for Fisher Markets with Homogeneous Linear Constraints} \label{ADMM-classical-fisher}

Algorithm~\ref{alg:Algo-BP-SOP-AMA}, like many tatonnement based algorithms for classical Fisher markets, only converges to the equilibrium price for strictly concave homogeneous degree one utility functions and requires the step size of the price updates to depend on the utility functions of agents. The primary reason for the non-convergence of AMA for general concave homogeneous degree one utility functions is that the objective in the individual updates of the $\mathbf{x}$'s at each iteration is not strictly concave. As a result, this suggests that to extend convergence results to the setting of general concave homogeneous degree one utility functions, the individual $\mathbf{x}$ update steps need to be made strictly concave through utility perturbations. In this section, we turn to a new class of distributed ADM algorithms, the Alternating Direction Method of Multipliers (ADMM), that guarantee convergence for all concave homogeneous degree one utility functions, including linear utilities, through a perturbed individual optimization problem. In the ADMM algorithm, prices are updated with a step size that does not depend on agent's utilities.

We begin by presenting the ADMM algorithm for \acrshort{bp-sop-adm}. Algorithm~\ref{alg:AlgoFisherADMM} presents the distributed implementation to solve \acrshort{bp-sop-adm} using the ADMM algorithm elucidated in Algorithm~\ref{alg:AlgoStandardADMM}.

\begin{algorithm} 
\label{alg:AlgoFisherADMM}
\SetAlgoLined
\SetKwInOut{Input}{Input}\SetKwInOut{Output}{Output}
\Input{Initial price vector $\mathbf{p}$, and initial baseline demand $\mathbf{y}_i^{(0)}$}
\Output{Equilibrium Price vector $\mathbf{p}^*$}
 \For{$k = 0, 1, 2, ...$}{
  $\mathbf{x}_i^{(k+1)} = \argmax_{\mathbf{x}_i \in \mathcal{X}_i} \small\{w_i \log \left(u_i(\mathbf{x}_i) \right) - \sum_j p_j^{(k+1)} x_{ij} - \frac{\beta}{2} \sum_{i, j} (x_{ij}-y_{ij}^{(k+1)})^2 \small\}$, for all $i$ \;
  $\mathbf{y}^{(k+1)} = \argmax_{\mathbf{y}} \small\{ - \frac{\beta}{2} \sum_{i, j} (x_{ij}^{(k+1)}-y_{ij})^2 - \frac{\beta}{2} \sum_j \left(\sum_i y_{ij} - \Bar{s}_j \right)^2  \small\}$ \;
  $\mathbf{p}^{(k+1)} \leftarrow \mathbf{p}^{(k)} + \beta (\sum_{i} \mathbf{y}_i^{(k+1)} - \Bar{\mathbf{s}})$, where $\Bar{\mathbf{s}} = (\Bar{s}_1, ..., \Bar{s}_m)$ \;}
\caption{Two Block ADMM for \acrshort{bp-sop-adm}}
\end{algorithm}

Note that we have used Proposition~\ref{prop-prices-same} to simplify Algorithm~\ref{alg:AlgoFisherADMM}, since the individual prices $\mathbf{\Tilde{p}}_i$ can be initialized to the market prices $\mathbf{p}$ for all agents $i$ and thus there is no longer a need to update the individual price vectors. We further observe that each agent now maximizes a regularized individual utility function, which ensures the strict concavity of objective in the $\mathbf{x}$ updates.

\subsubsection{Convergence Properties of Algorithm~\ref{alg:AlgoFisherADMM}}

We now characterize the convergence of Algorithm~\ref{alg:AlgoFisherADMM} for both \acrshort{bp-sop-adm} as well as for classical Fisher markets. The main result of this section states that Algorithm~\ref{alg:AlgoFisherADMM} converges to the equilibrium price vector of \acrshort{bp-sop-adm} for any concave homogeneous degree one utility function. 

\begin{theorem} (Convergence of ADMM for \acrshort{bp-sop-adm}) \label{thm:ADMMClassicalFisher}
Suppose that a market equilibrium for \acrshort{bp-sop-adm} exists and $u_i(\mathbf{x}_i)$ is any concave homogeneous degree one utility function. Then Algorithm~\ref{alg:AlgoFisherADMM} converges to the equilibrium price vector of \acrshort{bp-sop-adm} with rate $O(\frac{1}{k})$, where $k$ is the number of iterations of Algorithm~\ref{alg:AlgoFisherADMM}.
\end{theorem}

The proof of the above claim follows directly from Theorem~\ref{thm:ADMMStandardConvergence} and so Algorithm~\ref{alg:AlgoFisherADMM} converges for any $\beta > 0$, e.g., $\beta = 1$, which is independent of agent's utilities. The argument of this proof is near identical to that for the proof of Theorem~\ref{thm:AMAClassicalFisher} and so is omitted. As with the AMA algorithm, we can extend this convergence result to the classical Fisher market framework, which we formalize through the following corollary.

\begin{corollary} (Convergence of ADMM for Classical Fisher Markets) \label{cor:ADMMClassicalFisher}
If $u_i(\mathbf{x}_i)$ is any homogeneous degree 1 utility function then Algorithm~\ref{alg:AlgoFisherADMM} converges to the equilibrium price vector for classical Fisher markets with rate $O(\frac{1}{k})$, where $k$ is the number of iterations of Algorithm~\ref{alg:AlgoFisherADMM}, and $\mathcal{X}_i = \{\mathbf{x}_i : \mathbf{x}_i \in \mathbb{R}_{\geq 0}^m \}$.
\end{corollary}

We further note in the specific case of linear utilities that Corollary~\ref{cor:ADMMClassicalFisher} implies a convergence guarantee of $O(\frac{1}{k})$. This result overcomes a theoretical barrier for distributed algorithms for linear Fisher markets that do not rely on information on agent's utilities to compute equilibrium prices. To the best of our knowledge, existing convergence rates of distributed algorithms provide guarantees for Fisher markets with strictly concave utilities; however these results do not generalize to Fisher markets with linear utilities. For instance, Cole and Fleischer \cite{Cole2008FastconvergingTA} derive a tatonnement algorithm with fast convergence rates for markets with weak gross substitutes and concede that their results do not generalize to linear utilities. Furthermore, we note that existing tatonnement algorithms either rely on price updates wherein the step size depends on the specific utility functions of agents or the convergence of the algorithm is dependent on the specific form of the utility functions. However, Algorithm~\ref{alg:AlgoFisherADMM} provides a convergence guarantee when the step size is set to a constant parameter $\beta > 0$ that does not depend on the utility functions of agents. 

We close this section with another immediate corollary of Theorem~\ref{thm:ADMMClassicalFisher}, which states that in the limit when Algorithm~\ref{alg:AlgoFisherADMM} converges to the equilibrium price vector of \acrshort{bp-sop-adm} each agent will be maximizing their individual optimization problems without the utility perturbation term.

\begin{corollary} \label{ADMM-Objective-convergence} (Equivalence of \acrshort{iop} and Individual Objectives at Convergence)
When Algorithm~\ref{alg:AlgoFisherADMM} converges to the equilibrium price vector of \acrshort{bp-sop-adm}, then the individual optimization problem in the $\mathbf{x}$ updates for each agent in Algorithm~\ref{alg:AlgoFisherADMM} has the same mathematical structure as the \acrshort{iop}.
\end{corollary}

\begin{proof}
At convergence to the market equilibrium, we know that $x_{ij}^* = y_{ij}^*$ for all $i, j$, since the regularization terms all go to 0 in the ADMM algorithm. Thus, the individual utility update step reduces to:

\[\mathbf{x}_i^* = \argmax_{\mathbf{x}_i: \mathcal{X}_i} \small\{w_i \log \left(u_i(\mathbf{x}_i) \right) - \sum_j p_j^* x_{ij}  \small\} \text{for all } i \]
where $\mathcal{X}_i =  \{\mathbf{x}_i: \mathbf{x}_i \in \mathbb{R}_{\geq 0}^{m} \text{ and } A_t^{(i)} \mathbf{x}_i \leq 0, \forall t \in T_i \}$ and $\mathbf{p}^*$ is the equilibrium price in the market and so the market clears at this price. It is not hard to see that the first order necessary and sufficient KKT conditions (at the market clearing quantities) of the above problem are equivalent to that of the \acrshort{iop} for homogeneous linear constraints with a price vector $\mathbf{p}^*$. Thus, when Algorithm~\ref{alg:AlgoFisherADMM} converges to the equilibrium price vector of \acrshort{bp-sop-adm}, then the individual optimization problem in the $\mathbf{x}$ updates for each agent in Algorithm~\ref{alg:AlgoFisherADMM} has the same mathematical structure as the \acrshort{iop}.
\end{proof}

This result suggests that the regularization terms in the individual update step for each agent goes to zero when Algorithm~\ref{alg:AlgoFisherADMM} converges to the equilibrium price vector. This is a desirable outcome since agents solve their \acrshort{iop}s rather than a perturbed objective when observing equilibrium prices.

\subsubsection{Market Implementation Details of Algorithm~\ref{alg:AlgoFisherADMM}}

Algorithm~\ref{alg:AlgoFisherADMM} can be used to derive equilibrium prices in the market in an analogous manner to the market implementation of Algorithm~\ref{alg:Algo-BP-SOP-AMA}, with the modification that agents now distributedly solve a regularized utility function rather than their \acrshort{iop}. To incentivize agents to optimize for the regularized utility function rather than their true individual objective \acrshort{iop}, the market designer must post rebates in addition to the prices for each of the goods. In particular, the rebates must be designed such that the agent is compensated for the loss in their objective function value, i.e., the difference in their objective on optimizing the perturbed and unperturbed utility functions. The rebates will enable the agents to gain back ``lost'' utility on having to maximize the perturbed utility function. As a result, it can be ensured that it is in each agent's best interests to optimize the perturbed utility function as in Algorithm~\ref{alg:AlgoFisherADMM}. Furthermore, the utility perturbation term can be interpreted as a penalty or additional ``cost'' to the agent $i$ for deviating from a baseline demand $\mathbf{y}_i$ since each agent's individual optimization problem is equivalent to that of the individual ADMM objective without the squared regularization term.

The ADMM algorithm, unlike the AMA mechanism, requires no information of agent's utilities to derive the right step size $\beta$ for the price updates to guarantee convergence of Algorithm~\ref{alg:AlgoFisherADMM} to the equilibrium price. Thus, as long as agents are sophisticated enough to solve the regularized individual objective function, the ADMM algorithm can be used to derive market equilibrium prices when agents have any concave homogeneous degree 1 utility functions, including linear utilities.

\subsection{Numerical Results to Verify Theoretical Guarantees} \label{numerical-conv-admm}

We now verify the convergence results from Corollaries~\ref{cor:AMAClassicalFisher} and~\ref{cor:ADMMClassicalFisher} for Algorithms~\ref{alg:Algo-BP-SOP-AMA} and~\ref{alg:AlgoFisherADMM} respectively through numerical experiments. In this section, we present computational results of the two algorithms for classical Fisher markets with ten agents and ten goods, with a capacity of one for each of the goods. Furthermore, both the budgets and utilities of agents are generated from the Uniform $[0, 1]$ distribution. In the case of linear utilities we observe that the ADMM algorithm converges with a step size of $\beta = 1$, while the AMA mechanism does not converge with a step size as small as $\beta = 0.1$, as is depicted in Figure~\ref{ADMMLinearConvergence}. We then tested our algorithm for Cobb Douglas utilities, which is a strictly concave utility function, and observed convergence for both the ADMM and the AMA algorithms, thereby corroborating the results of Corollaries~\ref{cor:AMAClassicalFisher} and~\ref{cor:ADMMClassicalFisher}. The numerical convergence for Cobb Douglas utilities is shown in Figure~\ref{ADMMCobbDouglasConvergence}.

\begin{figure}[!h]
  \centering
  \begin{minipage}[b]{0.48\textwidth}
  
    \includegraphics[width=\textwidth]{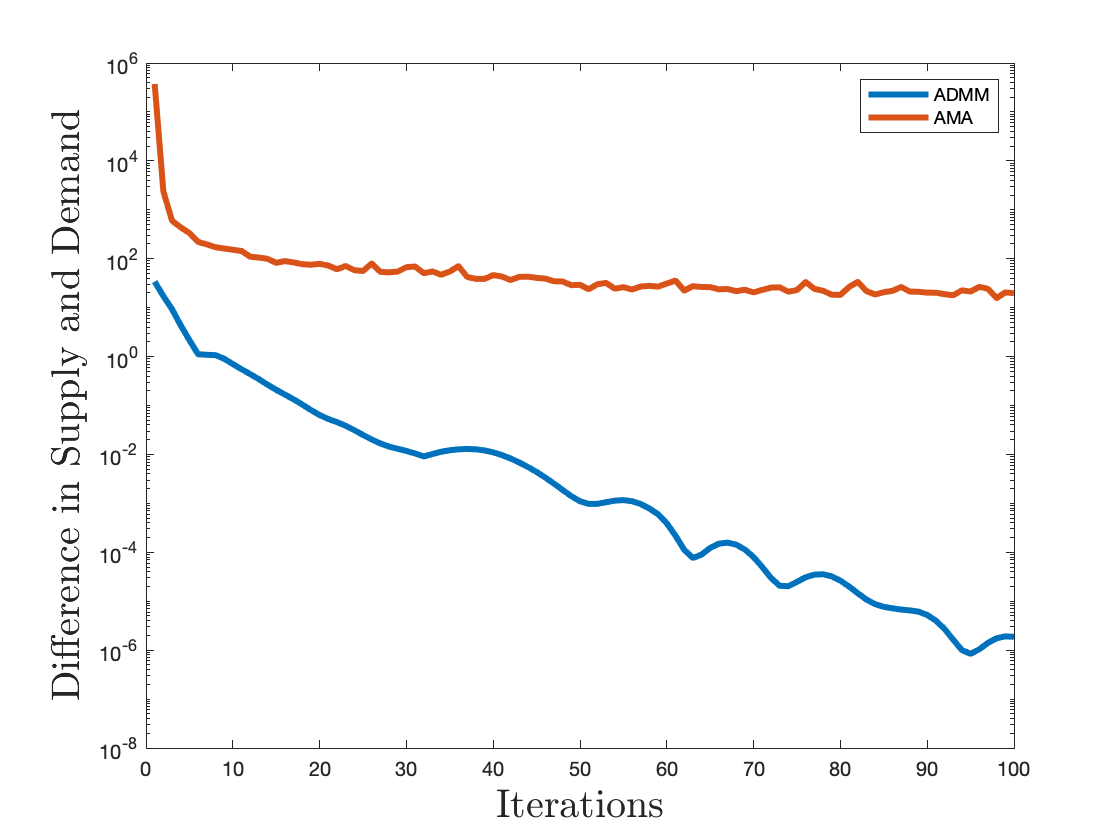}
    \caption{A comparison of the convergence of ADMM and AMA algorithms for Linear utility functions in classical Fisher markets.}
    \label{ADMMLinearConvergence}
  \end{minipage}
  \hfill
  \begin{minipage}[b]{0.48\textwidth}
    \includegraphics[width=\textwidth]{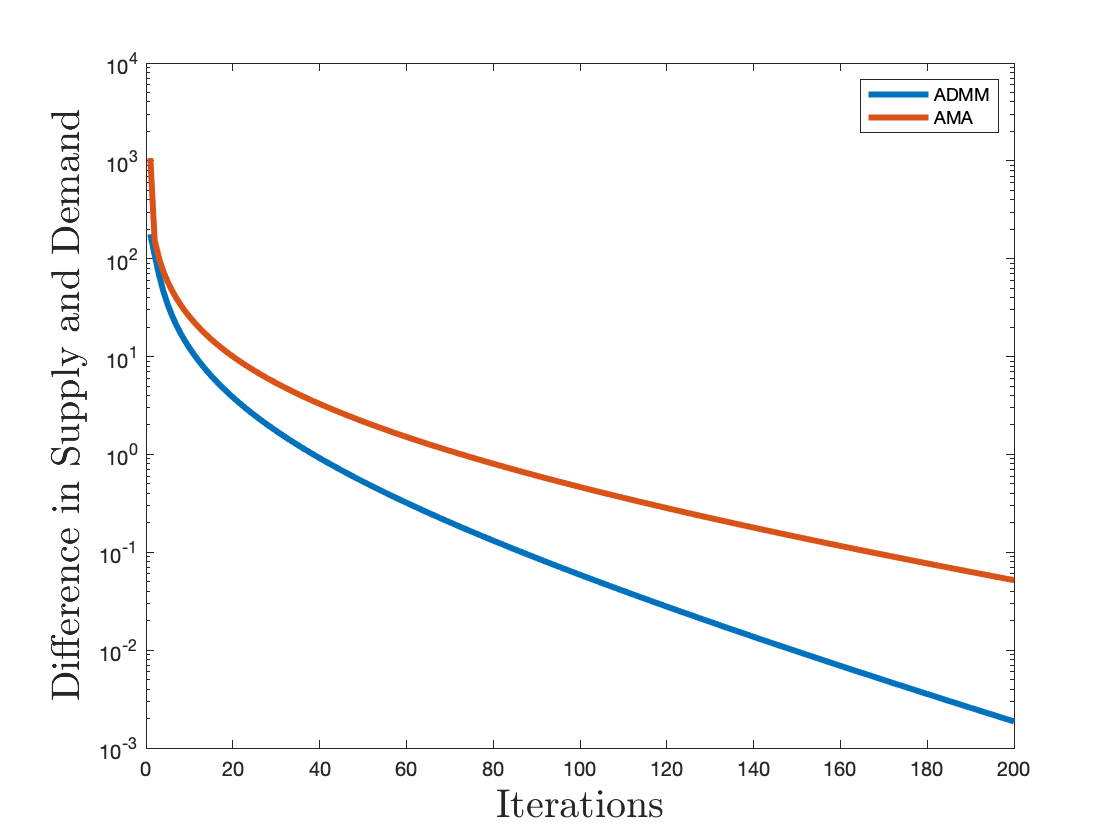}
    \caption{A comparison of the convergence of the ADMM and AMA algorithms for Cobb Douglas utility functions in classical Fisher markets.}
    \label{ADMMCobbDouglasConvergence}
  \end{minipage}
\end{figure}

\subsection{ADMM for Fisher Markets with Non-Homogeneous Linear Constraints} \label{ADMM-physical-constraints}

We now develop a distributed ADMM algorithm for Fisher markets with non-homogeneous linear constraints. To this end, we follow a similar procedure to that considered in Section~\ref{AMA-convergence} by defining the following reformulation of \acrshort{bp-sop}, denoted as \acrshort{bp-sop-admm}:

\begin{maxi!}|s|[2]                   
    {\mathbf{x}_i \in \mathbb{R}_{\geq 0}^{m}, \mathbf{y}_i \in \mathbb{R}^m, \forall i \in [n]}                               
    {\sum_i (w_i + \lambda_i) \log \left(\sum_j u_{ij} x_{ij}\right) \label{eq:ADMMOptPhys}}   
    {\label{eq:ADMMPhysExample1}}             
    {}                                
    \addConstraint{\sum_i y_{ij}}{ = \Bar{s}_j, \forall j \in [m] \label{eq:ADMMOptPhys1}}    
    \addConstraint{\mathbf{y}_i}{ = \mathbf{x}_i, \forall i \in [n] \label{eq:ADMMOptPhys2}}
    \addConstraint{\max \left\{ A_t^{(i)} \mathbf{x}_i - b_{it}, 0 \right\}}{= 0, \forall t \in T_i, \forall i \in [n] \label{eq:ADMMOptPhys3}}
\end{maxi!}

The above formulation is different from \acrshort{bp-sop-adm} in that we have now explicitly kept the additional linear constraints rather than dissolving these into $\mathcal{X}_i$ for each agent $i$. Doing so is necessary since the parameter $\lambda_i$ for each agent $i$ is not known \emph{a priori} as it depends on the dual variables of the additional linear constraints. We note that the dual variables $r_{it}$ of the additional linear constraints will enter the augmented Lagrangian of \acrshort{bp-sop-admm}, which can be used to set $\lambda_i = \sum_{t \in T_i} r_{it} b_{it}$, as is required for a market equilibrium in Theorem~\ref{thm:thm2}. Finally, since the ADMM procedure relies on the equality of constraints, we modified the linear inequality constraints to $\max \left\{ A_t^{(i)} \mathbf{x}_i - b_{it}, 0 \right\} = 0$ $\forall t \in T_i, \forall i \in [n]$, which is now piecewise linear. Note that with the nature of the added constraints the problem \acrshort{bp-sop-admm} will yield the same solution as the problem \acrshort{bp-sop}. Further, note that the objective is independent of $\mathbf{y}_i$, which implies that the objective is $h(\mathbf{x}, \mathbf{y}) = f(\mathbf{x}) + g(\mathbf{y})$, where $\mathbf{x} = (\mathbf{x}_1, ..., \mathbf{x}_n)$, $\mathbf{y} = (\mathbf{y}_1, ..., \mathbf{y}_n)$, $f(\mathbf{x}) = \sum_i (w_i + \lambda_i) \log \left(\sum_j u_{ij} x_{ij}\right)$ and $g(\mathbf{y}) = 0$. Algorithm~\ref{alg:Algo-BP-SOP-ADMM} presents the distributed implementation to solve \acrshort{bp-sop-admm} using the ADMM algorithm elucidated in Algorithm~\ref{alg:AlgoStandardADMM}.

\begin{algorithm} 
\label{alg:Algo-BP-SOP-ADMM}
\SetAlgoLined
\SetKwInOut{Input}{Input}\SetKwInOut{Output}{Output}
\Input{Initial price vector $\mathbf{p}^{(0)}$, and initial baseline demand $\mathbf{y}_i^{(0)}$}
\Output{Equilibrium Price vector $\mathbf{p}^*$}
 \For{$k = 0, 1, 2, ...$}{
  {$\begin{aligned} \mathbf{x}_i^{(k+1)} &= \argmax_{\mathbf{x}_i: \mathbf{x}_i \geq 0} \left\{\left(w_i + \sum_t r_{it}^{(k)}b_{it}\right) \log \left(\sum_j u_{ij} x_{ij} \right) - \sum_j p_{ij}^{(k)} x_{ij} - \sum_t r_{it}^{(k)} \max \left\{ A_t^{(i)} \mathbf{x}_i - b_{it}, 0 \right\} \right. \\ &{} \left.  - \frac{\beta}{2} \sum_{i, j} (x_{ij}-y_{ij}^{(k)})^2 - \frac{\beta}{2} \left( \sum_i \left(\max \left\{ A_t^{(i)} \mathbf{x}_i - b_{it}, 0 \right\}\right)^2 \right) \right\}  \end{aligned}$}, for all $i$ \;
  $\mathbf{y}^{(k+1)} = \argmax_{\mathbf{y}} \small\{ - \frac{\beta}{2} \sum_{i, j} (x_{ij}^{(k+1)}-y_{ij})^2 - \frac{\beta}{2} \sum_j \left(\sum_i y_{ij} - \Bar{s}_j \right)^2  \small\}$ \;
  $\mathbf{p}^{(k+1)} \leftarrow \mathbf{p}^{(k)} + \beta (\sum_{i} \mathbf{y}_i^{(k+1)} - \Bar{\mathbf{s}})$, where $\Bar{\mathbf{s}} = (\Bar{s}_1, ..., \Bar{s}_m)$ \;
  $r_{it}^{(k+1)} \leftarrow r_{it}^{(k)} + \beta \left(\max \left\{ A_t^{(i)} \mathbf{x}_i - b_{it}, 0 \right\} \right)$
  }
\caption{Two Block ADMM for \acrshort{bp-sop-admm}}
\end{algorithm}

A few comments about Algorithm~\ref{alg:Algo-BP-SOP-ADMM} are in order. First, by Proposition~\ref{prop-prices-same} we have simplified the algorithm by initializing the individual prices $\mathbf{\Tilde{p}}_i$ to the market prices $\mathbf{p}$ for all agents $i$. Thus, individual price vectors, i.e., the dual variables for constraint~\eqref{eq:ADMMOptPhys2}, no longer need to be updated. Next, at each step we have that each agent solves a regularized utility function, where $\lambda_i = \sum_{t \in T_i} r_{it}b_{it}$ as was required in Theorem~\ref{thm:thm2} for a market equilibrium. Furthermore, the distributed implementation implies that a central planner does not need any information on the utilities of agents provided that agents are sophisticated enough to compute the optimal solution of the objective in the $\mathbf{x}_i$ update steps. 

We now turn to studying the convergence of Algorithm~\ref{alg:Algo-BP-SOP-ADMM}. The primary difficulty in establishing theoretical convergence guarantees for non-homogeneous linear constraints is that at each step not only do the dual variables change but also the original function that we are trying to maximize changes. In particular, the function $f(\mathbf{x}) = \sum_i (w_i + \lambda_i) \log \left(\sum_j u_{ij} x_{ij}\right)$ changes with each iteration, since it is required that $\lambda_i = \sum_{t \in T_i} r_{it} b_{it}$ for a market equilibrium. As a result, we cannot necessarily extend the convergence guarantees for ADMM to Fisher markets with non-homogeneous linear constraints, as the function $f$ changes at each step, since the value of $\lambda$ is updated. Thus, we validate the convergence of Algorithm~\ref{alg:Algo-BP-SOP-ADMM} through numerical experiments and defer the theoretical treatment of the convergence of Algorithm~\ref{alg:Algo-BP-SOP-ADMM} as an open question for future research.

Figure~\ref{ADMM-Physical-Constraints} depicts the convergence to an equilibrium price in a market with 10 buyers, 20 goods and two identical knapsack linear constraints for each buyer (with 10 goods corresponding to each of the two knapsack constraints). The utilities and budgets were both generated from the Uniform $[0, 1]$ distribution and the capacities of each good were set to 0.5. It can be observed that both the additional linear constraints and the capacity constraints are met after about 60 iterations. Furthermore, we note that the degree of accuracy with which the constraints are met is far greater than that with the fixed point iterative scheme presented in Section~\ref{experiments} indicating the robustness of this ADMM procedure.
\begin{figure}[!h]
      \centering
      \includegraphics[width=0.8\linewidth]{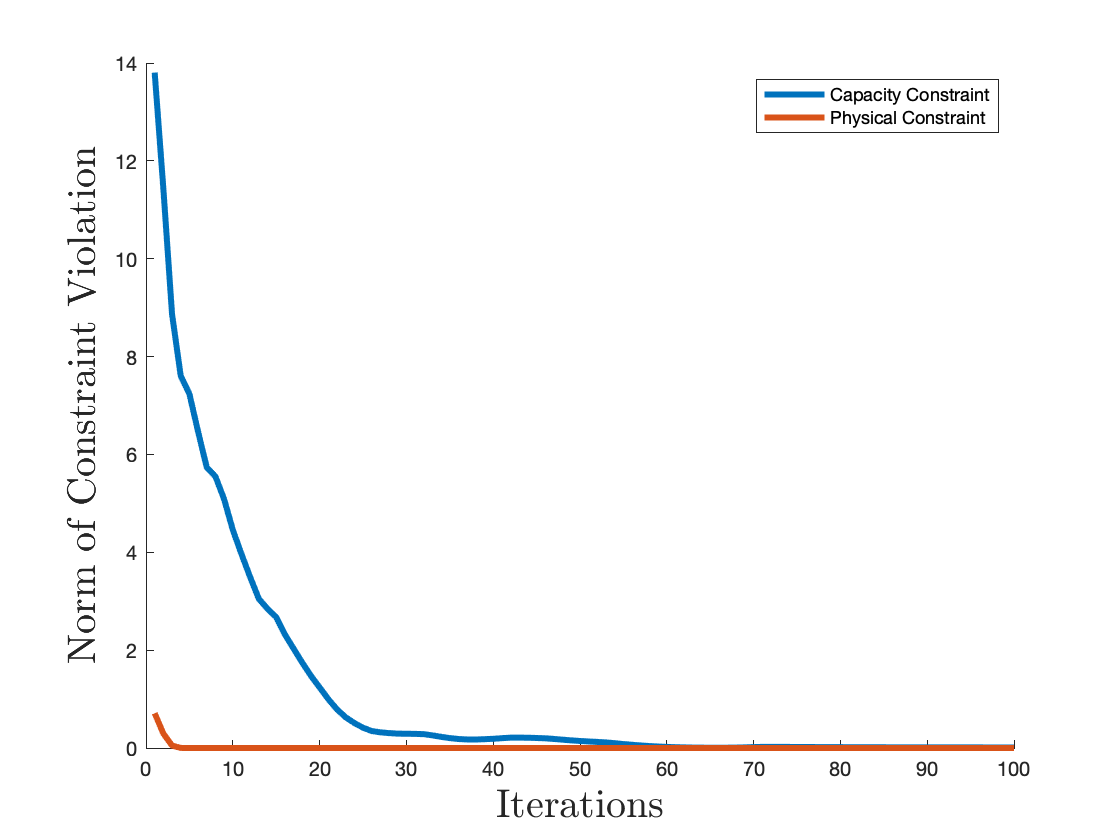}
      \caption{Convergence of Algorithm~\ref{alg:Algo-BP-SOP-ADMM} for Fisher market with additional Non-Homogeneous Linear constraints for a market with 10 buyers, 20 goods and two identical knapsack linear constraints for each buyer.}
      \label{ADMM-Physical-Constraints}
   \end{figure}

\section{Conclusions and Future Work} \label{Conclusion}

In this work, we have developed a generalization of Fisher markets to the setting of additional linear constraints and proposed novel distributed algorithms with provable convergence guarantees for the modified Fisher market with homogeneous linear constraints. To this end, we defined a new individual optimization problem \acrshort{iop} in the presence of additional linear constraints and established conditions to guarantee the existence of a market equilibrium. Even though the properties of \acrshort{iop} are fundamentally different from that of the individual optimization problem in Fisher markets, we proposed a mechanism to derive a market equilibrium in the presence of additional linear constraints, thereby generalizing the Fisher market framework. In particular, we formulated a new social optimization problem \acrshort{bp-sop} wherein we perturbed the budgets of agents to establish an equilibrium price in the market. To obtain the right budget perturbation parameters, we used a fixed point iterative scheme for the reformulated social optimization problem and numerically validated the convergence of this iterative scheme. Since this fixed point iterative scheme involved solving the centralized problem \acrshort{bp-sop}, we proposed a new class of distributed algorithms based on Alternating Direction Methods. In particular, we developed an AMA algorithm in which agents solve their \acrshort{iop} at each iteration, and developed an ADMM algorithm that does not require any knowledge of agent's utilities and updates prices with a utility independent step size. Finally, we proved the convergence of these distributed algorithms for Fisher markets with homogeneous linear constraints and extended these convergence results to the classical Fisher market setting. Since the ADMM algorithm converges to the equilibrium price vector for classical Fisher markets with linear utilities, it overcomes a theoretical barrier in distributed equilibrium computation for classical Fisher markets.

There are various interesting directions of future research that warrant more study. First, the allocations provided by the proposed market mechanisms are fractional, and as the market designer may like to make discrete allocations, it would be interesting to investigate the loss in social efficiency under integral constraints. Furthermore, while we have numerically shown convergence guarantees of the fixed point iterative scheme and the distributed ADMM algorithm for Fisher markets with non-homogeneous linear constraints, it would be beneficial to theoretically understand the convergence of these algorithms. The result on the non-convexity of the price equilibrium set suggests that strong convergence results may not be possible. Thus, we believe that a stronger characterization of the computational complexity of this problem would provide a more nuanced appreciation of whether computing the fixed point or the equilibrium prices with non-homogeneous linear constraints is feasible in polynomial time. Finally, an interesting area of research is generalizing this framework to an online setting in which customers arrive in the market platform sequentially and an irrevocable decision needs to be made about the prices in the market while still achieving a socially efficient allocation.

\section*{Acknowledgments} 

\noindent This work is supported by the Institute of Computational and Mathematical Engineering (ICME) Graduate Fellowship at Stanford University and the Human Centered Artificial Intelligence (HAI) Grant at Stanford University.

\bibliography{main}

\printglossary[type=\acronymtype]

\appendix
\section{Appendix}

\subsection{Definitions}

In this section, we introduce some definitions of terms commonly used throughout the paper.

\begin{definition} (Homogeneous Degree $\alpha$ Function) \label{def:homogeneous-degree-alpha}
For any scalar $\alpha$, a real valued function $u(\mathbf{x}),$ where $\mathbf{x} \in \mathbb{R}^m$, is homogeneous of degree $\alpha$ if
$$
u(t \mathbf{x})=t^{\alpha} f(\mathbf{x})
$$
for all $t>0$
\end{definition}

\begin{definition} (CES utility Function) \label{def:CES-utilities}
A utility function is a Constant Elasticity of Substitution utility if $u(\mathbf{x}) = \left(\sum_{j=1}^{m} \alpha_{j} x_{j}^{\rho}\right)^{1 / \rho}$, where $\mathbf{x} \in \mathbb{R}^m$ and $\rho \in (-\infty, 1]$ is some constant.
\end{definition}
Note that when $\rho=1$ then we have a linear utility function.


\begin{definition} (Homogeneous Linear Constraints) \label{def:homogeneous-linear-constraints}
A linear constraint of \acrshort{iop} of the form $A_t^{(i)} \mathbf{x}_i \leq b_{it}$ is homogeneous if $b_{it} = 0$. The problem \acrshort{iop} has homogeneous linear constraints if $b_{it} = 0$ for all additional linear constraints $t \in T_i$. If for any $t \in T_i$, $b_{it}>0$ then the additional linear constraints of agent $i$ are non-homogeneous.
\end{definition}

\begin{definition} (Strictly Increasing Vector Valued Function) \label{def:strict-increasing-function}
A vector valued function $u: \mathbb{R}^n \mapsto \mathbb{R}$ is strictly increasing if for any two bundles $\mathbf{x} \in \mathbb{R}^n$ and $\mathbf{y} \in \mathbb{R}^n$, with $\mathbf{x} \geq \mathbf{y}$ and $\mathbf{x} \neq \mathbf{y}$, it follows that $u(\mathbf{x}) \geq u(\mathbf{y})$. 
\end{definition}

\subsection{Proof of Proposition 1} \label{non-existence-prop1}

We now restate and proceed to prove Proposition~\ref{prop:non-existence-eq}.

\nonexistenceeq*

\begin{proof}
\textbf{Homogeneous Linear Constraints}: We consider a counter example with linear proportionality constraints. Consider a market with two buyers and two goods, with capacities $\Bar{s}_1 = 1$ and $\Bar{s}_2 = 1$, and utilities and budgets as specified in Table~\ref{tab:table-non-existence-homogeneous}. 

\begin{table}[tbh] 
\centering
\begin{tabular}{|c|c|c|c|} 
\hline
                 & \textbf{Utility for Good 1} & \textbf{Utility for Good 2} & \textbf{Budget} \\ \hline
\textbf{Buyer 1} & 1                         & 1                         & 1              \\ \hline
\textbf{Buyer 2} & 1                         & 1                        & 1               \\ \hline
\end{tabular} 
\caption{Utilities and budgets of buyers in a two buyer, two good market to establish the non-existence of a market equilibrium when agents have homogeneous linear constraints.}
\label{tab:table-non-existence-homogeneous}
\end{table}
Suppose that buyer 1 has the proportionality constraint $x_{11} - x_{12} \leq 0$ and buyer 2 has the proportionality constraint $2x_{21} - x_{22} \leq 0$. We further have that the capacity constraints can be written as $x_{11} + x_{21} \leq 1$ and $x_{12} + x_{22} \leq 1$.

For the market to clear it must be that $x_{11} + x_{21} = 1$ and $x_{12} + x_{22} = 1$ must both hold, i.e., both goods are completely sold out. We now show that if $x_{12} + x_{22} = 1$ holds then $x_{11} + x_{21} = 1$ cannot hold. So suppose that $x_{12} + x_{22} = 1$, then we have by summing the two proportionality constraints that $x_{11} + x_{21} \leq x_{11} + 2x_{21} \leq x_{12} + x_{22} = 1$. We now consider the following two cases:

\begin{enumerate}
    \item $x_{21} > 0$: In this case, the first inequality is a strict inequality, i.e., $x_{11} + x_{21} < x_{11} + 2x_{21}$, which implies that $x_{11} + x_{21} < 1$ indicating that the market does not clear.
    \item $x_{21} = 0$: In this case, since $x_{11} \leq x_{12}$ by the capacity constraint for good 2 it must be that that $x_{12} = 1$. This in turn implies that $x_{22} = 0$, i.e., agent 2 purchases neither good 1 nor good 2 and so does not spend its budget. Thus, if $x_{21} = 0$ then again the market does not clear.
\end{enumerate}

Thus, we have established that a market equilibrium does not exist when agents have homogeneous linear constraints.

\textbf{Non-homogeneous Linear Constraints}: We consider a counter example with knapsack linear constraints. Consider a market with two buyers and two goods, with capacities $\Bar{s}_1 = 1.5$ and $\Bar{s}_2 = 0.5$, and utilities and budgets as specified in Table~\ref{table-non-existence}. 

\begin{table}[tbh]
\centering
\begin{tabular}{|c|c|c|c|} 
\hline
                 & \textbf{Utility for Good 1} & \textbf{Utility for Good 2} & \textbf{Budget} \\ \hline
\textbf{Buyer 1} & 200                         & 0.1                         & 15              \\ \hline
\textbf{Buyer 2} & 100                         & 1.1                         & 5               \\ \hline
\end{tabular} 
\caption{Utilities and Budgets of buyers in a two buyer, two good market to establish the non-existence of a market equilibrium when agents have non-homogeneous linear constraints.}
\label{table-non-existence}
\end{table}
Suppose that we have the following knapsack constraints $x_{11} + x_{12} \leq 1$ and $x_{21} + x_{22} \leq 1$ for buyers 1 and 2 respectively.

To show that no equilibrium exists in this market, we consider the following two cases:

\begin{enumerate}
    \item $p_1 > 10$: From the knapsack constraint, buyer 1 can buy at most one unit of good 1 and buyer 2 cannot afford 0.5 units of good 1 and so it follows that $\sum_{i = 1}^{2} x_{i1}(p_1) < \Bar{s}_1 = 1.5$, establishing that good 1 cannot be cleared. Thus, $p_1>10$ cannot be a market clearing price.
    \item $p_1 \leq 10$: Since the utilities of both buyers is higher for good 1 than for good 2, it must be that the price for good 1 must be higher than the price for good 2. However, this implies that $p_2 \leq p_1 \leq 10$, which implies that buyer 1 cannot use up their budget, since the buyer can purchase at most one unit of both goods combined. Thus, $p_1 \leq 10$ cannot be a market clearing price.
\end{enumerate}

The two cases above establish that for this example a market equilibrium fails to exist in the presence of additional non-homogeneous linear constraints not considered in classical Fisher markets.
\end{proof}

\subsection{Proof of Proposition~\ref{prop:negative-price-eq}} \label{apdx:negative-price-eq}

We now restate and proceed to prove Proposition~\ref{prop:negative-price-eq}.

\negativeeq*

\begin{proof}

We consider an example market with knapsack linear constraints. Consider a market with two buyers and three goods, with capacities $\Bar{s}_1 = \Bar{s}_2 = \Bar{s}_3 = 1$, and utilities and budgets as specified in Table~\ref{table-negative-eq}. 

\begin{table}[tbh]
\centering
\begin{tabular}{|c|c|c|c|c|} 
\hline
                 & \textbf{Utility for Good 1} & \textbf{Utility for Good 2} & \textbf{Utility for Good 3} & \textbf{Budget} \\ \hline
\textbf{Buyer 1} & 1                         & 2 &      11                    & 10              \\ \hline
\textbf{Buyer 2} & 1                         & 10 & 1                         & 0.5               \\ \hline
\end{tabular}
\caption{Utilities and Budgets of buyers in a two buyer, three good market to establish the negativity of the equilibrium price when agents have knapsack linear constraints.}
\label{table-negative-eq}
\end{table}
Suppose that we have the following knapsack constraints $x_{11} + x_{12} \leq 1$ and $x_{21} + x_{22} \leq 1$ for buyers 1 and 2 respectively. Furthermore, we do not impose any linear constraints on the purchase of good 3 for either agent.

Then consider the price $\mathbf{p} = (-1, 0.5, 11)$. Solving the \acrshort{iop} at this price we have that $x_1^*(\mathbf{p}) = (1, 0, 1)$ and $x_2^*(\mathbf{p}) = (0, 1, 0)$. Note that at this price all goods are sold to capacity and each agent's budget is entirely used up. Thus, $\mathbf{p} = (-1, 0.5, 11)$ is an equilibrium price vector.
\end{proof}

\subsection{Proof of Theorem~\ref{thm:market-eq}} \label{existence-market-eq}

To prove the existence of a market equilibrium, we will show the existence of a fixed point. We do so leveraging Sperner's Lemma, which is a combinatorial analog of fixed point theorems, e.g., Brouwer's fixed point theorem. To elucidate Sperner's Lemma, we first introduce some notation. 

Let $\Delta_m$ be an $m$-dimensional simplex with vertices $v_0, ..., v_m$ and let $\Delta_i$ be the face of the simplex opposite to the vertex $v_i$. Furthermore, suppose that the simplex $\Delta$ is subdivided into smaller simplices forming a simplical complex $\mathcal{S}$. Then Sperner's Lemma can be stated as follows:

\begin{lemma} (\cite{ivanov2019sperners} Sperner's Lemma)
Suppose that the vertices of $\mathcal{S}$ are labelled by the numbers $\{0, 1, 2, ..., m \}$ in such a way that if a vertex $v$ belongs to a face $\Delta_i$ then the label of $v$ is not equal to $i$. Then there exists an $m$-dimensional simplex of $\mathcal{S}$ such that the labels of the vertices are equal to $\{0, 1, 2, ..., m \}$.
\end{lemma}

Leveraging Sperner's Lemma, we now restate and proceed to prove Theorem~\ref{thm:market-eq}.

\marketeq*

\begin{proof}
To present the proof we re-scale the constraints and the objective function such that the resulting problem consists of a market where the capacity of each good is one and the total budget of all agents is one. First, we scale agent $i$'s allocation $x_{ij}$ by a factor $\frac{1}{\Bar{s}_j}$, and her utility $u_{ij}$ by a factor $\Bar{s}_j$. Next, we scale the coefficients $A_{tj}^{(i)}$ for constraint $t$ by $\Bar{s}_j$. Finally, letting $\sum_{i = 1}^n w_i = W$, we scale price $p_j$ by a factor $\frac{\Bar{s}_j}{W}$, and the budgets by a factor $\frac{1}{W}$. It is easy to see that under the new budgets and prices, we obtain the same optimal allocations for the \acrshort{iop}.

Thus, we focus our attention to a market where capacities of each good are normalized to $1$, and total budget of all agents is normalized to 1, i.e., $\sum_{i = 1}^{n} w_i = 1$.  Further, by condition (i) in the statement of the theorem it is clear that if there exists an equilibrium then the price vector must be strictly positive. Thus, we restrict our focus to price vectors lying in the standard simplex.


We now define an excess demand function $f_j(\mathbf{p}^*) = \sum_{i = 1}^{n} x_{ij}(p_j^*) - 1$ and use Sperner's lemma to prove the existence of an equilibrium price vector $\mathbf{p}^* = (p_1^*, p_2^*, ..., p_m^*) \in \Delta_m$, where $\Delta_m$ is the standard simplex, such that for all goods $j$, the excess demand function $f_j(\mathbf{p}^*) = \sum_{i = 1}^{n} x_{ij}(p_j^*) - 1 = 0$.

For any $\mathbf{p} \in \Delta_m$, we define a coloring function $c: \mathbf{p} \mapsto \small\{0, 1, ..., m\small\}$, such that $c(\mathbf{p}) = j$ for a $j$ that satisfies $f_j(\mathbf{p}) \leq 0$ and $p_j \neq 0$. We label a vertex of the simplex with $0$ if the above conditions are not satisfied for all $j \in [m]$, e.g., when $\mathbf{p} = 0$. We claim that such a coloring satisfies Sperner's lemma.

To see this, we first note that all the $m$ non-zero corner points of the simplex must be colored with different colors, as at each corner point $k$ exactly one entry $p_k = 1 \neq 0$. In fact when $p_k = 1$, as the budgets are normalized to $1$, it must be that $f_k(\mathbf{p}) \leq 0$, as at most one unit of good $k$ can be purchased by all agents collectively. Thus, we have that all the corner points of the simplex must be colored with a different color. For any other bounding vertex there must exist $j$, such that $p_j \neq 0$ and $f_j(\mathbf{p}) \leq 0$. This is because if for all $j$, such that $p_j \neq 0$, we have that $f_j(\mathbf{p})>0$ then $\sum_{i = 1}^{n} \sum_{j = 1}^{m} p_j x_{ij} > \sum_{j = 1}^{m} p_j = 1 = \sum_{i = 1}^{n} w_i$, a contradiction. Thus, we have a valid coloring, which satisfies the condition of Sperner's lemma.

Next, by Sperner's lemma we have that there must exist a base simplex such that each of its corner points has distinct colors. Taking finer and finer triangulations and using that the demand function is continuous, we can find a $\mathbf{p}^*$, such that $f_j(\mathbf{p}^*) \leq 0$, $\forall j$, as each subsequent triangulation still lies on the standard simplex. Thus, we have shown $\forall j$ that $\sum_{i = 1}^{n} x_{ij}(p_j^*) \leq 1$.

Now to prove that the above inequality is in fact an equality, we proceed by contradiction. In particular, suppose that $\exists j$ with $p_j^*>0$, such that $\sum_{i = 1}^{n} x_{ij}(p_j^*) < 1$. Then clearly there must exist some buyer $i$ whose budget was not used up since $\sum_{i = 1}^{n} \sum_{j = 1}^{m} p_j^* x_{ij} < 1 = \sum_{i = 1}^{n} w_i$. However, by the condition that for all agents $i$ there exists a good $j$ without additional linear constraints, $i$ can buy more units of good $j$ (and gain a strictly positive utility), giving us our desired contradiction. Thus, $\mathbf{p}^*$ is an equilibrium price vector.
\end{proof}

\subsection{Proof of Theorem~\ref{thm:nonconvexity}} \label{non-convex-appendix}

We now restate and proceed to prove Theorem~\ref{thm:nonconvexity}.

\nonconvexity*

\begin{proof}
We consider a market with 4 agents and 4 products, with one knapsack constraint for goods 1-3 (and so we have a constraint $x_{i1}+x_{i2}+x_{i3} \leq 1$ for all $i$). Furthermore, since good 4 is not tied to any additional linear constraints, agents can purchase any amount of good 4 that is affordable. Next we consider the utility and budget values for each of the agents as well as the price and capacities for each of the goods as specified in Table~\ref{table-non-convexity}.

\begin{table}[tbh]
\centering
\begin{tabular}{|c|c|c|c|c|c}
\hline
Utility                                  & \textbf{Good 1} & \textbf{Good 2} & \textbf{Good 3} & \textbf{Good 4} & \multicolumn{1}{c|}{\textbf{Budget}} \\ \hline
\textbf{Buyer 1} & 2               & 0.0001          & 4               & 0.0001          & \multicolumn{1}{c|}{2}               \\ \hline
\textbf{Buyer 2} & 1               & 2               & 0.0001          & 0.0001          & \multicolumn{1}{c|}{1.5}             \\ \hline
\textbf{Buyer 3} & 0.0001          & 3               & 4               & 0.0001          & \multicolumn{1}{c|}{2.5}             \\ \hline
\textbf{Buyer 4} & 0.0001          & 0.0001          & 0.0001          & 1               & \multicolumn{1}{c|}{1}               \\ \hline
\textbf{Supply}                          & 1                                       & 1                                       & 1                                       & 1                                       &                                      \\ \cline{1-5}
\textbf{Price}                           & $p_1 = 2+\frac{2 \eta -1}{12 \eta^2 + 1}$                                       & $p_2 = 2-\frac{4 \eta}{12 \eta^2 + 1}$                                        & $p_3 = 2+\frac{2 \eta + 1}{12 \eta^2 + 1}$                                        & $p_4 = 1$                                       &                                      \\ \cline{1-5}
\end{tabular} 
\caption{Utilities and budgets of buyers as well as prices and capacities of goods in a four buyer, four good market to establish that the equilibrium price set may be non-convex.}
\label{table-non-convexity}
\end{table}

We will now show that for all $\eta \in [-\frac{1}{24}, 0]$ that $(p_1, p_2, p_3, p_4)$ is an equilibrium price set to establish non-uniqueness and then establish this set is non-convex.

We start by claiming that the optimal allocation vectors for all agents for some $\eta \in [-\frac{1}{24}, 0]$ given those prices is:

\[ \mathbf{x}_1^* = (0.5 + \eta, 0, 0.5-\eta, 0) \]
\[ \mathbf{x}_2^* = (0.5 - \eta, 0.5+\eta, 0, 0) \]
\[ \mathbf{x}_3^* = (0, 0.5-\eta, 0.5+\eta, 0) \]
\[ \mathbf{x}_4^* = (0, 0, 0, 1) \]

Before proving that these are in fact the optimal allocations, we first note that the market clears given these allocations, as $\sum_i x_{ij} = 1$ for all goods $j$ and that these allocations respect the knapsack linear constraints. Furthermore, it is easy to see that the budgets for all agents are completely used up. For instance, for the first agent, the money spent is: 

\[ (0.5+\eta) \cdot\left(2+\frac{2 \eta-1}{12 \eta^{2}+1}\right)+(0.5-\eta) \cdot\left(2+\frac{2 \eta+1}{12 y^{2}+1}\right)=2= w_1 \]
Analogous calculations for the remaining agents will establish that their budgets are also completely used up.

Now, to show that the above allocations are in fact optimal, we will use the characterization of the optimal solution of the \acrshort{iop} established in Theorem~\ref{thm:iopoptimal}. We prove this for agent 1 and analogous calculations will establish that the above allocation vectors are optimal for the remaining agents.

We note that for the given range of prices agent 1 will never purchase goods 2 and 4, as the utility derived given the range of prices is too low. Next for the given range of prices, we observe that $\frac{u_{11}}{p_1} > \frac{u_{13}}{p_3}$. As a result agent 1 will first purchase the \textit{virtual product} corresponding to the segment $\{(0, 0), (u_{11}, p_1)\}$ and then purchase the virtual product corresponding to segment $\{(u_{11}, p_1), (u_{13}, p_3)\}$. Thus, the agent will first buy the first virtual product, which costs $p_1$, resulting in a total amount of money remaining equal to $2 - \left(2+\frac{2 \eta -1}{12 \eta^2 + 1} \right) = \frac{1-2\eta}{12 \eta^2 + 1}$. Now this amount of remaining money will be spent on the next virtual product, whose price is $p_3-p_1 = \frac{2}{12 \eta^2 + 1}$. Thus, the fraction of the second virtual product purchased by agent 1 is $\frac{\frac{1-2\eta}{12 \eta^2 + 1}}{\frac{2}{12 \eta^2 + 1}} = 0.5 - \eta$. Thus, we get that $x_{13}^* = 0.5 - \eta$ and $x_{11}^* = 0.5 + \eta$ and so we have proven that the optimal consumption vector for agent 1 is $\mathbf{x}_1^* = (0.5 + \eta, 0, 0.5-\eta, 0)$.

Since for the prices $(p_1, p_2, p_3, p_4) = \left(2+\frac{2 \eta -1}{12 \eta^2 + 1}, 2-\frac{4 \eta}{12 \eta^2 + 1}, 2+\frac{2 \eta + 1}{12 \eta^2 + 1}, 1 \right)$, we have that the market clears and each agent's budget is completely used at their optimal consumption vectors, we have that the set of prices $(p_1, p_2, p_3, p_4)$ are the equilibrium prices in the market when $\eta \in [-\frac{1}{24}, 0]$. This establishes the non-uniqueness of the price vector in the market. In particular, the above result implies that for $\eta = 0$, we have $\mathbf{p}^{(1)} = \left(1, 2, 3, 1 \right)$ and for $\eta = \frac{-1}{24}$, we have $\mathbf{p}^{(2)} = \left(\frac{46}{49}, \frac{106}{49}, \frac{142}{49}, 1 \right)$, which are both equilibrium price vectors. 

Next, we prove the non-convexity of the equilibrium price set by showing that $\mathbf{p}^{(3)} = \frac{\mathbf{p}^{(1)} + \mathbf{p}^{(2)}}{2} = \left(\frac{95}{98}, \frac{204}{98}, \frac{289}{98}, 1 \right)$ is not an equilibrium price vector.

To see this, we note that the only agents that will buy good 1 at a price $\mathbf{p}^{(3)}$ are agents 1 and 2. In particular, agent 1 will only purchase goods 1 and 3 and since the agent purchased a positive amount of both goods at prices $\mathbf{p}^{(1)}$ and $\mathbf{p}^{(2)}$ this agent will continue to do so at the price $\mathbf{p}^{(3)}$. This can easily be seen from the optimal characterization of the \acrshort{iop} given in Theorem~\ref{thm:iopoptimal}. As a result, it must be that the knapsack linear constraint is tight, i.e., $x_{11} + x_{13} = 1$ at the price $\mathbf{p}^{(3)}$. Furthermore, since the agent will consume their whole budget to increase the utility level, we have that $x_{11} \cdot \frac{95}{98}+x_{13} \cdot \frac{289}{98}=w_{1}=2$. Solving this system of equations gives $x_{11} = \frac{93}{194}$ and $x_{13} = \frac{101}{194}$.

Similarly, if we look at agent 2, who only purchases goods 1 and 2, we obtain the following two equations: (i) $x_{21}+x_{22} = 1$ and (ii) $x_{21} \frac{95}{98}+x_{22} \frac{204}{98}=w_{2}=1.5$. Solving this we get $x_{21} = \frac{57}{109}$ and $x_{22} = \frac{52}{109}$.

Since agents 1 and 2 are the only ones that purchase good 1, we now obtain that $x_{11} + x_{21} = \frac{93}{194} + \frac{57}{109} > 1$. As a result, we have that the demand is greater than the capacity of good 1 at price $\mathbf{p}^{(3)}$ and so this cannot be an equilibrium price vector. Thus, we have established that the equilibrium price set is non-convex.
\end{proof}

\subsection{Proof of Theorem~\ref{thm:iopoptimal}} \label{iop-optimal-pf}

We now restate and proceed to prove Theorem~\ref{thm:iopoptimal}.

\iopopt*

\begin{proof}
From Definition~\ref{def-sol-set}, we know for all constraints $t \in T_i$ for each $i$ that $u_t = \sum_{j \in t} u_{ij} x_{ij}^*, w_t = \sum_{j \in t} x_{ij}^* p_j$.

Since $\mathbf{x}_i^*$ is an optimal solution, it must be that $(\mathbf{x}_i^*)_{j \in t}$, i.e., the restriction of $\mathbf{x}_i^*$ to the goods in linear constraint $t$, is an optimal solution to the following problem:

\begin{maxi!}|s|[2]                   
    {\left\{x_{ij} \right\}_{j \in t}}                               
    {\sum_{j \in t} u_{ij} x_{ij}  \label{eq:neweq1}}   
    {\label{eq:newExample001}}             
    {}                                
    \addConstraint{\sum_{j \in t} x_{ij} p_j}{\leq w_t \label{eq:newcon1}}    
    \addConstraint{\sum_{j \in t} x_{ij}}{\leq 1, \label{eq:newcon2}}
    \addConstraint{x_{ij}}{\geq 0, \forall j \in t  \label{eq:newcon3}}  
\end{maxi!}

The optimal solution $\left\{x_{ij}^* \right\}_{j \in t}$ corresponds to some point $C$ (as depicted in Figure \ref{virtual-products}) in the lower frontier of the convex hull $S_t$ from $(0, 0)$ to $(u_{ij_{\max}}, p_{j_{\max}})$, where $j_{\max} = \argmax_{j \in t} u_{ij}$. Thus, such a solution $\left\{x_{ij}^* \right\}_{j \in t}$ can be viewed as purchasing \textit{virtual products} given by the different slopes $\theta$ in the descending order of their \textit{bang-per-buck}, which is equivalent to the ascending order of the slopes $\theta$. In particular, one needs to purchase \textit{virtual products} in the descending order from the origin to the point $C$ (as depicted in Figure \ref{virtual-products}).

Having established that agent $i$ purchases goods in the descending order of their \textit{bang-per-buck} within each knapsack constraint $t \in T_i$, we now show that agent $i$ in fact purchases goods in the descending order of their \textit{bang-per-buck} amongst all goods belonging to different knapsack linear constraints. Thus, consider two different knapsack constraints, $t$, $t'$ and suppose for contradiction that agent $i$ does not purchase goods in the descending order of the \textit{bang-per-buck} of the \textit{virtual products} corresponding to these constraints. Thus, it must be that $\exists$ a \textit{virtual product} with slope $\theta_a$ corresponding to constraint $t$ and a \textit{virtual product} with slope $\theta_b$ corresponding to constraint $t'$, such that $\theta_a <\theta_b$, where some amount of the \textit{virtual product} with slope $\theta_b$ is bought but the \textit{virtual product} with slope $\theta_a$ is not completely purchased, i.e., strictly less than one unit of it is bought.

Without loss of generality, assume that $\theta_a$ is defined by two goods $j_1, j_2 \in t$ (with $u_{ij_1} < u_{ij_2}$) and $\theta_b$ is defined by two goods $j_3, j_4 \in t'$ (with $u_{ij_3} < u_{ij_4}$), as depicted in Figure~\ref{thm2-proof-opt-fig}. Next, suppose that we have an optimal solution for these four goods given by $(x_{ij_1}^*, x_{ij_2}^*, x_{ij_3}^*, x_{ij_4}^*)$.

\begin{figure}[!h]
      \centering
      \includegraphics[width=0.8\linewidth]{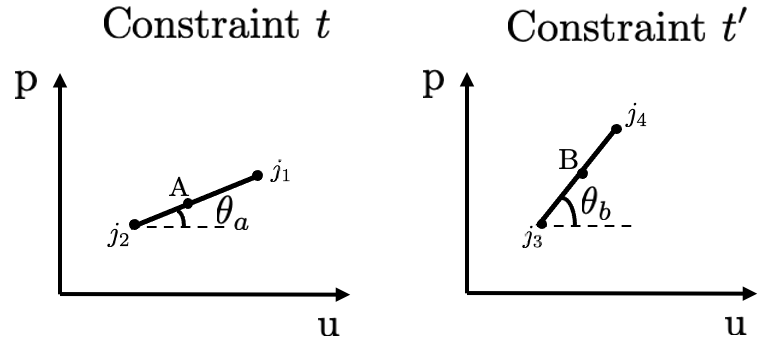}
      \caption{Goods $j_1$ and $j_2$ belong to constraint $t \in T_i$ and goods $j_3$ and $j_4$ belong to constraint $t' \in T_i$. The amount of \textit{virtual product} consumed lies somewhere along the line segment denoted by $A$ and $B$ for constraints $t$ and $t'$ respectively. We further have that $\theta_a<\theta_b$.}
      \label{thm2-proof-opt-fig}
   \end{figure}

Next, let $x_{ij_4}^{'} = x_{ij_4}^* - \epsilon$ and $x_{ij_3}^{'} = x_{ij_3}^* + \epsilon$ for some small $\epsilon>0$. We now observe that since $(x_{ij_3}^*, x_{ij_4}^*)$ is feasible with respect to the linear constraint, so is $(x_{ij_3}^{'}, x_{ij_4}^{'})$. Next, we observe that as we shift from the allocation $(x_{ij_3}^{*}, x_{ij_4}^{*})$ to $(x_{ij_3}^{'}, x_{ij_4}^{'})$, we ``save'' money given by $\epsilon (p_{j_4} - p_{j_3})$, but also reduce utility by $\epsilon (u_{ij_4} - u_{ij_3})$. But, now, we can use this ``saved'' money to purchase $\theta_a$.

So consider, $(x_{ij_1}^{'}, x_{ij_2}^{'})$, such that $x_{ij_2}^{'} = x_{ij_2}^* - \tau$ and $x_{ij_1}^{'} = x_{ij_1}^* + \tau$, where $\tau$ satisfies $\tau (p_{j_2} - p_{j_1}) = \epsilon (p_{j_4} - p_{j_3})$, i.e., we use our unspent or ``saved'' money. As a result we gain a utility given by $\tau (u_{ij_2} - u_{ij_1})$. Then we have that the change in the utility of agent $i$ is given by:

\[u_i^{'} - u_i^* = \tau (u_{ij_2} - u_{ij_1}) - \epsilon (u_{ij_4} - u_{ij_3}) = \epsilon (u_{ij_4} - u_{ij_3}) \cdot \left( \frac{\tau}{\epsilon} \cdot \frac{u_{ij_2} - u_{ij_1}}{u_{ij_4} - u_{ij_3}} - 1 \right) \]
Here, we have taken $u_i^{'}$ as the utility corresponding to the allocation $(x_{ij_1}^{'}, x_{ij_2}^{'}, x_{ij_3}^{'}, x_{ij_4}^{'})$ and $u_i^*$ as the utility corresponding to the allocation $(x_{ij_1}^{*}, x_{ij_2}^{*}, x_{ij_3}^{*}, x_{ij_4}^{*})$. Next, since we have that $\tau (p_{j_2} - p_{j_1}) = \epsilon (p_{j_4} - p_{j_3})$, it follows that:

\[u_i^{'} - u_i^* = \epsilon (u_{ij_4} - u_{ij_3}) \cdot \left( \frac{p_{j_4} - p_{j_3}}{p_{j_2} - p_{j_1}} \cdot \frac{u_{ij_2} - u_{ij_1}}{u_{ij_4} - u_{ij_3}} - 1 \right) = \epsilon (u_{ij_4} - u_{ij_3}) \cdot \left( \frac{\theta_b}{\theta_a} - 1 \right) >0 \]

This establishes that $(x_{ij_1}^{*}, x_{ij_2}^{*}, x_{ij_3}^{*}, x_{ij_4}^{*})$ cannot be optimal and so we have a contradiction to the claim that agent $i$ does not purchase goods in the descending order of the \textit{bang-per-buck} of the \textit{virtual products} belonging to two constraints $t, t'$. Thus, our claim follows.
\end{proof}

\subsection{Proof of Corollary~\ref{cor:coroptsol}} \label{cor-optimal-pf}

We now restate and proceed to prove Corollary~\ref{cor:coroptsol}.

\coroptsol*

\begin{proof}
An agent $i$ purchases \textit{virtual products} until their budget is completely used up or all the \textit{virtual products} are bought, which implies that it is just the last \textit{virtual product} that the agent purchases, which may be bought for less than one unit. As the remaining \textit{virtual products} are all purchased completely, these correspond to endpoints of the line segments of the lower frontier of the convex hull of the solution set $S_t$. This indicates that agent $i$ either does not buy any good (if we are at the origin in Figure~\ref{virtual-products}) or buys one unit of one good only corresponding to an end point of one of the line segments characterizing the lower frontier of the convex hull. Since the last \textit{virtual product} that the agent purchases may be bought for less than one unit, this corresponds any point on the lower frontier of the convex hull. In particular, if such a point lies in the middle of a line segment on the lower frontier of the convex hull then we have that the agent buys a corresponding fraction of two goods (associated with the line segment) associated with a given constraint $t \in T_i$.
\end{proof}

\subsection{Remark on Optimal Solution of \acrshort{iop}} \label{rmk-iopopt}

The optimal solution of the \acrshort{iop} can also be characterized for goods that do not have any additional linear constraints. For such a product $j$ without linear Constraints~\eqref{eq:con2}, we can define the value $\theta_j = \frac{p_j}{u_{ij}}$ and augment this to the list of the $\theta$ values for each of the \textit{virtual products}. As we ordered the \textit{virtual products} in the descending order of their \textit{bang-per-buck} $ = \frac{1}{\theta_{j_1j_2}}$ for \textit{virtual products} corresponding to goods $j_1$ and $j_2$, we can do the same for this new list of $\theta$ values that includes $\theta_j$.

Now buyer $i$ will purchase goods in the ascending order of the $\theta$ values that includes $\theta_j$ analogous to the result in Theorem~\ref{thm:iopoptimal}. As in Corollary~\ref{cor:coroptsol}, at most one unit of each \textit{virtual product} can be purchased; however, any amount of good $j$ (the good without linear Constraints~\eqref{eq:con2}) can be purchased. Thus, if agent $i$ still has budget at $\theta_j$, then the remaining budget will be used to purchase product $j$.

\subsection{Examples of Optimal Solution of \acrshort{iop}} \label{examples-iopopt}

We provide two examples to illustrate the use of \textit{virtual products} in characterizing the optimal solution of the \acrshort{iop} when agents have knapsack linear constraints. For both examples, we consider the market with $6$ products and an agent $i$. Furthermore, we denote the utilities of the agent for each of the products as $u_{i1} = 1, u_{i2} = 2, u_{i3} = 3, u_{i4} = 4, u_{i5} = 5, u_{i6} = 6$ and the price that the agent observes in the market as $p_1 = 0.1, p_2 = 0.4, p_3 = 0.7, p_4 = 1.2, p_5 = 1.7, p_6 = 2.4$. The difference between the two examples we consider for this market will be in terms of the knapsack linear constraints and budgets. In particular, the first example we consider will be one wherein each of the goods in the market has knapsack constraints and the second example will be one in which one of the goods, i.e., good $5$, will not be tied to any knapsack linear constraint.

\textbf{Example 1} (All goods are associated with knapsack linear constraints): We endow agent $i$ with a budget of $w_i = 2.4$, and consider the following linear constraints for knapsack one $x_{i1} + x_{i3} + x_{i5} \leq 1$, and for knapsack two $x_{i2} + x_{i4} + x_{i6} \leq 1$.

Given the price and the utilities, we can derive the $\theta$ values for the \textit{virtual products} by calculating the slopes of the lower frontier of the convex hull in Figure~\ref{example1-iop-opt}, giving $\theta_1 = 0.1, \theta_2 = 0.2, \theta_3 = 0.3, \theta_4 = 0.4, \theta_5 = 0.5, \theta_6 = 0.6$.

\begin{figure}[!h]
      \centering
      \includegraphics[width=0.9\linewidth]{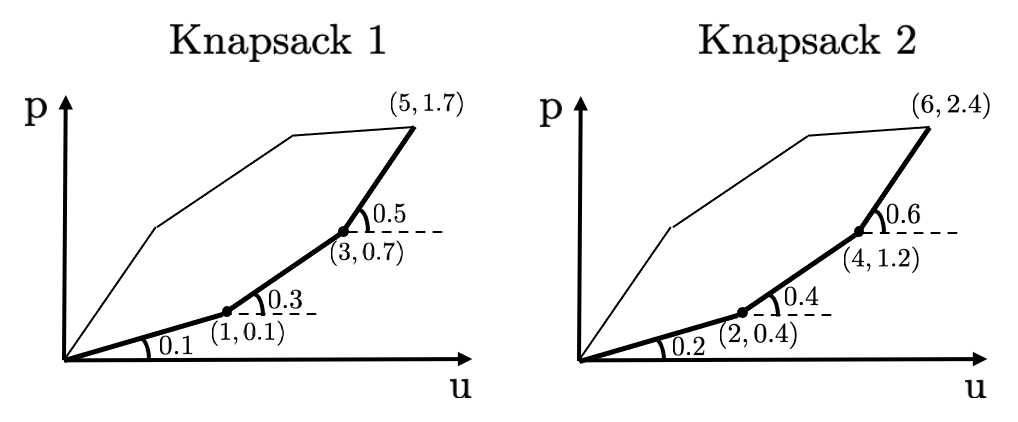}
      \caption{The solution set of the goods in example 1 for both the knapsacks are depicted along with their corresponding $\theta$ values, i.e., the slope of the line segments on the lower frontier of the convex hull (shown in bold).}
      \label{example1-iop-opt}
   \end{figure}

We also note that for a \textit{virtual product} corresponding to points $A = (u_{ij_1}, p_{j_1})$ and $B = (u_{ij_2}, p_{j_2})$, the price is $p_{j_2} - p_{j_1}$ and the utility is $u_{ij_2} - u_{ij_1}$, with slope $\theta_{j_1j_2} = \frac{p_{j_2} - p_{j_1}}{u_{ij_2} - u_{ij_1}}$.

Now, as the agent will buy products in the ascending order of the $\theta$ values until the agent's budget is exhausted by Theorem~\ref{thm:iopoptimal}, we will have that the buyer purchases the following \textit{virtual products}:

\begin{itemize}
    \item 1 unit of $\theta_1$, with a total cost of $0.1$
    \item 1 unit of $\theta_2$, with a total cost of $0.4$
    \item 1 unit of $\theta_3$, with a total cost of $0.6$
    \item 1 unit of $\theta_4$, with a total cost of $0.8$
    \item 0.5 units of $\theta_5$, with a total cost of $0.5$
\end{itemize}

Mapping the purchase of these products to Figure~\ref{example1-iop-opt}, we observe that since \textit{virtual products} with slope $\theta_2$ and $\theta_4$ are bought 1 unit, we have that this corresponds to the point on the lower frontier of the convex hull corresponding to good 4 and so $x_4^* = 1$. Furthermore, since $0.5$ units of the \textit{virtual product} with slope $\theta_5$ is purchased, this maps to a point on the line segment corresponding to good 3 and good 5. Thus, we have that $x_3^* = 0.5$ and $x_5^* = 0.5$. Thus, the final allocations for the agent are $x_1^* = 0, x_2^* = 0, x_3^* = 0.5, x_4^* = 1, x_5^* = 0.5, x_6^* = 0$.

We also note that these allocations confirm the claim of Corollary~\ref{cor:coroptsol}.

\textbf{Example 2} (Good 5 has no knapsack linear constraints): We endow agent $i$ with a budget of $w_i = 4.5$, and consider the following linear constraint for knapsack one $x_{i1} + x_{i3} \leq 1$, and for knapsack two $x_{i2} + x_{i4} + x_{i6} \leq 1$. Good $5$ does not have any linear constraint, i.e., any affordable amount of good $5$ can be purchased.

The calculations of the $\theta$ values of all products other than $\theta_5$ remain the same. We calculate $\theta_5$ as mentioned in the remark in~\ref{rmk-iopopt}, and so $\theta_5 = \frac{1.7}{5} = 0.34$.

Now, as the agent will buy products in the ascending order of the $\theta$ values until the agent's budget is exhausted by Theorem~\ref{thm:iopoptimal}, we will have that the buyer purchases the following \textit{virtual products}:

\begin{itemize}
    \item 1 unit of $\theta_1$, with a total cost of $0.1$
    \item 1 unit of $\theta_2$, with a total cost of $0.4$
    \item 1 unit of $\theta_3$, with a total cost of $0.6$
    \item 2 units of $\theta_5$, with a total cost of $3.4$
\end{itemize}

The buyer purchases the \textit{virtual product} with slope $\theta_5$ before the \textit{virtual product} with slope $\theta_4$, since $\theta_5 = 0.34 < 0.4 = \theta_4$. Further, since good 5 is not associated with any knapsack linear constraint, as mentioned in~\ref{rmk-iopopt}, the buyer can purchase an arbitrary amount of the good 5 until their budget is exhausted.

We observe that since \textit{virtual products} with slope $\theta_1$ and $\theta_3$ are bought 1 unit, we have that this corresponds to the point on the lower frontier of the convex hull (for goods in knapsack 1) corresponding to good 3 and so $x_3^* = 1$. Next, since one unit of the \textit{virtual product} with slope $\theta_2$ is bought and no units of the \textit{virtual product} with slope $\theta_4$ are bought, we have that $x_2^* = 1$. Finally, $2$ units of $\theta_5$ are purchased, as this good does not correspond to any linear constraints, and so we have $x_5^* = 2$.

Thus, the final allocations for the agent are $x_1^* = 0, x_2^* = 1, x_3^* = 1, x_4^* = 0, x_5^* = 2, x_6^* = 0$. We note that this allocation is in line with the remark in~\ref{rmk-iopopt}.

\subsection{Proof of Proposition~\ref{prop:propCompensation}} \label{prop-Compensation}

We now restate and proceed to prove Proposition~\ref{prop:propCompensation}.

\propCompensation*

\begin{proof}
We start by observing that the amount of income compensation for the agent is given by $(p_j' - p_j)x_{ij}$ and that the original bundle of goods is affordable under this income compensation, as $w_i' = \mathbf{p}' \cdot \mathbf{x}_i$ from the statement of the theorem. Since our original bundle of goods satisfied the additional linear constraints, this will be the case at this new income level as well indicating that our original bundle of goods $\mathbf{x}_i$ is in fact feasible at the income $w_i'$.

Having established feasibility of $\mathbf{x}_i$ at the income $w_i'$, we now wish to show that for any optimal bundle it must be the case that $x_{ij}' \leq x_{ij}$. Put differently, we claim that an agent cannot increase their utility from $u(\mathbf{x}_i)$ by purchasing $x_{ij}' > x_{ij}$. To prove this result, it suffices for us to consider the case when $x_{ij}>0$, as if $x_{ij} = 0$, then there is no income compensation and so the agent is solving the \acrshort{iop} with $w_i' = w_i$, which implies $x_{ij}' = x_{ij} = 0$ would correspond to the optimal allocation. 

Thus, we now consider the case when $x_{ij}>0$. Now first consider the case that $j$ is the only good corresponding to its knapsack linear constraint. Then when we increase the price of good $j$ from $p_j$ to $p_j'$, we have reduced the \textit{bang-per-buck} ratio of the \textit{virtual product} corresponding to the points $(0, 0)$ and that corresponding to good $j$. Now, if $x_{ij}<1$, then this is the last \textit{virtual product} to be purchased. If the increase in the \textit{bang-per-buck} ratio of this virtual product maintains the same ordering of the \textit{bang-per-buck} ratios of all \textit{virtual products}, then at the income $w_i'$ the agent will again purchase good $j$ last. In particular, if the agent spent $p_j x_{ij} = w_j$ units of currency on buying $x_{ij}$ units of good $j$ then the agent can now spend $w_{ij} + (p_j' - p_j)x_{ij} = p_j' x_{ij}$ units of currency to buy $x_{ij}$ units of good $j$, i.e., $x_{ij}' = x_{ij} < 1$. If the increase in \textit{bang-per-buck} ratio of this \textit{virtual product} changes the ordering of the \textit{bang-per-buck} ratios of the virtual products, this implies that the good $j$ will now be purchased after some other good is bought. This means that the under the compensated income it must be that $x_{ij}' \leq x_{ij}$, since less income is remaining when purchasing good $j$ in this case.

Next, if $x_{ij} = 1$ units of good $j$ are purchased and the last \textit{virtual product} to be purchased is the one corresponding to knapsack constraint $t$ the argument of the previous paragraph still applies. So, suppose that the last \textit{virtual product} to be purchased is not the one corresponding to knapsack constraint $t$. If the change in \textit{bang-per-buck} ratio of \textit{virtual products} in knapsack constraint $t$ does not influence the ordering of the \textit{bang-per-buck} ratio of all \textit{virtual products}, then it is clear that the optimal action for the agent is to consume the original bundle $\mathbf{x}_i$ by its affordability at the new income level and the argument in the previous paragraph, which implies that $x_{ij}' = x_{ij}$. Furthermore, if the ordering of the \textit{bang-per-buck} ratios of \textit{virtual products} changes then it must be that $x_{ij}' \leq x_{ij}$, where we have equality when the \textit{bang-per-buck} ratios are such that the same goods are bought and inequality when the \textit{bang-per-buck} ratio of knapsack constraint $t$ exceeds that of some \textit{virtual product} that was not originally purchased within the bundle $\mathbf{x}_i$.

Having considered the case when there is just one good in knapsack constraint $t$, we now look at the case when there are at least two goods in knapsack constraint $t$. Since $x_{ij}>0$, this implies that we are either at the point on the price utility plane corresponding to good $j$ or at a point $A$ on the line segment between goods $j$ and $j'$ or at a point $B$ on the line segment between goods $j^{''}$ and $j$ as depicted in Figure~\ref{IncomeCompensation}, where we take good $j$ as belonging to some knapsack constraint $t$.

\begin{figure}[!h]
      \centering
      \includegraphics[width=0.3\linewidth]{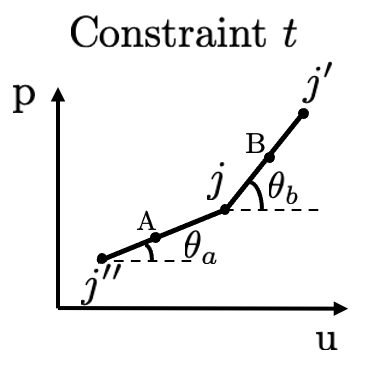}
      \caption{Feasible points for the optimal consumption vector at a price $p_j$ when $x_{ij}>0$ and goods $j, j', j''$ belong to knapsack constraint $t$}
      \label{IncomeCompensation}
   \end{figure}
Note that the above figure represents the most general case, as if $j$ took the place of $j'$ or $j''$ then we would only need to consider one of the two line segments.

When we increase the price of good $j$ the slope $\theta_a$ increases and the slope $\theta_b$ decreases. Now, if we are at a point $B$ on the line segment for the consumption bundle $\mathbf{x}_i$, this means that this is the last \textit{virtual product} to be purchased, i.e., the \textit{virtual product} with slope $\theta_b$ is the \textit{virtual product} with the largest slope $\theta$ that is purchased. But if the increase in $\theta_a$ is such that $\theta_a'>\theta_b'$ when the price of good $j$ is increased from $p_j$ to $p_j'$, then good $j$ will never be purchased, while $\theta_{j^{''}, j'} < \theta_b$ and so goods $j^{''}$ and $j'$ will be purchased within knapsack constraint $t$ and so $x_{ij}' = 0 < x_{ij}$. Instead if $\theta_a' \leq \theta_b'$, then as $\theta_b'<\theta_b$, its \textit{bang-per-buck} ratio ordering could have only increased. Thus, when the price increases from $p_j$ to $p_j'$ and $\theta_b'<\theta_b$ it follows that the point $B$ could have only moved further towards $j'$ as the \textit{bang-per-buck} ratio of the \textit{virtual product} corresponding to goods $j$ and $j'$ is increased. This again establishes that $x_{ij}' \leq x_{ij}$.

On the other hand if we are at a point $A$ on the \textit{virtual product} corresponding to goods $j^{''}$ and $j$ or at the point $j$ then the analysis reduces to the case when $j^{''} = \{(0, 0) \}$, as has been analysed above.

Thus, we have established that $x_{ij}' \leq x_{ij}$ at the new income $w_i'$, indicating that Giffen behavior is eliminated under income compensation.
\end{proof}

\subsection{Proof of Theorem~\ref{thm:imposs}} \label{Appendix1}

We now restate and proceed to prove Theorem~\ref{thm:imposs}.

\imposs*

To prove Theorem~\ref{thm:imposs} we need to compare the KKT conditions of the individual and social optimization problems, where the price vector $\textbf{p}$ that agents observe are those set based on the dual variables of the capacity constraints of \acrshort{sop-l}. As in the proof sketch, we only derive the KKT conditions of the social optimization problem \acrshort{sop-l} and show that the budgets of agents will in general not be completely used up and thus a market clearing equilibrium cannot hold.

\subsubsection{KKT conditions of Social Optimization Problem \acrshort{sop-l}}

We now derive the first order necessary and sufficient KKT conditions for \acrshort{sop-l}. To do so, we introduce the dual variables $\mathbf{p} \in \mathbb{R}^{m}$ for the capacity constraint~\eqref{eq:ImpossSocOpt1}, $\mathbf{r}_{i} \in \mathbb{R}_{\geq 0}^{l_i}, \forall i$ for each of the linear constraints in~\eqref{eq:ImpossSocOpt2} and $s_{ij} \leq 0, \forall i, j$ for each of the non-negativity Constraints~\eqref{eq:ImpossSocOpt3}. Next, we derive the Lagrangian of this problem as:

\begin{align}
    \mathcal{L} = \sum_{i = 1}^{n} w_i \log \left(\sum_{j = 1}^{m} u_{ij} x_{ij} \right) - \sum_{j = 1}^{m} p_j \left(\sum_i x_{ij} - \Bar{s}_j \right) - \sum_{i = 1}^{n} \sum_{t \in T_i} r_{it} \left(\sum_{j = 1}^{m} A_{tj}^{(i)} x_{ij} - b_{it} \right) - \sum_{i = 1}^{n} \sum_{j = 1}^{m} s_{ij} x_{ij} 
\end{align}
The first order derivative condition is found by taking the derivative of the Lagrangian with respect to $x_{ij}$:

\begin{align}\label{eq:ImpossLagDerSoc}
    \frac{w_i}{\sum_{j = 1}^{m} u_{ij} x_{ij}} u_{ij} - p_j - \sum_{t \in T_i} r_{it} A_{tj}^{(i)} \leq 0 
\end{align}
Next, the complimentary slackness condition for this problem can be derived by multiplying~\eqref{eq:ImpossLagDerSoc} by $x_{ij}$:

\begin{align}\label{eq:ImpossCompSlackSoc1}
    \frac{w_i}{\sum_{j = 1}^{m} u_{ij} x_{ij}} u_{ij} x_{ij} - p_j x_{ij} - \sum_{t \in T_i} r_{it} A_{tj}^{(i)} x_{ij} = 0 
\end{align}
Thus, the KKT conditions of the social optimization problem are given by equations~\eqref{eq:ImpossSocOpt1}-\eqref{eq:ImpossSocOpt3},~\eqref{eq:ImpossLagDerSoc},~\eqref{eq:ImpossCompSlackSoc1} and the sign constraints of the dual variables.

\subsubsection{Establishing Theorem~\ref{thm:imposs}} \label{AppendixImposs}

We now use the above derivations of the first order necessary and sufficient KKT conditions to complete the proof of Theorem~\ref{thm:imposs}.

\begin{proof}
We observe that for there to be an equilibrium price vector it must be that the budgets of all agents must be entirely used up. We now show that this will in general not be true when agents have additional linear constraints and prices are derived based on the dual variables of the capacity constraints of \acrshort{sop-l}. To see this, consider the KKT condition in equation~\eqref{eq:ImpossCompSlackSoc1} of the social optimization problem \acrshort{sop-l}. If we sum over all goods $j$, we observe that this equation can be expressed as:

\begin{align}\label{eq:ImpossCompSlackSoc2}
    w_i - \sum_{j = 1}^{m} p_j x_{ij} - \sum_{t \in T_i} r_{it} b_{it} = 0 
\end{align}
However, the above constraint implies that the only way $w_i = \sum_{j = 1}^{m} p_j x_{ij}$ is if $\sum_{t \in T_i} r_{it} b_{it} = 0$, which implies that $r_{it} = 0$ $\forall i, t$ such that $b_{it}>0$, as each $r_{it} \geq 0$. However, this in general cannot be expected, as at the market clearing outcome the linear constraint~\eqref{eq:ImpossSocOpt2} may be met with equality for some or all agents (as was observed in the examples in~\ref{non-existence-prop1} and~\ref{examples-iopopt}). In particular, $r_{it} = 0$ $\forall i, t$ when $b_{it}>0$ implies by sensitivity analysis that if we loosen any of the linear constraints then the objective function value will remain unchanged. However, if there is a difference between the utilities of agents across different goods then for the linear constraint~\eqref{eq:ImpossSocOpt2}, e.g., a knapsack linear constraint, that is loosened at least one person will be better off by purchasing more of the goods belonging to that linear constraint~\eqref{eq:ImpossSocOpt2}. This is because in the case of knapsack linear constraints agents may have higher utilities in purchasing goods in one knapsack (whose constraint is loosened) and and so the objective function will increase in this case. As a result, in general the scenario $r_{it} = 0$ $\forall i, t$ such that $b_{it}>0$ is not possible, particularly in cases when the linear Constraints~\eqref{eq:ImpossSocOpt2} are met with equality, which may be the case for many markets. Thus, we have that the market clearing KKT conditions of the social and individual optimization problems are not necessarily the same, establishing our claim.
\end{proof}

\subsection{Proof of Theorem~\ref{thm:thm2}} \label{proof-thm2}

We now restate and proceed to prove Theorem~\ref{thm:thm2}.

\thmm*

We now derive the KKT conditions for the budget perturbed social optimization problem \acrshort{bp-sop} and individual optimization problem \acrshort{iop} that are used in the derivation of Theorem~\ref{thm:thm2}.

\subsubsection{Derivation of KKT Conditions for Social Optimization Problem \acrshort{bp-sop}} \label{Appendix4}

To establish the equivalence between individual optimization Problem~\eqref{eq:eq1}-\eqref{eq:con3} and the perturbed social optimization problem \acrshort{bp-sop}, we start by deriving the first order necessary and sufficient KKT conditions for the Problem~\eqref{eq:SocOpt}-\eqref{eq:SocOpt3}. To do so, we introduce the dual variables $\mathbf{p} \in \mathbb{R}^{m}$ for the capacity constraint~\eqref{eq:SocOpt1}, $\mathbf{r}_{i} \in \mathbb{R}_{\geq 0}^{l_i}, \forall i$ for each of the linear constraints in~\eqref{eq:SocOpt2} and $s_{ij} \leq 0, \forall i, j$ for each of the non-negativity Constraints~\eqref{eq:SocOpt3}. Next, we derive the Lagrangian of this problem as:

\begin{align}
    \mathcal{L} = \sum_{i = 1}^{n} (w_i + \lambda_i) \log \left(\sum_{j = 1}^{m} u_{ij} x_{ij} \right) - \sum_{j = 1}^{m} p_j \left(\sum_i x_{ij} - \Bar{s}_j \right) - \sum_{i = 1}^{n} \sum_{t \in T_i} r_{it} \left(\sum_{j = 1}^{m} A_{tj}^{(i)} x_{ij} - b_{it} \right) - \sum_{i = 1}^{n} \sum_{j = 1}^{m} s_{ij} x_{ij} 
\end{align}
The first order derivative condition is found by taking the derivative of the Lagrangian with respect to $x_{ij}$:

\begin{align}\label{eq:LagDerSoc}
    \frac{w_i + \lambda_i}{\sum_{j = 1}^{m} u_{ij} x_{ij}} u_{ij} - p_j - \sum_{t \in T_i} r_{it} A_{tj}^{(i)} \leq 0 
\end{align}
Next, the complimentary slackness condition for this problem can be derived by multiplying~\eqref{eq:LagDerSoc} by $x_{ij}$:

\begin{align}\label{eq:CompSlackSoc1}
    \frac{w_i + \lambda_i}{\sum_{j = 1}^{m} u_{ij} x_{ij}} u_{ij} x_{ij} - p_j x_{ij} - \sum_{t \in T_i} r_{it} A_{tj}^{(i)} x_{ij} = 0 
\end{align}
Thus, the KKT conditions of the social optimization problem \acrshort{bp-sop} are given by equations~\eqref{eq:SocOpt1}-\eqref{eq:SocOpt3},~\eqref{eq:LagDerSoc},~\eqref{eq:CompSlackSoc1} and the sign constraints of the dual variables.

\subsubsection{KKT conditions of Individual Optimization Problem \acrshort{iop}} \label{Appendix4_2}

We now derive the KKT conditions of \acrshort{iop} by formulating a Lagrangian and introducing the dual variable $y_i \geq 0$ for~\eqref{eq:con1}, $\mathbf{\Tilde{r}}_{i} \in \mathbb{R}_{\geq 0}^{l_i}$ for the budget constraint in~\eqref{eq:con2} and $\Tilde{s}_{ij} \leq 0, \forall j$ for each of the non-negativity Constraints~\eqref{eq:con3}. Thus, our Lagrangian is:

\begin{align}
    \mathcal{L}(\mathbf{x}_i, y_i, \mathbf{r}_i, \mathbf{s}) = \sum_{j = 1}^{m} u_{ij} x_{ij} - y_i \left (\sum_{j = 1}^{m} p_j x_{ij} - w_i \right) - \sum_{t \in T_i} \Tilde{r}_{it} \left (\sum_{j = 1}^{m} A_{tj}^{(i)} x_{ij} - b_{it} \right) - \sum_{j = 1}^{m} \Tilde{s}_{ij} x_{ij} 
\end{align}
The first order constraint for this problem is found by taking the derivative of our Lagrangian with respect to $x_{ij}$ and noting that $\Tilde{s}_{ij} \leq 0$:

\begin{align}\label{eq:LagDer}
    u_{ij} - y_i p_j - \sum_{t \in T_i} \Tilde{r}_{it} A_{tj}^{(i)} \leq 0 
\end{align}
Next, we derive the complimentary slackness condition for this problem by multiplying~\eqref{eq:LagDer} by $x_{ij}$:

\begin{align}\label{eq:CompSlack1}
    u_{ij}x_{ij} - y_i p_j x_{ij} - \sum_{t \in T_i} \Tilde{r}_{it} A_{tj}^{(i)} x_{ij} = 0 
\end{align}
Thus, the final KKT conditions for our problem are given by equations~\eqref{eq:con1},~\eqref{eq:con2},~\eqref{eq:con3},~\eqref{eq:LagDer} and~\eqref{eq:CompSlack1}, and the sign constraints on the dual variables. Furthermore, at equilibrium conditions, it must hold that the goods must be sold to capacity, which is given by $\sum_{i = 1}^{n} x_{ij} = \Bar{s}_j, \forall j$.

We use the derived KKT conditions of the individual and social optimization problems to prove our theorem in the analysis that follows.

\begin{proof}
($\Rightarrow$) To prove the forward direction of our claim, we need to show that given a market equilibrium with price vector $\mathbf{p}$ and optimal allocation $\mathbf{x}_i$ of the \acrshort{iop} for each $i$, we can construct $\lambda_i$, such that $\lambda_i = \sum_{t \in T_i} r_{it} b_{it}$, for each agent $i$.

We proceed by letting $y_i>0$, where $y_i$ is the dual variable of the budget constraint in the \acrshort{iop}. Note that $y_i = 0$ is not possible since there is a good $j$ that agent $i$ can purchase any amount of where $u_{ij}>0$. To see this, if we perform a sensitivity analysis by relaxing the budget constraint then the agent will always have a strictly higher utility since the agent $i$ can consume more of good $j$, and so $y_i$ must be greater than 0.

Next, we note that at the equilibrium condition $\sum_{i = 1}^{n} x_{ij} = \Bar{s}_j, \forall j \in [m]$ that the constraints of the individual optimization problem \acrshort{iop} already implies the other constraints of \acrshort{bp-sop}, i.e., the Constraints~\eqref{eq:SocOpt2} and~\eqref{eq:SocOpt3}. Thus, all we need to do is check the Lagrangian derivative condition~\eqref{eq:LagDerSoc}, complementary slackness condition~\eqref{eq:CompSlackSoc1} and dual multiplier sign constraints.

When we sum over all goods $j$, we have from the complimentary slackness condition~\eqref{eq:CompSlack1} that:

\begin{align} \label{eq:comp_slack_iopsop}
    \sum_{j = 1}^{m} u_{ij}x_{ij} = y_i \sum_{j = 1}^{m} p_j x_{ij} + \sum_{j = 1}^{m} \sum_{t \in T_i} \Tilde{r}_{it} A_{tj}^{(i)} x_{ij} = y_i w_i + \sum_{t \in T_i} \Tilde{r}_{it} b_{it}
\end{align}
The second equality follows from the complimentary slackness condition for the budget and additional linear constraints.

Next, using the Lagrangian derivative condition~\eqref{eq:LagDer}, and taking $r_{it} = \frac{\Tilde{r}_{it}}{y_i}$ we have that:

\begin{align} \label{eq:new-ineq-iopsop}
    \frac{1}{y_i} \leq \frac{p_j + \sum_{t \in T_i} r_{it} A_{tj}^{(i)}} {u_{ij}}, \quad \forall j
\end{align}

Now setting $\lambda_i = \sum_{t \in T_i} r_{it} b_{it} = \frac{\sum_{t \in T_i} \Tilde{r}_{it} b_{it}}{y_i}$, we obtain that:

\begin{align}
    \frac{w_i + \lambda_i}{\sum_{j = 1}^{m} u_{ij}x_{ij}} = \frac{1}{y_i} \cdot \frac{ \left(y_i w_i + \sum_{t \in T_i} \Tilde{r}_{it} b_{it} \right)}{\sum_{j = 1}^{m} u_{ij}x_{ij}} 
\end{align}
Now observing Equation~\eqref{eq:comp_slack_iopsop} and using the inequality~\eqref{eq:new-ineq-iopsop}, we can rewrite the above expression as:

\begin{align} \label{eq:reprove-comp-slack}
    \frac{w_i + \lambda_i}{\sum_{j = 1}^{m} u_{ij}x_{ij}} = \frac{1}{y_i} \leq \frac{p_j + \sum_{t \in T_i} r_{it} A_{tj}^{(i)}} {u_{ij}}, \quad \forall j
\end{align}
This is exactly the first order derivative condition of the \acrshort{bp-sop} as in Equation~\eqref{eq:LagDerSoc}. Multiplying Equation~\eqref{eq:reprove-comp-slack} by $x_{ij}$, we obtain by complimentary slackness the condition in equation~\eqref{eq:CompSlackSoc1}.

Finally, since we have that all the dual variables in \acrshort{iop} are scaled by a positive constant, as $y_i>0$ the corresponding signs of the dual variables in \acrshort{bp-sop} remain intact. Thus, we have shown that when $y_i>0$ and when we set $\lambda_i = \sum_{t \in T_i} r_{it} b_{it}$, $\forall i$ then the market equilibrium KKT conditions of the \acrshort{iop} are the same as that of \acrshort{bp-sop}.

($\Leftarrow$) To establish the converse of the theorem, we first note that the assumption that a market equilibrium exists implies that the capacity constraint of \acrshort{bp-sop} is met with equality, i.e., $\sum_{i} x_{ij} = \Bar{s}_j$ for all goods $j$.

The constraints of the social optimization problem also imply the \acrshort{iop} Constraints~\eqref{eq:con2} and~\eqref{eq:con3}. Next, we use the notation and equations in~\ref{Appendix4} and the corresponding KKT conditions of \acrshort{bp-sop} and \acrshort{iop}.
\newline
If we divide~\eqref{eq:LagDerSoc} and~\eqref{eq:CompSlackSoc1} by $\frac{w_i + \lambda_i}{\sum_{j = 1}^{m} u_{ij} x_{ij}}$, then these equations map respectively to the Lagrangian derivative equation~\eqref{eq:LagDer} and complimentary slackness equation~\eqref{eq:CompSlack1} of the individual optimization problem. We show this by denoting $\mathbf{p} \in \mathbb{R}^{m}$ as the dual variable for the capacity constraint, i.e., $p_j$ is the dual variable corresponding to the capacity constraint for good $j$, and then we have that~\eqref{eq:LagDerSoc} becomes:

\begin{align}\label{eq:LagDerSoc2}
    u_{ij} - \frac{\sum_{j = 1}^{m} u_{ij} x_{ij}}{w_i + \lambda_i} p_j - \sum_{t \in T_i} \frac{\sum_{j = 1}^{m} u_{ij} x_{ij}}{w_i + \lambda_i}r_{it} A_{tj}^{(i)} \leq 0 
\end{align}
The above equation is equivalent to~\eqref{eq:LagDer}, where $y_i = \frac{\sum_{j = 1}^{m} u_{ij} x_{ij}}{w_i + \lambda_i} \geq 0$, $\Tilde{r}_{it} = \frac{\sum_{j = 1}^{m} u_{ij} x_{ij}}{w_i + \lambda_i}r_{it} \geq 0$. The same analysis holds for the complementary slackness condition equivalence of \acrshort{bp-sop} and \acrshort{iop}. 

Now, all it remains for us to satisfy is the \acrshort{iop} budget constraint~\eqref{eq:con1}. To do this, we use the complementary slackness condition in equation~\eqref{eq:CompSlackSoc1} and sum over $j$. Then realizing that $r_{it} \sum_{j = 1}^{m} A_{tj}^{(i)} x_{ij} = r_{it} b_{it}$ by complimentary slackness we get:

\begin{align}\label{eq:CompSlackSoc2}
    w_i + \lambda_i - \sum_{j = 1}^{m} p_j x_{ij} - \sum_{t \in T_i} r_{it} b_{it} = 0 
\end{align}
Thus, if we set $\lambda_i = \sum_{t \in T_i} r_{it} b_{it}$ as in the statement of the theorem, we obtain that the budget constraint condition~\eqref{eq:con1} is satisfied with equality. This implies that we have a market clearing outcome, as the budgets are completely used and the goods are sold to capacity, while the KKT conditions of the two problems are equivalent at market clearing conditions. This establishes our claim.
\end{proof}

\end{document}